\documentclass[12pt]{article}
\usepackage[charter]{mathdesign}

\usepackage[T1]{fontenc}
\usepackage[latin9]{inputenc}
\usepackage{float}
\usepackage{booktabs}
\usepackage{mathrsfs}
\usepackage{amsmath}
\usepackage{amsthm}
\usepackage{graphicx}
\usepackage[letterpaper]{geometry}
\usepackage{setspace}
\usepackage[authoryear]{natbib}
\usepackage{microtype}
\usepackage[bookmarks=false,
 breaklinks=false,pdfborder={0 0 1},backref=section,colorlinks=false]
 {hyperref}

\usepackage{epigraph}
\makeatletter

\providecommand{\tabularnewline}{\\}

\theoremstyle{plain}
\newtheorem{assumption}{\protect\assumptionname}
\theoremstyle{plain}
\newtheorem{lem}{\protect\lemmaname}
\theoremstyle{plain}
\newtheorem{prop}{\protect\propositionname}
\theoremstyle{plain}
\newtheorem{cor}{\protect\corollaryname}
\theoremstyle{remark}
\newtheorem{rem}{\protect\remarkname}
\theoremstyle{plain}
\newtheorem{thm}{\protect\theoremname}

\usepackage{amsfonts}
\usepackage{eurosym}
\usepackage{ulem}
\usepackage{graphicx}
\usepackage{caption}
\usepackage{color}
\usepackage{sectsty}
\usepackage{comment}
\usepackage{footmisc}
\usepackage{caption}
\usepackage{pdflscape}
\usepackage{subfigure}
\usepackage{array}
\usepackage{upgreek}
\usepackage{bbm}

\usepackage{amsthm}
\usepackage{graphicx}
\usepackage{verbatim}
\usepackage{ulem}
\usepackage{textpos}
\usepackage{changepage}
\usepackage{url}

\newcommand{\Perp}{\perp \! \! \! \perp}

\def\bs{\boldsymbol}

\newcolumntype{L}[1]{>{\raggedright\let\newline\\arraybackslash\hspace{0pt}}m{#1}}
\newcolumntype{C}[1]{>{\centering\let\newline\\arraybackslash\hspace{0pt}}m{#1}}
\newcolumntype{R}[1]{>{\raggedleft\let\newline\\arraybackslash\hspace{0pt}}m{#1}}

\makeatother

\providecommand{\assumptionname}{Assumption}
\providecommand{\corollaryname}{Corollary}
\providecommand{\lemmaname}{Lemma}
\providecommand{\propositionname}{Proposition}
\providecommand{\remarkname}{Remark}
\providecommand{\theoremname}{Theorem}

\begin{document}
\title{{\Large Identification and Estimation of Production Function and Consumer
Demand Function under Monopolistic Competition from Revenue Data\thanks{{Acknowledgement: We are grateful for comments made at seminars
and conferences regarding an earlier version of this work. We thank
Zheng Han, Kohei Hashinokuchi, Yoko Sakamoto, and Makoto Tanaka for
excellent research assistance. Sugita acknowledges financial support
from JSPS KAKENHI (grant numbers: 17H00986, 19H01477, 19H00594, 20H01498,
22H00060). Hiroyuki Kasahara acknowledges financial support from the
SSHRC Insight Grant. }} }}
\author{Chun Pang Chow\thanks{Department of Economics, University of British Columbia, Canada. (Email:
alexccp@student.ubc.ca)}\and Hiroyuki Kasahara\thanks{Department of Economics, University of British Columbia, Canada. (Email:
hkasahar@mail.ubc.ca)} \and Yoichi Sugita\thanks{Faculty of Business and Commerce, Keio University, Japan. (E-mail:
ysugita@fbc.keio.ac.jp)}}
\date{\today }

\maketitle
 

\begin{abstract}
We establish nonparametric identification of production functions, total factor productivity (TFP), price markups, and firms' output prices and quantities, as well as consumer demand, using firm-level revenue data, without observing output quantity, in a monopolistically competitive environment with a fully nonparametric demand system. This result overturns the widely held view---formalized by  \citet*{bond2021some}---that output elasticities and markups are not nonparametrically identifiable from revenue data without quantity information. Under the additional restriction that demand satisfies the homothetic single-aggregator (HSA) structure of \citet{matsuyama2017beyond}, we further nonparametrically identify the representative consumer's utility function from firm-level revenue data. This new identification result enables counterfactual welfare analysis without parametric assumptions on preferences. We propose a semiparametric estimator that is feasible for standard firm-level datasets under a Cobb--Douglas production specification. Monte Carlo simulations show that the estimator performs well, while treating revenue as output induces substantial bias. Applying the estimator to Chilean manufacturing data, we reject the CES specification in favor of HSA, and find that market power reduces welfare by approximately 3\%--6\% of industry revenue in the three largest manufacturing industries in 1996.
\end{abstract}

\section{Introduction}

\epigraph{\itshape The derivatives of the function $\phi$ [with respect to inputs] will be random for \dots\ differences in the prices paid or received by various firms if we drop, as we realistically must, the assumption of perfect competition.}{--- Marschak and Andrews (1944, p.~145)}

The estimation of production functions and markups is a central tool
in empirical analyses of market outcomes, underpinning research on
firm-level productivity (\citealp{bartelsman2000understanding}; \citealp{syverson2011determines}),
aggregate productivity and misallocation (\citealp{olley1996dynamics};
\citealp{hsieh2009misallocation}), technological change (\citealp{van2003productivity};
\citealp{doraszelski2018measuring}), and the evolution of market
power (\citealp{hall1988relation}; \citealp{de2012markups}; \citealp{de2020rise}).\footnote{\citet{griliches_mairesse_1999}
and \citet*{ACKERBERG20074171} provide excellent surveys of production
function estimation methods.} A common assumption underlying many estimation methods is that firms\textquoteright{}
output quantities are observable. In practice, however, most firm-level
datasets contain only revenue, and researchers typically deflate revenue
using an industry-level price deflator.\footnote{A few studies employ
datasets with firm-level quantity information (e.g., \citealp*{foster2008reallocation};\citealp*{doraszelski2013, doraszelski2018measuring};\citealp*{de2016prices}; \citealp{lu2015trade}; \citealp{nishioka2019measuring}),
but such data are typically limited to specific countries, industries,
and time periods and remain inaccessible to most researchers.} 


Following \citet{ma44ecma}'s  pioneering critique,\footnote{Marschak and Andrews were acutely aware that economic data typically come in the form of monetary values (sales, value added) rather than physical counts (tons, bushels). In the introduction to \citet{ma44ecma}, they explicitly flag the danger of conflating these concepts. In Footnote~3, attached to the definition of net output ($x_0$), they write: ``[S]ome changes in notation will be needed later to distinguish between physical output; (gross) revenue; net revenue (`net output'): see \S\S4 and 9.''}  a large body of subsequent research has shown that replacing output quantity with revenue can severely bias estimates of production function parameters (e.g., 
\citealp{klette1996jae}; \citealp*{de2011product}), TFP (e.g., \citealp*{foster2008reallocation}; \citealp*{katayama2009firm}), and markups (\citealp*{bond2021some}). \citet*{bond2021some} formalize this difficulty, arguing that "in the usual setting in which the researcher observes only revenue, and does not have separate information on the price and quantity of output, the output elasticity for a flexible input is not identified non-parametrically from estimation of the revenue production function" (p.~2). Under a nonparametric demand function, the identification challenge has two dimensions: revenue becomes a nonparametric function of inputs, unobserved TFP, and an unobserved demand shock; and flexible inputs may be correlated with both TFP and the demand shock.\footnote{As noted by Klette and Griliches (1996), Marschak and Andrews (1944) were the first to recognize these two identification challenges and to criticize the practice of replacing output quantities with revenue in production function estimation.} Despite these concerns, researchers continue to use revenue in place of quantity due to the scarcity of firm-level quantity data.\footnote{Researchers also rely on revenue when products differ in quality, since physical output alone may not reflect true production, though such practices often lack theoretical foundations.}


In this paper, we establish that, under monopolistic competition,
nonparametric identification of the production function, total factor
productivity, consumer demand, and counterfactual welfare effects
is in fact possible using firm-level revenue data, without observing output quantity. We consider
the same demand-side setting as \citet{bond2021some}, in which an
individual firm faces a nonparametric demand function that depends
on its output, observable characteristics, and an unobserved demand
shock; for identification, we allow this shock to be transitory.\footnote{In
the appendix, we also consider identification under persistent demand
shocks (e.g., AR(1)) by utilizing lagged firm characteristics (such
as R\&D) as instruments or by explicitly modeling persistent unobserved
quality heterogeneity.}  Our identification proof is constructive: it yields closed-form mappings from observables to the objects of interest. The proof combines standard assumptions maintained in the proxy variable literature \citep[e.g.,][]{levinsohn2003estimating}, such as strict monotonicity of input demand, with the first-order condition approach of \citet*{doraszelski2013,doraszelski2018measuring} and \citet{gandhi2020identification}, while imposing additional nonparametric restrictions on firms' demand functions.


To address the two identification challenges described above, we develop
a three-step approach that combines the control function method of
\citet{olley1996dynamics}, \citet{levinsohn2003estimating}, and
\citet*{ackerberg2015identification} with the first-order condition
approach of \citet*{doraszelski2013, doraszelski2018measuring} and \citet*{gandhi2020identification} in a novel way. The
key insight is to treat the control function---the inverse of a material
demand function, which serves as a proxy for TFP---not merely as
an auxiliary device, but as an object of nonparametric identification
in its own right. In the first step, we identify the unobserved transitory
demand shock that nonlinearly affects revenue by combining the control
function as in \citet{ackerberg2015identification} with the instrumental
variable quantile regression of \citet{chernozhukov2005iv}. Intuitively,
this step separates demand-side variation from productivity variation
in revenue, exploiting the fact that lagged information shifts inputs and TFP
but not the transitory demand shock. In the second step, we identify
the control function for TFP by applying the nonparametric identification
of transformation models (e.g., \citealp{horowitz1996semiparametric})
examined by \citet*{ekeland2004identification} and \citet*{chiappori2015nonparametric}.
This step recovers TFP (up to normalization) from the dynamics of
inputs and the demand shock, without requiring output quantity data,
because the control function maps observable inputs and the identified
demand shock into unobserved productivity. In the third step, we identify
the production function, markups, and the demand function using the
first-order condition for materials and the control function identified
in the second step. 


Our method identifies several key objects from revenue data. In our
main setting, markups and output elasticities are identified up to
scale, while the output price, output quantity, TFP, gross production
function, and consumer demand function are identified up to scale
and location---normalizations that are standard in nonparametric settings. Proposition~\ref{prop:equivalence} in Appendix~\ref{app:equivalence_class} formally characterizes this equivalence class, showing that any existing identification result for output elasticities or markup levels from revenue data necessarily imposes location and/or scale normalizations, whether explicitly or implicitly.
Identification is cross-sectional, allowing these objects to vary
over time. With the additional assumption of \emph{local} constant
returns to scale, we identify the levels of markups and output elasticities,
and identify the output price, output quantity, TFP, production function,
and consumer demand function up to location.\footnote{\citet*{flynn2019measuring}
impose \emph{global} constant returns to scale to identify a production
function. In Subsection~\ref{subsec:LCRS}, we clarify the distinction
between \emph{local} and \emph{global} constant returns to scale.}

By assuming the homothetic single-aggregator (HSA) demand system of
\citet{matsuyama2017beyond}, we further identify the firms\textquoteright{}
demand system and the representative consumer\textquoteright s utility
function nonparametrically, without imposing parametric functional
form restrictions.\footnote{In the Appendix, we establish identification
of two additional families of homothetic demand systems proposed by
\citet{matsuyama2017beyond}, the homothetic demand systems with direct
implicit additivity (HDIA) and those with indirect implicit additivity
(HIIA), and discuss how to conduct counterfactual analyses using them.}
The HSA class---which nests CES while permitting variable markups,
incomplete pass-through, and non-monotonic relationships between firm
size and markups (see \citealp{matsuyama2023non}; \citealp{matsuyama2025homothetic}
for comprehensive reviews)---has become a leading framework for
monopolistic competition with heterogeneous firms. To our knowledge,
the nonparametric identification of the consumer\textquoteright s
utility function within the HSA class from firm-level data is a new
result; it enables counterfactual welfare analysis under flexible
demand without the specification errors that arise from imposing CES.

Our result bridges two literatures that have largely operated in
isolation: the production function literature, which recovers supply-side
objects from revenue data, and the demand estimation literature, which
recovers consumer-side objects from market-level data. In the existing
production function literature, empirical measures of productivity
or markups constructed from revenue data neither permit welfare evaluation nor counterfactual analysis,
because the underlying consumer demand system and utility function
remain unidentified. Our result enables counterfactual  welfare analysis even when
only revenue data are available. In particular, we show how to compute
a counterfactual marginal-cost-pricing equilibrium and compare it
with the observed monopolistic-competition equilibrium, thereby enabling
a structural evaluation of firms\textquoteright{} market power and
its welfare consequences through counterfactual changes in prices,
quantities, consumer utility, and firm profits.


 While the nonparametric identification results establish informational sufficiency in principle, practical estimation with moderate sample sizes requires additional structure. We develop a semiparametric estimator that assumes a Cobb-Douglas production function but leaves the demand system unrestricted. The estimation proceeds in three steps. First, we nonparametrically estimate the transitory demand shock using the smooth GMM IV quantile regression of \citet*{firpo2022gmm}, which ensures quantile monotonicity. Second, we estimate the control function via the profile likelihood estimator of \citet*{linton2008estimation}. Third, we recover the production function, markups, and TFP. This three-step procedure provides a standalone estimate of the production function without parametric demand assumptions. In a fourth step, for counterfactual welfare analysis and testing the CES restriction, we estimate the CoPaTh-HSA demand system of \citet{matsuyama2020constant}.

Simulation results show that our estimator performs well in recovering
structural parameters, markups, and TFP. Applying the estimator to
Chilean plant-level data from the three largest manufacturing industries
(SIC 31, 32, and 38), we find evidence of misspecification under the
CES demand system in favor of the HSA demand system. Our counterfactual
welfare analysis reveals that market power results in welfare losses
of approximately 3\%--6\% of industry revenue in the three largest
Chilean manufacturing industries in 1996. To put this in context,
these losses exceed standard Harberger-triangle calculations and are
broadly consistent with the welfare costs of markups estimated by
\citet{edmond2023costly} and the productivity losses from misallocation
analyzed by \citet{baqaee2020productivity}, though our estimates are
derived from a different structural framework.

Our analysis contributes to a growing literature that addresses the
revenue-versus-quantity problem in production function estimation.
\citet{klette1996jae}, \citet{de2011product}, and \citet{gandhi2020identification} impose a CES demand
structure under which log revenue is linear in log inputs and log
TFP; in contrast, we identify the demand system and the representative
consumer\textquoteright s utility function nonparametrically, without
relying on parametric assumptions. \citet*{flynn2019measuring} achieve
identification under global constant returns to scale; we require
only local constant returns to scale for certain normalizations. Recent
work by \citet{demirer2025factor} studies factor misallocation using
revenue data, \citet{doraszelski2021reexamining} reexamine the internal
consistency of the \citet{de2012markups} method under market power,
\citet{deridder2025hitchhiker} provides practical guidance
on production function estimation, and \citet{kirov2025measuring}
develops alternative approaches to measuring markups. Our framework
complements these contributions by providing nonparametric identification
of both the production and demand sides from revenue data. On the
demand side, our identification of the HSA utility function contributes
to the growing literature on non-CES demand systems under monopolistic
competition (\citealp{matsuyama2017beyond}; \citealp{matsuyama2023non};
\citealp{matsuyama2025homothetic}), providing a new empirical foundation
for counterfactual welfare analysis within this class.\footnote{One
frequently sees within the literature an assumption of market structure
for the identification of demand and supply side objects. For example,
\citet*{blp1995ecta} identify firm-level marginal costs by specifying
oligopolistic competition; meanwhile, \citet*{ekeland2004identification}
and \citet*{heckman2010nonparametric} identify various demand and
supply side objects of a hedonic model by exploiting the properties
of perfect competition.}

The remainder of this paper is organized as follows. Subsection \ref{subsec:Setting}
introduces the model and setting, while Subsection \ref{subsec:example}
illustrates our three-step identification strategy through a parametric
example. Subsection \ref{subsec:Nonparametric} establishes the main
nonparametric identification results, and Subsection \ref{sec:extensions}
discusses alternative settings. Subsection \ref{subsec:Normalization}
introduces additional assumptions to fix scale and location normalizations,
and Subsection \ref{subsec:demand} discusses the identification of
the demand system and counterfactual analysis. Section \ref{sec:estimator}
presents our semiparametric estimator. Section \ref{sec:simulation}
reports simulation results. Section \ref{sec:Chilean} provides an
empirical application. Section \ref{sec:Concluding-Remarks} concludes.
The Appendix contains proofs and simulation details, and presents
identification results for several extensions, including endogenous
labor input, endogenous or discrete firm characteristics, persistent
demand shocks, unobserved quality heterogeneity, and two additional
families of homothetic demand systems proposed by \citet{matsuyama2017beyond}.
\section{Identification}

\label{sec:Identification}


\subsection{Setting}

\label{subsec:Setting}

We denote the logarithm of physical output, material, capital, and
labor as $y_{it}$, $m_{it}$, $k_{it}$, and $l_{it}$, respectively,
with their respective supports denoted as $\mathcal{Y}$, $\mathcal{M}$,
$\mathcal{K}$, and $\mathcal{L}$. We collect the three inputs (material,
capital, and labor) into a vector as $x_{it}:=(m_{it},k_{it},l_{it})'\in\mathcal{X}:=\mathcal{M}\times\mathcal{K}\times\mathcal{L}$.

At time $t$, output $y_{it}$ is related to inputs $x_{it}=(m_{it},k_{it},l_{it})'$
through the production function 
\begin{equation}
y_{it}=f_{t}\!\left(x_{it},z_{it}^{s}\right)+\omega_{it},\label{prod}
\end{equation}
where $z_{it}^{s}$ is a vector of exogenous characteristics with
support $\mathcal{Z}_{s}$ that may affect either the functional form
of $f_{t}(\cdot)$ or the level of total factor productivity (TFP)
(e.g., ownership status).\footnote{Throughout the paper, the subscript $t$ on functions (e.g., $f_t$, $\psi_t$) indicates that the structural form is time-varying (common to all firms at time $t$), while the subscript $it$ denotes firm-specific realizations.} Firm-level productivity $\omega_{it}$ follows
a first-order Markov process given by 
\begin{align}
\omega_{it} & =h_{t}\!\left(\omega_{it-1},z_{it-1}^{h}\right)+\eta_{it},\qquad\eta_{it}\stackrel{iid}{\sim}G_{\eta_{t}},\label{omega}
\end{align}
where $\eta_{it}$ is an innovation to productivity that is serially
uncorrelated, and $z_{it-1}^{h}$ is a vector of lagged characteristics
with support $\mathcal{Z}_{h}$ that may affect the productivity process
(e.g., previous import status as in \citealp{kasahara2008}). 


The demand function for a firm's product is strictly decreasing in
its price, and its inverse demand function is given by 
\[
p_{it}=\tilde{\psi}_{t}(y_{it},z_{it}^{d},\epsilon_{it}),
\]
where $z_{it}^{d}$ is an observable firm characteristic with support
$\mathcal{Z}_{d}$ that affects firm's demand (e.g., firm's export
status in \citet{de2012markups}) while $\epsilon_{it}$ represents
an unobserved demand shock.

We assume the demand shock $\epsilon_{it}$ has limited persistence to facilitate identification via lagged instruments, and interpret $\epsilon_{it}$ as the demand fluctuation remaining after conditioning on observable characteristics.
Specifically, $\epsilon_{it}$ is generated by
\begin{equation}
\epsilon_{it}=\Upsilon_{t}(\zeta_{it},\zeta_{it-1},...,\zeta_{it-\upsilon}),\qquad\zeta_{it-s}\overset{iid}{\sim}F_{\zeta_{t}}\quad\text{for }s=0,1,..,\upsilon.\label{eq:epsilon}
\end{equation}
Therefore, conditional on $z_{it}^{d}$, the underlying innovation $\zeta_{it}$ has a transitory effect on the demand shock $\epsilon_{it}$. Consequently, $\epsilon_{it}$ is serially correlated over $\upsilon$ periods but its persistence is limited: $\epsilon_{it}$ is independent of $\epsilon_{i,t-s}$ for $s \geq \upsilon + 1$. In contrast, an innovation to productivity $\eta_{it}$ has a permanent effect on future productivity in (\ref{omega}). This difference between the demand and supply shock specifications in (\ref{omega}) and (\ref{eq:epsilon}) captures the idea that demand shocks are temporary while supply shocks are permanent \citep[e.g.,][]{nelson1982trends}. Appendix \ref{app:persistent_demand} provides an alternative identification approach and shows that identification remains possible under persistent demand shocks when a supply-side instrument is available. Appendix \ref{app:quality} discusses how productivity $\omega_{it}$ may capture persistent differences in product quality across firms, where the firm's output $y_{it}$ can be interpreted as a quality-adjusted measure of output quantity; from this perspective, heterogeneous demands attributable to quality differences are captured by $\omega_{it}$.

As shown in \citet{matzkin2003nonparametric}, the identification
of a non-additive unobservable $\epsilon_{it}$ has to be up to its
monotonic transformation. Let $F_{\epsilon_{t}}$ be the c.d.f. of
$\epsilon_{it}$. Without loss of generality, we transform $\epsilon_{it}$
to a uniform variable, using $u_{it}:=F_{\epsilon_{t}}(\epsilon_{it})$,
\begin{align}
p_{it} & =\tilde{\psi}_{t}(y_{it},z_{it}^{d},F_{\epsilon_{t}}^{-1}(u_{it}))=\psi_{t}(y_{it},z_{it}^{d},u_{it}),\,\,\,u_{it}\sim\text{Unif}(0,1).\label{eq:inverse_demand}
\end{align}
Given $t$, $u_{it}$ cross-sectionally follows an independent and
identical uniform distribution.

The inverse demand function (\ref{eq:inverse_demand}) is non-parametrically
specified and generalizes the constant-elasticity demand function
examined by \citet{ma44ecma}, \citet{klette1996jae}, and \citet{de2011product}.
Equation (\ref{eq:inverse_demand}) implicitly imposes two key assumptions.
First, $\psi_{t}(\cdot,z_{it}^{d},u_{it})$ is common across firms
once we control for observed demand characteristics $z_{it}^{d}$
and a transitory scalar unobserved demand shock $u_{it}$. Second,
$\psi_{t}(\cdot,z_{it}^{d},u_{it})$ represents the demand curve that
each individual firm takes as given. This assumption is satisfied
under monopolistic competition (without free entry), where $\psi_{t}$
can be expressed as $\psi_{t}(y_{it},z_{it}^{d},u_{it},a_{t})$, with
$a_{t}$ denoting a vector of aggregate price and quantity indices
which each firm treats as exogenous.

Let $r_{it}$ and $\mathcal{R}$ be the logarithm of revenue and its
support, respectively. Then, from (\ref{prod}), the observed revenue
relates to output and input as follows: 
\begin{align}
r_{it} & =\varphi_{t}(y_{it},z_{it}^{d},u_{it})\label{rev-0}\\
 & =\varphi_{t}\left(f_{t}(m_{it},k_{it},l_{it},z_{it}^{s})+\omega_{it},z_{it}^{d},u_{it}\right)\label{rev}
\end{align}
where $\varphi_{t}(y_{it},z_{it}^{d},u_{it}):=\psi_{t}(y_{it},z_{it}^{d},u_{it})+y_{it}.$

We make the following timing assumption. \begin{assumption}\label{A-0}
(a) $(l_{it},k_{it})$ is determined at the end of period $t-1$ and
is independent of $\eta_{is}$ and $\zeta_{is}$ for $s\geq t$. (b)
$m_{it}$ is determined after firm's observing $(\omega_{it},u_{it},z_{it}^{s},z_{it}^{d})$
but is independent of $\eta_{is}$ and $\zeta_{is}$ for $s\geq t+1$.
(c) $\left(z_{it}^{s},z_{it}^{d},z_{it-1}^{h}\right)$ is continuous
and independent of $u_{is}$ and $\eta_{is}$ for $s\geq t$. (d)
each firm is a price-taker for material input. \end{assumption}

Assumptions \ref{A-0}(a)(b) specify the timing structure, which is
similar to that in \citet{gandhi2020identification}.\footnote{Stochastic independence between $(l_{it},k_{it})$ and $\eta_{it}$ is stronger than the standard mean independence assumption $E[\eta_{it}\mid\mathcal{I}_{it}]=0$, where $\mathcal{I}_{it}$ is the firm's information set at the beginning of period~$t$. Most identification results below require only mean independence; full stochastic independence is used in the IVQR step (Proposition~\ref{lem:step2}) and the control function argument (Proposition~\ref{P-step1}).}
In Appendix
\ref{subsec:Endogenous-Labor-Input}, we present identification results
when $l_{it}$ is also endogenous. The continuity requirement in Assumption~\ref{A-0}(c)
can be relaxed, but the exogeneity of $\left(z_{it}^{s},z_{it}^{d},z_{it-1}^{h}\right)$
remains an important---albeit potentially strong---assumption, though
it is commonly maintained in the empirical literature.\footnote{In Appendix~\ref{app:Alternative-Settings}, we further discuss identification
when these variables are discrete and endogenous, under the availability
of suitable instruments.} Assumption~\ref{A-0}(d) is also standard in most empirical applications.\footnote{ 
We treat deflated expenditures on materials as measures of inputs.
This abstracts from unobserved heterogeneity in material prices.
For instance, geographically segmented input markets may induce
systematic price differences across regions.} Appendix \ref{subsec:Identification-Expenditure} shows that the
identification strategy extends to such environments when material prices
vary across firms as a function of observed characteristics.

Under Assumption \ref{A-0}, the firm chooses $m_{it}=\mathbb{M}_{t}\left(\omega_{it},k_{it},l_{it},z_{it}^{s},z_{it}^{d},u_{it}\right)$
at time $t$ to maximize the profit: 
\begin{equation}
\mathbb{M}_{t}\left(\omega_{it},k_{it},l_{it},z_{it}^{s},z_{it}^{d},u_{it}\right)\in\arg\max_{m}\exp(\varphi_{t}(f_{t}(m,k_{it},l_{it},z_{it}^{s})+\omega_{it},z_{it}^{d},u_{it}))-\exp(p_{t}^{m}+m),\label{eq:profit_maximization}
\end{equation}
where $p_{t}^{m}$ denotes the logarithm of the material price
at time $t$.

Equation (\ref{rev}) highlights two identification issues, originally
raised by \citet{ma44ecma}. First, $m_{it}$ correlates with two
unobservables $\omega_{it}$ and $u_{it}$. Second, $r_{it}$ relates
to $x_{it}=(m_{it},k_{it},l_{it})$ via two unknown nonlinear functions
$\varphi_{t}(\cdot,z_{it}^{d},u_{it})$ and $f_{t}(\cdot)$. From
the first-order condition for profit maximization, $P_{it}\left(1+\partial\psi_{t}(y_{it},z_{it}^{d},u_{it})/\partial y_{it}\right)=MC_{it}$,
where $P_{it}$ and $MC_{it}$ denote the price and marginal cost
of output, respectively, the elasticity of revenue with respect to
output equals the inverse of the markup: 
\begin{equation}
\frac{\partial\varphi_{t}(y_{it},z_{it}^{d},u_{it})}{\partial y_{it}}=\frac{MC_{it}}{P_{it}}.\label{eq:markup_foc}
\end{equation}
Thus, the revenue elasticity relates to the output elasticity via
markup: 
\[
\frac{\partial\varphi_{t}(f_{t}(x_{it})+\omega_{it},z_{it}^{d},u_{it})}{\partial v_{it}}=\frac{MC_{it}}{P_{it}}\frac{\partial f_{t}(x_{it})}{\partial v_{it}}\text{ for }v_{it}\in\{m_{it},k_{it},l_{it}\}.
\]

For identification, we make the following assumptions. \begin{assumption}
\label{A-1} (a) $f_{t}(\cdot)$ is continuously differentiable with
respect to $(m,k,l,z^{s})$ on $\mathcal{M}\times\mathcal{K}\times\mathcal{L}\times\mathcal{Z}_{s}$
and strictly increasing in $m$. (b) For every $\left(z^{d},u\right)\in\mathcal{Z}_{d}\times[0,1]$,
$\varphi_{t}(\cdot,z^{d},u)$ is strictly increasing and invertible
with its inverse $\varphi_{t}^{-1}(r,z^{d},u)$, which is continuously
differentiable with respect to $(r,z^{d},u)$ on $\mathcal{R}\times\mathcal{Z}_{d}\times[0,1]$.
(c) For every $(k,l,z^{s},z^{d},u)\in\mathcal{K}\times\mathcal{L}\times\mathcal{Z}_{s}\times\mathcal{Z}_{d}\times[0,1]$,
$\mathbb{M}_{t}(\cdot,k,l,z^{s},z^{d},u)$ is strictly increasing
and invertible with its inverse $\mathbb{M}_{t}^{-1}(m,k,l,z^{s},z^{d},u)$,
which is continuously differentiable with respect to $(m,k,l,z^{s},z^{d},u)$
on $\mathcal{M}\times\mathcal{K}\times\mathcal{L}\times\mathcal{Z}_{s}\times\mathcal{Z}_{d}\times[0,1]$.
(d) $(\zeta_{it},...,\zeta_{it-\upsilon})$ are independent from $\eta_{it}$.
\end{assumption} Assumptions \ref{A-1}(a)(b) are standard assumptions
about smooth production and demand functions. In Assumption \ref{A-1}(b),
the condition $\partial\varphi_{t}(y_{it},z_{it}^{d},u_{it})/\partial y_{it}>0$
is equivalent to that the elasticity of demand with respect to price,
$-\left(\partial\psi_{t}(y_{it},z_{it}^{d},u_{it})/\partial y_{it}\right)^{-1}$,
being greater than 1; this necessarily holds under profit maximization.
Assumption \ref{A-1}(c) is a standard assumption in the control function
approach that uses material as a control function for TFP \citep{levinsohn2003estimating,ackerberg2015identification}.
Assumption \ref{A-1}(d) requires the demand shock and the productivity
shock are independent.

Let $w_{it}:=(k_{it},l_{it},z_{it}^{s},z_{it}^{d})$ be observable
exogenous variables at $t$. The inverse function of the material
demand function with respect to TFP 
\[
\omega_{it}=\mathbb{M}_{t}^{-1}(m_{it},w_{it},u_{it})
\]
is used as a control function for $\omega_{it}$. Since $\partial\varphi_{t}(y_{it},z_{it}^{d},u_{it})/\partial y_{it}>0$,
there exists the inverse function $\varphi_{t}^{-1}(\cdot,z_{it}^{d},u_{it})$
so that the revenue function $r_{it}=\varphi_{t}(f_{t}(x_{it},z_{it}^{s})+\omega_{it},z_{it}^{d},u_{it})$
can be written as: 
\begin{align}
\varphi_{t}^{-1}\left(r_{it},z_{it}^{d},u_{it}\right)=f_{t}(x_{it},z_{it}^{s})+\mathbb{M}_{t}^{-1}(m_{it},w_{it},u_{it}).\label{eq:model_step2}
\end{align}

Let 
$v_{it}:=(w_{it},u_{it},m_{it-1},w_{it-1},u_{it-1},z_{it-1}^{h})'\in\mathcal{V}:=\mathcal{W}\times[0,1]\times\mathcal{M}\times\mathcal{W}\times[0,1]\times\mathcal{Z}_{h}$,
where $\mathcal{W}:=\mathcal{K}\times\mathcal{L}\times\mathcal{Z}_{s}\times\mathcal{Z}_{d}$.
We assume that the data constitute a random sample of $N$ firms observed
over multiple periods, 
$\{\{r_{is},m_{is},v_{is}\}_{s=t-\upsilon-2}^{t}\}_{i=1}^{N}$, 
drawn from the population. Given a sufficiently large $N$, the econometrician
can consistently recover the corresponding population joint distributions.
\begin{assumption} \label{A-data} An econometrician is assumed to
know the following objects: (a) the population joint distribution
of $\{r_{is},m_{is},v_{is}\}_{s=t-\upsilon-2}^{t}$; and (b) the material
input cost for each firm, $\exp(p_{t}^{m}+m_{it})$. \end{assumption}
Our objective is to identify $\{\varphi_{t}^{-1}(\cdot),f_{t}(\cdot),\mathbb{M}_{t}^{-1}(\cdot)\}$
from the population joint distribution of $\{r_{is},m_{is},v_{is}\}_{s=t-\upsilon-2}^{t}$.
Let $\{\varphi_{t}^{*-1}(\cdot),f_{t}^{*}(\cdot),\mathbb{M}_{t}^{*-1}(\cdot)\}$
be the true model structure that satisfies (\ref{eq:model_step2}).
Then, for any $(a_{1t},a_{2t},b_{t})\in\mathbb{R}^{2}\times\mathbb{R}_{++}$,
\begin{align}
\varphi_{t}^{-1}\left(r_{it},z_{it}^{d},u_{it}\right) & =(a_{1t}+a_{2t})+b_{t}\varphi_{t}^{*-1}\left(r_{it},z_{it}^{d},u_{it}\right),\ f_{t}(x_{it},z_{it}^{s})=a_{1t}+b_{t}f_{t}^{*}(x_{it},z_{it}^{s}),\nonumber \\
 & \text{{and}\ }\mathbb{M}_{t}^{-1}(m_{it},w_{it},u_{it})=a_{2t}+b_{t}\mathbb{M}_{t}^{*-1}(m_{it},w_{it},u_{it})\label{eq:equivalence}
\end{align}
also satisfy (\ref{eq:model_step2}). 
Hence, the true structure $\{\varphi_{t}^{*-1}(\cdot),f_{t}^{*}(\cdot),\mathbb{M}_{t}^{*-1}(\cdot)\}$
is observationally equivalent to the structure (\ref{eq:equivalence}) and is therefore identified only up to location and scale normalization $(a_{1t},a_{2t},b_{t})$
from restriction (\ref{eq:model_step2}).

Normalization is unavoidable absent further assumptions. Because the unobserved output level $y_{it}$ enters the revenue function through the unknown nonlinear function $\varphi_t$ in \eqref{rev-0}, it has no natural scale or location, so the structural functions are identified only up to a class of transformations. Proposition~\ref{prop:equivalence} in Appendix~\ref{app:equivalence_class} formally characterizes this equivalence class for $(\varphi_t^{-1}, f_t, \mathbb{M}_t^{-1})$; accordingly, any existing identification result for output elasticities or markup level from revenue data either explicitly or implicitly imposes location and/or scale normalizations.

We may fix $(a_{1t},a_{2t},b_{t})$ in (\ref{eq:equivalence}) by fixing
the values of $\{\varphi_{t}^{-1}(\cdot),f_{t}(\cdot),\mathbb{M}_{t}^{-1}(\cdot)\}$
at some points. Specifically, choosing two points $(m_{t0}^{*},w_{t}^{*},u_{t}^{*})$
and $(m_{t1}^{*},w_{t}^{*},u_{t}^{*})$ on the support $\mathcal{M}\times\mathcal{W}\times[0,1]$
where $m_{t0}^{*}<m_{t1}^{*}$, we denote 
\begin{equation}
c_{1t}:=f_{t}(m_{t0}^{*},k_{t}^{*},l_{t}^{*},z_{t}^{s*}),\ c_{2t}=\mathbb{M}_{t}^{-1}(m_{t0}^{*},w_{t}^{*},u_{t}^{*}),\ \text{{and}\ }c_{3t}:=\mathbb{M}_{t}^{-1}(m_{t1}^{*},w_{t}^{*},u_{t}^{*}).\label{eq:norm}
\end{equation}
Note that $\partial\mathbb{M}_{t}^{-1}/\partial m_{it}>0$ implies
that $c_{2t}<c_{3t}$. Then, there exists a unique one-to-one mapping
between $(c_{1t},c_{2t},c_{3t})$ in (\ref{eq:norm}) and $(a_{1t},a_{2t},b_{t})$
in (\ref{eq:equivalence}) such that $b_{t}=\left(c_{3t}-c_{2t}\right)/\left(\mathbb{M}_{t}^{*-1}(m_{t1}^{*},w_{t}^{*},u_{t}^{*})-\mathbb{M}_{t}^{*-1}(m_{t0}^{*},w_{t}^{*},u_{t}^{*})\right)$,
$a_{1t}=c_{1t}-b_{t}f_{t}^{*}(m_{t0}^{*},k_{t}^{*},l_{t}^{*},z_{t}^{s*})$
and $a_{2t}=c_{2t}-b_{t}\mathbb{M}_{t}^{*-1}(m_{t0}^{*},w_{t}^{*},u_{t}^{*})$.
Thus, we can fix the value of $(a_{1t},a_{2t},b_{t})$ by choosing
arbitrary values $(c_{1t},c_{2t},c_{3t})\in\mathbb{R}^{3}$ that satisfies
$c_{2t}<c_{3t}$. In particular, we impose the following normalization
that corresponds to (N2) in \citet{chiappori2015nonparametric}. \begin{assumption}
(Normalization) \label{A-2} The support $\mathcal{M}\times\mathcal{W}\times[0,1]$
includes two points $(m_{t0}^{*},w_{t}^{*},u_{t}^{*})$ and $(m_{t1}^{*},w_{t}^{*},u_{t}^{*})$
such that $c_{1t}=c_{2t}=0$ and $c_{3t}=1$ in (\ref{eq:norm}).
\end{assumption} As \citet{chiappori2015nonparametric} demonstrates,
this choice of normalization makes the identification proofs transparent.
In Section \ref{subsec:Normalization}, we discuss how we can use
additional restrictions and data to identify the normalization parameters
$(a_{1t},a_{2t},b_{t})$.

%

\subsection{Identification in a Parametric Example: Generalized CES Demand with
Heterogeneity}

\label{subsec:example}

Before presenting the nonparametric identification results, we demonstrate
our identification approach by applying it to a simple parametric
example without exogenous covariates, i.e., where $(z_{it}^{d},z_{it}^{s},z_{it}^{h})$
is empty. Consider a monopolistically competitive market where each
firm $i$ faces the following constant elastic inverse demand function
with heterogeneity: 
\begin{equation}
p_{it}=\alpha_{t}(u_{it})+(\rho(u_{it})-1)y_{it},\label{eq:demand}
\end{equation}
where $\alpha_{t}(\cdot)$ and $\rho(\cdot)$ are unknown functions,
where $0<\rho(\cdot)\le1$.\footnote{The demand function \eqref{eq:demand} can be derived from a constant
elasticity of substitution (CES) utility function, where the elasticity
of substitution parameter is heterogenous across firms, depending
on $u$. The term $\alpha_{t}$ implicitly captures aggregate expenditure
and an aggregate price index.} We assume that $\rho'(u)<0$, which implies that the markup $1/\rho(u)$
is increasing in $u$.

Firm $i$ has a Cobb--Douglas production function with the TFP $\omega_{it}$
that follows a first-order autoregressive (AR(1)) process: 
\begin{align}
y_{it} & =\theta_{0}+\theta_{m}m_{it}+\theta_{k}k_{it}+\theta_{l}l_{it}+\omega_{it},\quad\omega_{it}=h_{1}\omega_{it-1}+\eta_{it},\label{eq:AR1}
\end{align}
where $\{\theta_{0},\theta_{m},\theta_{k},\theta_{l},h_{1}\}$ are
unknown parameters. The firm's revenue function is expressed as: 
\begin{equation}
r_{it}=\alpha_{t}(u_{it})+\rho(u_{it})\theta_{0}+\rho(u_{it})\theta_{m}m_{it}+\rho(u_{it})\theta_{k}k_{it}+\rho(u_{it})\theta_{l}l_{it}+\rho(u_{it})\omega_{it}.\label{eq:r_t}
\end{equation}
Denote the ratio of material cost to revenue as 
$s_{it}^{m}:=\frac{\exp(p_{t}^{m}+m_{it})}{\exp(r_{it})}$. 
Then, the first-order condition for (\ref{eq:profit_maximization})
can be written as 
\begin{equation}
\rho(u_{it})\theta_{m}=s_{it}^{m},\label{eq:foc}
\end{equation}
which, in turn, determines the control function for $\omega_{it}$
as 
\begin{align}
\omega_{it} & =\mathbb{M}_{t}^{-1}(m_{it},k_{it},l_{it},u_{it})=\beta_{t}(u_{it})+\beta_{m}(u_{it})m_{it}+\beta_{k}k_{it}+\beta_{l}l_{it}\label{eq:omega}
\end{align}
where $\beta_{t}(u_{it})=\left(p_{t}^{m}-\alpha_{t}(u_{it})-\rho(u_{it})\theta_{0}-\ln\rho(u_{it})\theta_{m}\right)/\rho(u_{it})$,
$\beta_{m}(u_{it})=\left(1-\rho(u_{it})\theta_{m}\right)/\rho(u_{it})>0$,
$\beta_{k}=-\theta_{k}$ and $\beta_{l}=-\theta_{l}$.

For notational brevity, assume that the support $\mathcal{X}$ includes
two points $(m_{t0}^{*},k_{t}^{*},l_{t}^{*})=(0,0,0)$ and $(m_{t1}^{*},k_{t}^{*},l_{t}^{*})=(1,0,0)$.
Following Assumption \ref{A-2}, we fix the location and scale of
$f_{t}(\cdot)$ and $\mathbb{M}_{t}^{-1}(\cdot)$ by imposing the
following normalization: 
\begin{align}
0 & =f_{t}(0,0,0)=\theta_{0},\,0=\mathbb{M}_{t}^{-1}(0,0,0,0.5)=\beta_{t}(0.5),\nonumber \\
1 & =\mathbb{M}_{t}^{-1}(1,0,0,0.5)=\beta_{t}(0.5)+\beta_{m}(0.5)\label{eq:normalization_example}
\end{align}
which implies $\theta_{0}=0$, $\beta_{t}(0.5)=0$, and $\beta_{m}(0.5)=1$.

Our identification approach follows three steps.

\subsubsection{Step 1: Identification of the Demand Shocks}

The first step identifies the demand shock $u_{it}$. Substituting
$\omega_{it}=\mathbb{M}_{t}^{-1}(m_{it},k_{it},l_{it},u_{it})$ and
using $\theta_{0}=0$, we obtain 
\begin{align}
r_{it}= & \left(\alpha_{t}(u_{it})+\rho(u_{it})\beta_{t}(u_{it})\right)+\rho(u_{it})\left(\theta_{m}+\beta_{m}(u_{it})\right)m_{it}\nonumber \\
 & +\rho(u_{it})\left(\theta_{k}+\beta_{k}\right)k_{it}+\rho(u_{it})\left(\theta_{l}+\beta_{l}\right)l_{it}\label{eq:expr1}\\
= & \tilde{\phi}_{t}(u_{it})+m_{it},\label{eq:expr2}
\end{align}
where the second equality uses $\beta_{k}=-\theta_{k}$ and $\beta_{l}=-\theta_{l}$
from (\ref{eq:omega}), so that $\rho(u_{it})(\theta_{k}+\beta_{k})=\rho(u_{it})(\theta_{l}+\beta_{l})=0$,
and $\tilde{\phi}_{t}(u_{it}):=\alpha_{t}(u_{it})+\rho(u_{it})\beta_{t}(u_{it})$
with $\tilde{\phi}'_{t}(u)=-\theta_{m}\rho'(u)/\rho(u)>0$ for all $u$.


From (\ref{eq:expr2}), we have 
$\Pr[r_{it}-m_{it}\leq\tilde{\phi}_{t}\left(u\right)]=u\quad\,\text{for all }u\in[0,1]$
because $\Pr[r_{it}-m_{it}\leq\tilde{\phi}_{t}\left(u\right)]=\Pr[\tilde{\phi}_{t}\left(u_{it}\right)\leq\tilde{\phi}_{t}\left(u\right)]=u$
by the monotonicity of $\tilde{\phi}_{t}(\cdot)$. Therefore, the
quantile of $r_{it}-m_{it}$ identifies $u_{it}$ while the moment
condition $E\left[1\left\{ r_{it}-m_{it}\le\tilde{\phi}_{t}\left(u\right)\right\} -u\right]=0$
for $u\in[0,1]$ identifies $\tilde{\phi}_{t}(\cdot)$.

Alternatively, from the first-order condition (\ref{eq:foc}) and
the monotonicity of $\rho(\cdot)$ with $\rho'(\cdot)<0$, the demand
shock $u_{it}$ is identified as the quantile of $1/s_{it}^{m}$.
This equivalence arises because the quantile of $r_{it}-m_{it}$ coincides
with that of $1/s_{it}^{m}=\exp(r_{it}-m_{it}-p_{t}^{m})$.

\paragraph{Step 2: Identification of Control Function and TFP}

The second step identifies the control function $\mathbb{M}_{t}^{-1}(\cdot)$.
Substituting (\ref{eq:omega}) into the AR(1) process (\ref{eq:AR1})
leads to 
\begin{equation}
\mathbb{M}_{t}^{-1}(m_{it},k_{it},l_{it},u_{it})=h_{1}\mathbb{M}_{t-1}^{-1}(m_{it-1},k_{it-1},l_{it-1},u_{it-1})+\eta_{it}.\label{eq:model_ex1}
\end{equation}
Since $\mathbb{M}_{t}^{-1}(m_{it},k_{it},l_{it},u_{it})$ is linear
in $m_{it}$ from (\ref{eq:omega}), we can rearrange (\ref{eq:model_ex1})
as: 
\begin{align}
m_{it} & =\gamma(u_{it},u_{t-1})+\gamma_{k}(u_{it})k_{it}+\gamma_{l}(u_{it})l_{it}+\delta_{m}(u_{it},u_{it-1})m_{it-1}\nonumber \\
 & +\delta_{k}(u_{it})k_{it-1}+\delta_{l}(u_{it})l_{it-1}+\tilde{\eta}_{it},\label{eq:m_model}
\end{align}
where 
\begin{align}
 & \gamma_{k}(u_{it})=-\frac{\beta_{k}}{\beta_{m}(u_{it})},\ \gamma_{l}(u_{it})=-\frac{\beta_{l}}{\beta_{m}(u_{it})},\ \delta_{k}(u_{it})=\frac{h_{1}\beta_{k}}{\beta_{m}(u_{it})},\ \delta_{l}(u_{it})=\frac{h_{1}\beta_{l}}{\beta_{m}(u_{it})},\label{eq:parameters1}\\
 & \gamma(u_{it},u_{it-1})=\frac{-\beta_{t}(u_{it})+h_{1}\beta_{t-1}(u_{it-1})}{\beta_{m}(u_{it})},\label{eq:parameters2}
\end{align}
$\tilde{\eta}_{it}={\eta_{it}}/{\beta_{m}(u_{it})},$ and $\delta_{m}(u_{it},u_{it-1})={h_{1}\beta_{m}(u_{it-1})}/{\beta_{m}(u_{it})}$.
For a given $(u_{it},u_{it-1})$, (\ref{eq:m_model}) is a linear
model. Since $E\left[\left.\tilde{\eta}_{it}\right|v_{it}\right]=E\left[\left.\eta_{it}\right|v_{it}\right]/\beta_{m}(u_{it})=0$,
where $v_{it}:=(k_{it},l_{it},x_{it-1},u_{it},u_{it-1})$, we can
identify $\{\gamma(\cdot,\cdot)$, $\gamma_{k}(\cdot)$, $\gamma_{l}(\cdot)$,
$\delta_{m}(\cdot,\cdot)$, $\delta_{k}(\cdot)$, $\delta_{l}(\cdot)\}$
in (\ref{eq:m_model}) from the conditional moment restriction $E\left[\left.\tilde{\eta}_{it}\right|v_{it}\right]=0$.

Then, because $\beta_{m}(0.5)=1$, we can recover $(\theta_{k},\theta_{l},h_{1})$
from (\ref{eq:parameters1}) as 
\[
\theta_{k}=-\beta_{k}=\gamma_{k}(0.5),\ \theta_{l}=-\beta_{l}=\gamma_{l}(0.5),\text{ and }h_{1}=-\frac{\delta_{k}(0.5)}{\gamma_{k}(0.5)}=-\frac{\delta_{l}(0.5)}{\gamma_{l}(0.5)}.
\]
Also, applying the normalization (\ref{eq:normalization_example})
to (\ref{eq:parameters2}), we have $\gamma(0.5,u_{t-1})=h_{1}\beta_{t-1}(u_{t-1})$.
Then, $\beta_{m}(u)$ and $\beta_{t}(u)$ are identified from (\ref{eq:parameters1})-(\ref{eq:parameters2})
as 
\[
\beta_{m}(u)=\frac{\gamma_{k}(0.5)}{\gamma_{k}(u)}=\frac{\gamma_{l}(0.5)}{\gamma_{l}(u)}\text{ and }\beta_{t}(u)=\gamma(0.5,u_{t-1})-\frac{\gamma(u,u_{t-1})\gamma_{k}(0.5)}{\gamma_{k}(u)}.
\]
Given the identification of $(\beta_{k},\beta_{l},\beta_{m}(\cdot),\beta_{t}(\cdot))$,
we can identify $\omega_{it}$ from (\ref{eq:omega}).

\paragraph{Step 3: Identification of Production Function, Markup, and Demand
Function}

The identification of $\rho(u)$ follows from substituting (\ref{eq:foc})
into $\beta_{m}(u)=\left(1-\rho(u)\theta_{m}\right)/\rho(u)$, and
rearranging the terms, which yields 
\[
\rho(u_{it})=\frac{1-s_{it}^{m}}{\beta_{m}(u_{it})}=\frac{1-s_{it}^{m}}{\gamma_{k}(0.5)/\gamma_{k}(u_{it})}.
\]
Therefore, the markup $1/\rho(u_{it})$ is identified.

The first order condition (\ref{eq:foc}) implies that the revenue
share of material expenditure is a function of $u_{it}$, which we
denote by $s(u)$, such that $s_{it}^{m}=s(u_{it})$. In particular,
$s(0.5)$ represents the median revenue share of material expenditure.
Then, the identification of $\theta_{m}$ follows from the identification
of $\rho(u)$ and the first order condition (\ref{eq:foc}) as 
\begin{equation}
\theta_{m}=\frac{s(0.5)}{\rho(0.5)}=\frac{s(0.5)}{1-s(0.5)}.\label{eq:theta_m}
\end{equation}

Appendix~\ref{app:ces} further discusses counterfactual analysis
under the CES specification. 

\paragraph{The Identification under Normalization}

In view of the first-order condition $\rho(u_{it})\theta_{m}=s_{it}^{m}$,
it is clear from the argument above that the markup level cannot be
separately identified from the material input coefficient $\theta_{m}$
without imposing the normalization restriction $\beta_{m}(0.5)=1$.

More generally, the parameters are identified under the scale and
location normalization of $f_{t}(\cdot)$ and $\mathbb{M}_{t}^{-1}(\cdot)$
in (\ref{eq:normalization_example}). Let $\theta_{i}$ ($i=0,m,k,l$)
and $\beta_{j}(u_{t})$ ($j=t,m,k,l$) be those parameters identified
above and let $\theta_{j}^{*}$ and $\beta_{i}^{*}(u_{t})$ be the
true parameters. Then, there exist unknown normalization parameters
$(a,b)\in\mathbb{R}\times\mathbb{R}_{+}$ such that 
\[
\theta_{0}=a+b\theta_{0}^{*},\ \beta_{t}=a+b\beta_{t}^{*},\ \theta_{i}=b\theta_{i}^{*},\ \beta_{j}(u_{t})=b\beta_{j}^{*}(u_{t}).
\]
We can fix the normalization by imposing further restrictions. For
instance, if constant returns to scale $\theta_{m}^{*}+\theta_{k}^{*}+\theta_{l}^{*}=1$
holds, then the scale parameter $b$ can be identified as 
\[
b=b\left(\theta_{m}^{*}+\theta_{k}^{*}+\theta_{l}^{*}\right)=\theta_{m}+\theta_{k}+\theta_{l}=\frac{s(0.5)}{1-s(0.5)}-\beta_{k}-\beta_{l}.
\]
We discuss in subsection \ref{subsec:Normalization} additional assumptions
for fixing normalization.

The above identification argument is illustrative but relies on the
linearity of $\mathbb{M}_{t}^{-1}(m_{it},k_{it},l_{it},u_{it})$ under
restrictive parametric assumptions. The next subsection establishes
nonparametric identification in a more general framework presented
in Section~\ref{subsec:Setting}.


\subsection{Nonparametric Identification}

\label{subsec:Nonparametric}

\subsubsection{Step 1: Identification of the Demand Shocks}


Substituting $\omega_{it}=\mathbb{M}_{t}^{-1}(m_{it},w_{it},u_{it})$
into the revenue function (\ref{rev}), we can rewrite it as 
\begin{align}
r_{it} & =\varphi_{t}\left(f_t(m_{it},k_{it},l_{it},z_{it}^{s})+\mathbb{M}_{t}^{-1}(m_{it},w_{it},u_{it}),z_{it}^{d},u_{it}\right)\nonumber \\
 & =:\phi_{t}\left(m_{it},w_{it},u_{it}\right),\,\,\,u_{it}\sim\text{Unif}(0,1).\label{eq:Phi}
\end{align}
We impose the following assumptions. \begin{assumption} \label{assu: revenue monotonicity}
(a) (\textit{Monotonicity}) ${\displaystyle {\partial\phi_{t}(m,w,u)}/{\partial u}>0}$
for all $(m,w,u)\in\mathcal{M}\times\mathcal{W}\times[0,1]$. \\
 (b) (\textit{Completeness}) The conditional distribution of $(m_{it},w_{it})$
given $(m_{it-\upsilon-1},w_{it-\upsilon})$ is complete in the sense
of \citet{chernozhukov2005iv}; that is, for any measurable function
$g(m,w)$, 
\[
E[g(m_{it},w_{it})\mid m_{it-\upsilon-1},w_{it-\upsilon}]=0\text{ a.s. }\Rightarrow g(m_{it},w_{it})=0\text{ a.s.}
\]
\end{assumption}

Assumption \ref{assu: revenue monotonicity}(a) implicitly imposes
restrictions on the shape of the demand function. The Appendix  \ref{app:proof-5} shows
that Assumption \ref{assu: revenue monotonicity}(a) holds if and
only if $\frac{\partial\varphi_{t}}{\partial u}\frac{\partial\sigma_{t}}{\partial y}>\frac{\partial\varphi_{t}}{\partial y}\frac{\partial\sigma_{t}}{\partial u}$,
where $\sigma_{t}(y,z^{d},u):=-1/\left(\frac{\partial\psi_{t}(y,z^{d},u)}{\partial y}\right)>0$
denotes the demand elasticity. Since $\frac{\partial\varphi_{t}}{\partial u}>0$
and $\frac{\partial\varphi_{t}}{\partial y}>0$, a sufficient condition
for Assumption \ref{assu: revenue monotonicity}(a) is that an increase
in the demand shock $\epsilon_{it}$ makes demand less elastic (i.e.,
increases the markup), while an increase in consumption makes demand
more elastic (i.e., decreases the markup).

Under Assumption \ref{assu: revenue monotonicity}(a), given values
of $(m,w)$, $\phi_{t}(m,w,\cdot)$ in (\ref{eq:Phi}) can be interpreted
as the quantile function of revenue $r$. Although $m_{it}$ is endogenous
and correlated with $u_{it}$, Assumption \ref{A-0}(a)--(c) and equation
(\ref{eq:epsilon}) imply that $u_{it}$ is independent of $(m_{it-\upsilon-1},w_{it-\upsilon})$
while $u_{it}$ is serially correlated with $u_{is}$ for $s=1,...,v$.
Then, we have 
$\Pr[r_{it}\leq\phi_{t}\left(m_{it},w_{it},u\right)|m_{it-\upsilon-1},w_{it-\upsilon}]=u$
for all $u\in[0,1]$.\footnote{This follows because 
\[
\begin{aligned}\Pr[r_{it}\leq\phi_{t}\left(m_{it},w_{it},u\right)|m_{it-\upsilon-1},w_{it-\upsilon}]= & \Pr[\phi_{t}\left(m_{it},w_{it},u_{it}\right)\leq\phi_{t}\left(m_{it},w_{it},u\right)|m_{it-\upsilon-1},w_{it-\upsilon}]\\
= & \Pr[u_{it}\leq u|m_{it-\upsilon-1},w_{it-\upsilon}]\\
= & u,
\end{aligned}
\]
where the second equality follows from the monotonicity of $\phi_{t}\left(m,w,\cdot\right)$
while the last equality holds because $u_{it}\Perp(m_{it-\upsilon-1},w_{it-\upsilon})$.} 

Assumption \ref{assu: revenue monotonicity}(b), referred to as the
completeness condition, implies the following uniqueness property:
for any two candidate functions $\phi_{t}^{1}$ and $\phi_{t}^{2}$
and any fixed $u\in[0,1]$, $E\left[\left.1\left\{ r_{it}\le\phi_{t}^{1}\left(m_{it},w_{it},u\right)\right\} \right|m_{it-\upsilon-1},w_{it-\upsilon}\right]=E\left[\left.1\left\{ r_{it}\le\phi_{t}^{2}\left(m_{it},w_{it},u\right)\right\} \right|m_{it-\upsilon-1},w_{it-\upsilon}\right]$
a.s. implies that $\phi_{t}^{1}(\cdot,\cdot,u)=\phi_{t}^{2}(\cdot,\cdot,u)$
almost surely. Then, following \citet{chernozhukov2005iv}, the moment
condition 
\begin{equation}
E\left[\left.1\left\{ r_{it}\le\phi_{t}\left(m_{it},w_{it},u\right)\right\} -u\right|m_{it-\upsilon-1},w_{it-\upsilon}\right]=0\,\,\,\text{for }u\in[0,1]\label{eq: IVQR moment condition}
\end{equation}
identifies $\phi_{t}(\cdot)$, and the demand shock $u_{it}$ is identified
as $u_{it}=\phi_{t}^{-1}(r_{it},m_{it},w_{it})$ under Assumption
\ref{assu: revenue monotonicity}. 
\begin{prop}
\label{lem:step2} Under Assumptions \ref{A-0}, \ref{A-1}, \ref{A-data},
and \ref{assu: revenue monotonicity} hold, $\phi_{t}(\cdot)$ and
$u_{it}$ are identified. 
\end{prop}

Hereafter, $\phi_{t}(\cdot)$ and $u_{it}$ are assumed to be known.

\subsubsection{Step 2: Identification of Control Function and TFP}

From (\ref{omega}), the control function $\omega_{it}=\mathbb{M}_{t}^{-1}(m_{it},w_{it},u_{it})$
satisfies 
\begin{align}
\mathbb{M}_{t}^{-1}(m_{it},w_{it},u_{it}) & =\bar{h}_{t}\left(m_{it-1},w_{it-1},u_{it-1},z_{it-1}^{h}\right)+\eta_{it},\label{eq:model}
\end{align}
where $\bar{h}_{t}\left(m_{it-1},w_{it-1},u_{it-1},z_{it-1}^{h}\right):=h_{t}\left(\mathbb{M}_{t-1}^{-1}(m_{it-1},w_{it-1},u_{it-1}),z_{it-1}^{h}\right)$.
As $\partial\mathbb{M}_{t}^{-1}/\partial m_{it}>0$, given the values
of $(w_{it},u_{it})$, the dependent variable in (\ref{eq:model})
is a monotonic transformation of $m_{it}$. Therefore, the model (\ref{eq:model})
belongs to a class of transformation models, the identification of
which \citet{chiappori2015nonparametric} analyze.

We make the following assumption, which corresponds to Assumptions
A1--A3, A5, and A6 in \citet{chiappori2015nonparametric}.\footnote{Assumption \ref{A-1} (c) corresponds to Assumption A4 of \citet{chiappori2015nonparametric}.}


\begin{assumption} \label{A-3} 
(a) The distribution $G_{\eta_{t}}(\cdot)$ of $\eta_{it}$ is absolutely continuous with a density function $g_{\eta_{t}}(\cdot)$ that is continuous on its support. 
(b) $\eta_{it}$ is independent of $v_{it}:=(w_{it},u_{it},m_{it-1},w_{it-1},u_{it-1},z_{it-1}^{h})'$ with $E[\eta_{it}|v_{it}]=0$. 
(c) The support $\mathcal{M}\times \mathcal{V}$ of the vector $(m_{it}, v_{it}')$ is an open, connected subset of Euclidean space, and the normalization point defined in Assumption \ref{A-2} lies in the interior of $\mathcal{M}\times \mathcal{V}$. 
(d) The support $\varOmega$ of $\omega_{it}$ is an interval $[\text{\underbar{\ensuremath{\omega}}},\bar{\omega}]\subset\mathbb{R}$, where $\text{\underbar{\ensuremath{\omega}}}<0$ and $1<\bar{\omega}$. 
(e) $h(\cdot)$ is continuously differentiable with respect to $\left(\omega,z_{h}\right)$ on $\Omega\times\mathcal{Z}_{h}$. 
(f) Let $\mathcal{S}_{m,w,u}$ denote the joint support of $(m_{it}, w_{it}, u_{it})$. The set 
\[
\mathcal{A}_{q_{t-1}}:=\left\{ (m_{it-1},w_{it-1},u_{it-1},z_{it-1}^{h})\in\text{Proj}_{v}(\mathcal{M}\times \mathcal{V}):\frac{\partial G_{m_{t}|v_{t}}(m_{it}|v_{it})}{\partial q_{it-1}}\neq0\text{ for all }(m_{it},w_{it},u_{it})\in\mathcal{S}_{m,w,u}\right\}
\]
is nonempty for some $q_{it-1}\in\{m_{it-1},k_{it-1},l_{it-1},z_{it-1}^{s},z_{it-1}^{d},u_{it-1},z_{it-1}^{h}\}$, where $\text{Proj}_{v}(\mathcal{M}\times \mathcal{V}) := \left\{ v \in \mathbb{R}^{d_v} : \exists m \text{ such that } (m, v) \in \mathcal{M}\times \mathcal{V} \right\}$.
\end{assumption}

We can relax Assumption \ref{A-3}(b) by allowing $z_{it}^{h}$ and $l_{it}$ to correlate with $\eta_{it}$ as discussed in Appendix \ref{app:Alternative-Settings}. The sign restriction in Assumption \ref{A-3}(d) holds without loss of generality because we can choose any two points in place of $\{0,1\}$ on the support of $\omega_{it}$ without changing the essence of our argument.

Assumption \ref{A-3}(c) relaxes the "full support" condition (i.e., that the support is a Cartesian product of intervals) typically required in identification proofs of this type \citep[e.g.,][]{chiappori2015nonparametric}. By requiring only that the support $\mathcal{M}\times \mathcal{V}$ be an open, connected set, we accommodate data structures where inputs are highly persistent (e.g., $k_{it} \approx k_{it-1}$), which often results in "diagonal" support shapes. As long as the support is connected, identification is achieved via line integration as detailed in the proof.

Assumption \ref{A-3}(f) can be interpreted as a generalized rank condition. Suppose $g_{\eta_{t}}\left(\eta_{it}\right)>0$ for all $\eta_{it}\in\mathbb{R}$. Then, as will be shown below in (\ref{eq:dG3}), Assumption \ref{A-3}(f) holds if either 
\begin{align*}
\frac{\partial\bar{h}_{t}\left(\tilde{m}_{it-1},\tilde{w}_{it-1},\tilde{u}_{it-1},\tilde{z}_{it-1}^{h}\right)}{\partial\tilde{z}_{it-1}^{h}} & \neq0\text{ or }\\
\frac{\partial h_{t}\left(\mathbb{M}_{t-1}^{-1}(\tilde{m}_{it-1},\tilde{w}_{it-1},\tilde{u}_{it-1}),\tilde{z}_{it-1}^{h}\right)}{\partial\omega_{it-1}}\frac{\partial\mathbb{M}_{t-1}^{-1}(\tilde{m}_{it-1},\tilde{w}_{it-1},\tilde{u}_{it-1})}{\partial q_{it-1}} & \neq0
\end{align*}
holds for some vector of lagged variables $\left(\tilde{m}_{it-1},\tilde{w}_{it-1},\tilde{u}_{it-1},\tilde{z}_{it-1}^{h}\right)$ in the projection $\text{Proj}_{v}(\mathcal{M}\times \mathcal{V})$ and some instrument $q_{it-1}\in\{m_{it-1},k_{it-1},l_{it-1},z_{it-1}^{s},z_{it-1}^{d},u_{it-1},z_{it-1}^{h}\}$. The latter condition is equivalent to (i) $\omega_{it-1}$ has a causal impact on $\omega_{it}$ ($\partial h/\partial\omega_{it-1}\neq0$) and (ii) $q_{it-1}$ has a causal impact on $\omega_{it-1}$ ($\partial\mathbb{M}_{it-1}^{-1}/\partial q_{it-1}\neq0$). These conditions must be satisfied for at least one exogenous variable $q_{it-1}$ at some point in the support.

Proposition \ref{P-step1} shows that the control function is identified from the distribution of $(m_{it},v_{it})$. The identification of TFP also follows from Proposition \ref{P-step1} as $\omega_{it}=\mathbb{M}_{t}^{-1}(m_{it},w_{it},u_{it})$. 
\begin{prop}
\label{P-step1} Suppose that Assumptions \ref{A-0}--\ref{A-3} hold. Then, we can identify $\mathbb{M}_{t}^{-1}(\cdot)$ up to scale and location and $G_{\eta_{t}}(\cdot)$ up to the scale normalization of $\eta_{it}$. 
\end{prop}

\begin{proof}
The proof of Proposition \ref{P-step1} follows the identification strategy of \citet{chiappori2015nonparametric}, adapted to the connected support assumption.

In view of equation (\ref{eq:model}), the conditional distribution of $m_{it}$ given $v_{it}$ satisfies
\begin{align*}
G_{m_{t}|v_{t}}(m_{it}|v_{it}) & =G_{\eta_{t}}\left(\mathbb{M}_{t}^{-1}(m_{it},w_{it},u_{it})-\bar{h}_{t}\left(m_{it-1},w_{it-1},u_{it-1},z_{it-1}^{h}\right)\right),
\end{align*}
where the equality follows from $\eta_{it} \perp v_{it}$ in Assumption \ref{A-3}(b). Let $q_{it}\in\{m_{it},k_{it},l_{it},z_{it}^{s},z_{it}^{d},u_{it}\}$ and $q_{it-1}\in\{m_{it-1},k_{it-1},l_{it-1},z_{it-1}^{s},z_{it-1}^{d},z_{it-1}^{h},u_{it-1}\}$. The derivatives of $G_{m_{t}|v_{t}}(m_{it}|v_{it})$ are
\begin{align}
\frac{\partial G_{m_{t}|v_{t}}\left(m_{it}|v_{it}\right)}{\partial q_{it}} & =\frac{\partial\mathbb{M}_{t}^{-1}(m_{it},w_{it},u_{it})}{\partial q_{it}}g_{\eta_{t}}\left(\eta_{it}\right),\label{eq:dG2}\\
\frac{\partial G_{m_{t}|v_{t}}\left(m_{it}|v_{it}\right)}{\partial q_{it-1}} & =-\frac{\partial\bar{h}_{t}\left(m_{it-1},w_{it-1},u_{it-1},z_{it-1}^{h}\right)}{\partial q_{it-1}}g_{\eta_{t}}\left(\eta_{it}\right),\label{eq:dG3}
\end{align}
where $\eta_{it}=\mathbb{M}_{t}^{-1}(m_{it},w_{it},u_{it})-\bar{h}_{t}\left(m_{it-1},w_{it-1},u_{it-1},z_{it-1}^{h}\right)$.
Using Assumption \ref{A-3}(f), we can choose $q_{it-1}$ and $(\tilde{m}_{it-1},\tilde{w}_{it-1},\tilde{u}_{it-1},\tilde{z}_{it-1}^{h})\in\mathcal{A}_{q_{t-1}}$ such that both sides of (\ref{eq:dG3}) are non-zero. Dividing (\ref{eq:dG2}) by (\ref{eq:dG3}) yields:
\begin{align}
\frac{\partial\mathbb{M}_{t}^{-1}(m_{it},w_{it},u_{it})}{\partial q_{it}} & =-\frac{\partial\bar{h}_{t}\left(\tilde{m}_{it-1},\tilde{w}_{it-1},\tilde{u}_{it-1},\tilde{z}_{it-1}^{h}\right)}{\partial q_{it-1}} \nonumber\\
&\quad \times\frac{\partial G_{m_{t}|v_{t}}\left(m_{it}|w_{it},u_{it},\tilde{m}_{it-1},\tilde{w}_{it-1},\tilde{u}_{it-1},\tilde{z}_{it-1}^{h}\right)/\partial q_{it}}{\partial G_{m_{t}|v_{t}}\left(m_{it}|w_{it},u_{it},\tilde{m}_{it-1},\tilde{w}_{it-1},\tilde{u}_{it-1},\tilde{z}_{it-1}^{h}\right)/\partial q_{it-1}}.\label{eq:dM/dq}
\end{align}
Evaluating (\ref{eq:dM/dq}) for $q_{it}=m_{it}$ and applying the normalization Assumption \ref{A-2}, we identify the scaling factor:
\begin{align*}
1 & =\mathbb{M}_{t}^{-1}(m_{t1}^{*},w_t^*,u_t^*)-\mathbb{M}_{t}^{-1}(m_{t0}^{*},w_t^*,u_t^*)=-\frac{1}{S_{q_{t-1}}}\frac{\partial\bar{h}_{t}\left(\tilde{m}_{it-1},\tilde{w}_{it-1},\tilde{u}_{it-1},\tilde{z}_{it-1}^{h}\right)}{\partial q_{t-1}}, 
\end{align*}
where $$S_{q_{t-1}}:=\left(\int_{m_{t0}^*}^{m_{t1}^*} 
\frac{\partial G_{m_{t}|v_{t}}\left(m|w_t^*,u^*_t,\tilde{m}_{it-1},\dots\right)/\partial q_{it}}{\partial G_{m_t|v_{t}}\left( m|w_t^*,u^*_t,\tilde{m}_{it-1},\dots\right)/\partial q_{it-1}} dm\right)^{-1}.
$$ Thus, we identify $\partial\bar{h}_{t}/\partial q_{it-1}=-S_{q_{t-1}}$. Substituting this back into (\ref{eq:dM/dq}), the partial derivative for any argument $q_{it}$ is identified as:
\begin{align}
\frac{\partial\mathbb{M}_{t}^{-1}(m_{it},w_{it},u_{it})}{\partial q_{it}} & =S_{q_{t-1}}\frac{\partial G_{m_{t}|v_{t}}\left(m_{it}|w_{it},u_{it},\tilde{m}_{it-1},\dots\right)/\partial q_{t}}{\partial G_{m_{t}|v_{t}}\left(m_{it}|w_{it},u_{it},\tilde{m}_{it-1},\dots\right)/\partial q_{t-1}}.\label{eq:dM_dq_identified}
\end{align}

To recover the level of the function $\mathbb{M}_{t}^{-1}(\cdot)$, we rely on the Connected Support Assumption \ref{A-3}(c). Let $\mathbf{x}_0 := (m_{t0}^{*},k_{t}^{*},l_{t}^{*},z_{t}^{s*},z_{t}^{d*},u_{t}^{*})$ be the normalization point where $\mathbb{M}_{t}^{-1}(\mathbf{x}_0) = 0$. For any target point $\mathbf{x} := (m_{t},k_{t},l_{t},z_{t}^{s},z_{t}^{d},u_{t})$ in the interior of the support $\text{int}(\mathcal{S}_{m,w,u})$, there exists a piecewise smooth path $\gamma: [0,1] \to \text{int}(\mathcal{S}_{m,w,u})$ such that $\gamma(0) = \mathbf{x}_0$ and $\gamma(1) = \mathbf{x}$, lying entirely within the support.

By the Fundamental Theorem of Line Integrals \citep[][p. 1075]{stewart2012calculus}, we identify $\mathbb{M}_{t}^{-1}(\mathbf{x})$ as:
\begin{align}
\mathbb{M}_{t}^{-1}(\mathbf{x}) & = \mathbb{M}_{t}^{-1}(\mathbf{x}_0) + \int_{\gamma} \nabla \mathbb{M}_{t}^{-1}(\mathbf{z}) \cdot d\mathbf{z} \nonumber \\
& = 0 + \int_{0}^{1} \left[ \sum_{q \in \{m, k, l, z^s, z^d, u\}} \frac{\partial \mathbb{M}_{t}^{-1}(\gamma(\tau))}{\partial q_{it}} \frac{d\gamma_{q}(\tau)}{d\tau} \right] d\tau, \label{eq:M_t_line_integral}
\end{align}
where $\gamma_q(\tau)$ denotes the $q$-th component of $\gamma(\tau)$. Since the vector field $\nabla \mathbb{M}_{t}^{-1}$ is conservative, the integral is path-independent. This ensures uniquely identified values for $\mathbb{M}_{t}^{-1}(\cdot)$ regardless of the specific path chosen, provided the path remains within the connected support of the data.

Finally, from $\omega_{it}=\mathbb{M}_{t}^{-1}(m_{it},w_{it},u_{it})$, we identify $\bar{h}_{t}(\cdot) = E[\omega_{it}|\cdot]$ and the residual $\eta_{it}$, allowing for the identification of the distribution $G_{\eta_{t}}(\cdot)$.
\end{proof}

\begin{rem}[Support Requirements and Persistence]
\label{rem:limited_support}
Assumption \ref{A-3}(f) and the integration in (\ref{eq:M_t_line_integral}) require the support of observables to be \emph{connected}, relaxing the "full support" condition discussed in \cite{chiappori2015nonparametric}. While inputs like capital are persistent and have bounded conditional support, identification remains valid if the union of these supports forms a connected set. Local variation identifies partial derivatives pointwise; global identification is then achieved by integrating ("stitching") these derivatives along any path connecting the normalization point to the point of interest. Thus, high persistence does not preclude identification provided there are no isolated "islands" of data. 
\end{rem}

\subsubsection{Step 3: Identification of Production Function, Markup, and Demand Function}

The final step identifies the production function, markup, and demand function. From $\text{\ensuremath{\phi_{t}}(\ensuremath{m_{it}},\ensuremath{w_{it}},\ensuremath{u_{it}})}=\varphi_{t}\left(f_{t}(x_{it},z_{it}^{s})+\mathbb{M}_{t}^{-1}\left(m_{it},w_{it},u_{it}\right),z_{it}^{d},u_{it}\right)$ and the monotonicity of $\varphi_{t}$, differentiating $\varphi_{t}^{-1}(\phi_{t}(m_{it},w_{it},u_{it}),z_{it}^{d},u_{it})=f_{t}(x_{it},z_{it}^{s})+\mathbb{M}_{t}^{-1}\left(m_{it},w_{it},u_{it}\right)$ with respect to $q_{it}^{s}\in\{m_{it},k_{it},l_{it},z_{it}^{s}\}$ and $q_{it}^{d}\in\{z_{it}^{d},u_{it}\}$ gives: 
\begin{align}
\frac{\partial\varphi_{t}^{-1}(r_{it},z_{it}^{d},u_{it})}{\partial r_{it}}\frac{\partial\phi_{t}(m_{it},w_{it},u_{it})}{\partial q_{it}^{s}} & =\frac{\partial f_{t}(x_{it},z_{it}^{s})}{\partial q_{it}^{s}}+\frac{\partial\mathbb{M}_{t}^{-1}\left(m_{it},w_{it},u_{it}\right)}{\partial q_{it}^{s}},\label{eq:phi_qt}\\
\frac{\partial\varphi_{t}^{-1}(r_{it},z_{it}^{d},u_{it})}{\partial r_{it}}\frac{\partial\phi_{t}(m_{it},w_{it},u_{it})}{\partial q_{it}^{d}} & =\frac{\partial\mathbb{M}_{t}^{-1}\left(m_{it},w_{it},u_{it}\right)}{\partial q_{it}^{d}}-\frac{\partial\varphi_{t}^{-1}(r_{it},z_{it}^{d},u_{it})}{\partial q_{it}^{d}}.\label{eq:phi_zt}
\end{align}
Note that $\partial\varphi_{t}^{-1}(r_{it},z_{it}^{d},u_{it})/\partial r_{it}=\left(\partial\varphi_{t}(y_{it},z_{it}^{d},u_{it})/\partial y_{it}\right)^{-1}$ represents the markup from (\ref{eq:markup_foc}). If the markup $\partial\varphi_{t}^{-1}(r_{it},z_{it}^{d},u_{it})/\partial r_{it}$ were known, then equations (\ref{eq:phi_qt}) and (\ref{eq:phi_zt}) could identify $\partial f_{t}(x_{it},z_{it}^{s})/\partial q_{it}^{s}$ and $\partial\varphi_{t}^{-1}(r_{it},z_{it}^{d},u_{it})/\partial q_{it}^{d}$ given that $\mathbb{M}_{t}^{-1}(m_{it},w_{it},u_{it})$ is identified. However, since the markup is unknown, identification requires further restriction. Following \citet*{doraszelski2013,doraszelski2018measuring}  and \citet{gandhi2020identification}, we use the first-order condition with respect to the material as an additional restriction. 

\begin{assumption} \label{A-FOC} The first-order condition with respect to material for the profit maximization problem (\ref{eq:profit_maximization})
\begin{equation}
\frac{\partial f_{t}(x_{it},z_{it}^{s})}{\partial m_{it}}=\frac{\partial\varphi_{t}^{-1}(r_{it},z_{it}^{d},u_{it})}{\partial r_{it}}\text{\ensuremath{\frac{\exp(p_{t}^{m}+m_{it})}{\exp\left(r_{it}\right)}}}\label{eq:foc_m}
\end{equation}
holds for all firms. \end{assumption}


Rearranging the first-order condition, we obtain the markup equation used by \citet{de2012markups}:
\begin{equation}
\frac{\partial\varphi_{t}^{-1}(r_{it},z_{it}^{d},u_{it})}{\partial r_{it}}=\frac{\partial f_{t}(x_{it},z_{it}^{s})/\partial m_{it}}{\exp(p_{t}^{m}+m_{it})/\exp\left(r_{it}\right)}.\label{eq:DLW}
\end{equation}

We establish the following proposition. 
\begin{prop}
\label{P-step3} Suppose that Assumptions \ref{A-0}--\ref{A-FOC} hold. Then, we can identify $\varphi_{t}^{-1}(\cdot)$, $f_{t}(\cdot)$, and $\psi_{t}(\cdot)$ up to scale and location and each firm's markup $\partial\varphi_{t}^{-1}(r_{it},z_{it}^{d},u_{it})/\partial r_{it}$ up to scale.
\end{prop}

\begin{proof}
From (\ref{eq:phi_qt}) and (\ref{eq:foc_m}), the markup $\partial\varphi_{t}^{-1}(r_{it},z_{it}^{d},u_{it})/\partial r_{it}$ is identified for each firm as
\begin{equation}
\frac{\partial\varphi_{t}^{-1}(r_{it},z_{it}^{d},u_{it})}{\partial r_{it}}=\frac{\partial\mathbb{M}_{t}^{-1}\left(m_{it},w_{it},u_{it}\right)}{\partial m_{it}}\left(\frac{\partial\phi_{t}(m_{it},w_{it},u_{it})}{\partial m_{it}}-\frac{\exp(p_{t}^{m}+m_{it})}{\exp\left(r_{it}\right)}\right)^{-1}.\label{eq:markup_identified}
\end{equation}
From $\phi_{t}$ and (\ref{eq:markup_identified}), the markup is expressed as a function of $(m_{it},w_{it},u_{it})$ as
\begin{align}
\mu_{t}(m_{it},w_{it},u_{it}) & :=\frac{\partial\varphi_{t}^{-1}(\phi_{t}(m_{it},w_{it},u_{it}),z_{it}^{d},u_{it})}{\partial r_{it}},\label{eq:markup_function}
\end{align}
which we can identify from data and the identified markups. Substituting (\ref{eq:markup_function}) into (\ref{eq:phi_qt}), we identify $\partial f_{t}(x_{it},z_{it}^{s})/\partial q_{it}$ for $q_{it}^{s}\in\{m_{it},k_{it},l_{it},z_{it}^{s}\}$ as follows:
\begin{align}
\frac{\partial f_{t}(x_{it},z_{it}^{s})}{\partial q_{it}} & =\mu_{t}(m_{it},w_{it},u_{it})\frac{\partial\phi_{t}(m_{it},w_{it},u_{it})}{\partial q_{it}}-\frac{\partial\mathbb{M}_{t}^{-1}\left(m_{it},w_{it},u_{it}\right)}{\partial q_{it}}.\label{eq:output_elasticities_pro2}
\end{align}
Consequently, the gradient field $\nabla f_{t}$ is identified everywhere on the support. 

To recover the production function level $f_{t}(x_{t},z_{t}^{s})$ from these identified derivatives, we apply the connected support condition. Let $\mathbf{v}_{it} := (m_{it}, k_{it}, l_{it}, z_{it}^s)$ denote the vector of production function arguments. Let $\mathbf{v}^* := (m_{t0}^{*}, k_{t}^{*}, l_{t}^{*}, z_{t}^{s*})$ be the normalization point where $f_t(\mathbf{v}^*) = 0$ (Assumption \ref{A-2}). For any point $\mathbf{v}_{it}$ in the connected support, there exists a path $\gamma$ connecting $\mathbf{v}^*$ to $\mathbf{v}_{it}$. By the Fundamental Theorem of Line Integrals, we identify $f_t$ as:
\begin{align}
f_{t}(\mathbf{v}_{it}) & = f_t(\mathbf{v}^*) + \int_{\gamma} \nabla f_{t}(\mathbf{z}) \cdot d\mathbf{z} \nonumber \\
& = 0 + \int_{0}^{1} \left[ \sum_{q \in \{m, k, l, z^s\}} \frac{\partial f_{t}(\gamma(\tau))}{\partial q_{it}} \frac{d\gamma_{q}(\tau)}{d\tau} \right] d\tau. \label{eq:f_line_integral}
\end{align}

Let $\mathcal{R}:=\{r_{t}:r_{t}=\phi_{t}(m_{t},w_{t},u_{t})\ \text{{for}\ \text{{some}\ }}(m_{t},w_{t},u_{t})\in\mathcal{X}\times\mathcal{Z}\times[0,1]\}$ be the support of $r_{t}$. For given $(r_{t},z_{t}^{d})\in\mathcal{R}\times\mathcal{Z}_{d}$, $B_{t}(r_{t},z_{t}^{d},u_{t}):=\{\left(x_{t},z_{t}^{s}\right)\in\mathcal{X}\times\mathcal{Z}_{s}:\phi_{t}(x_{t},z_{t}^{s},z_{t}^{d},u_{t})=r_{t}\}$ is non-empty by the construction of $\mathcal{R}$. Then, because $f_{t}(x_{t},z_{t}^{s})$ and $\mathbb{M}_{t}^{-1}(m_{t},w_{t},u_{t})$ are identified, the output quantity $\varphi_{t}^{-1}(r_{t},z_{t},u_{t})$ for any $(r_{t},z_{t},u_{t})\in\mathcal{R}\times\mathcal{Z}\times[0,1]$ is identified by 
\[
\varphi_{t}^{-1}(r_{t},z_{t}^{d},u_{t})=f_{t}(x_{t},z_{t}^{s})+\mathbb{M}_{t}^{-1}(m_{t},w_{t},u_{t})\text{ for }\left(x_{t},z_{t}^{s}\right)\in B_{t}(r_{t},z_{t}^{d},u_{t}).
\]
By monotonicity, $\varphi_{t}(y_{t},z_{t}^{d},u_{t})$ is identified from $\varphi_{t}^{-1}(r_{t},z_{t}^{d},u_{t})$.  Then, we can identify $\psi_{t}(y_{t},z_{t}^{d},u_{t})$ as $\psi_{t}(y_{t},z_{t}^{d},u_{t})=\varphi_{t}(y_{t},z_{t}^{d},u_{t})-y_{t}$.
\end{proof}
The output quantity and price for individual firms are identified
as $y_{it}=\varphi_{t}^{-1}(r_{it},z_{it}^{d},u_{it})$ and $p_{it}=\psi_{t}(y_{it},z_{it}^{d},u_{it})=r_{it}-\varphi_{t}^{-1}(r_{it},z_{it}^{d},u_{it})$,
respectively. 
\begin{cor}
Suppose that Assumptions \ref{A-0}--\ref{A-FOC} hold. Then, the
production function, the demand function, output quantities, output
prices, and TFP are identified up to scale and location; markups and
output elasticities are identified up to scale. 
\end{cor}

\begin{rem}
Examination of the proofs reveals that we have over-identifying restrictions.
In particular, the proof of Proposition \ref{P-step1} goes through
with any choice of $q_{it-1}\in\{k_{it-1},l_{it-1},m_{it-1},z_{it-1}^{s},z_{it-1}^{d},u_{it-1},z_{it-1}^{h}\}$
in (\ref{eq:dM_dq_identified}). Furthermore, the proof of Proposition
\ref{P-step3} does not rely on the restriction in (\ref{eq:phi_zt})
for identifying $\varphi_{t}^{-1}(\cdot)$. These over-identifying
restrictions can be exploited to construct specification tests for
the model and to obtain more efficient estimation. 
\end{rem}

\subsubsection{Comparison to Existing Identification Approaches}

Our setup extends existing identification analyses of production functions
by allowing prices to depend on output through an inverse demand function
and by incorporating transitory unobserved demand shocks as a source
of heterogeneous markups. While our approach builds on existing identification
methods, our use of control functions and the IVQR framework differs
from conventional formulations.

First, because the model includes both productivity and demand shocks,
the standard control function approach cannot account for two sources
of unobserved heterogeneity. We therefore assume that demand shocks
are transitory while productivity shocks are persistent and use the
IVQR approach to identify demand shocks in Step 1.

Second, Step 2 identifies the control function from the dynamics of
input choices without relying on output measures, distinguishing our
approach from the standard control function framework (e.g.,\citealp{olley1996dynamics};
\citealp{levinsohn2003estimating}; \citealp{ackerberg2015identification}).

Third, \citet{ackerberg2015identification} identify a structural
value-added function, $y_{it}=\tilde{f}_{t}(k_{it},l_{it})+\omega_{it}$,
derived under perfect competition from a Leontief production function
$y_{it}=\min\{\tilde{f}_{t}(k_{it},l_{it})+\omega_{it},a+m_{it}\}$.
This formulation is difficult to apply under imperfect competition
because $y_{it}<\tilde{f}_{t}(k_{it},l_{it})+\omega_{it}$ can occur.
The maximum output capacity $y_{it}^{*}:=\tilde{f}_{t}(k_{it},l_{it})+\omega_{it}$
is determined before a firm chooses $m_{it}$ and $y_{it}$, so when
$y_{it}^{*}$ is large---e.g., due to a high productivity shock---a
profit-maximizing firm may produce $y_{it}<y_{it}^{*}$.\footnote{As noted by \citet{ackerberg2015identification}, under perfect competition
$y_{it}<y_{it}^{*}$ implies zero output, so only firms with $y_{it}=y_{it}^{*}$
are observed. Under imperfect competition, however, positive output
with $y_{it}<y_{it}^{*}$ is possible.} Intuitively, when TFP doubles, a firm may avoid a large price decline
by expanding output less than proportionally.

Fourth, our approach differs from \citet{gandhi2020identification}
in the use of the first-order condition for materials. Their method
identifies the material elasticity $\partial f_{t}(x_{it},z_{it}^{s})/\partial m_{it}$
from the first-order condition (\ref{eq:foc_m}) in their first step:
$\ln\frac{\partial f_{t}(x_{it},z_{it}^{s})}{\partial m_{it}}=\ln\text{\ensuremath{\frac{\exp(p_{t}^{m}+m_{it})}{\exp\left(r_{it}\right)}}}+\ln\frac{\partial\varphi_{t}^{-1}(r_{it},z_{it}^{d},u_{it})}{\partial r_{it}}$,
assuming perfect competition or identical constant markups where $\partial\varphi_{t}^{-1}(r_{it},z_{it}^{d},u_{it})/\partial r_{it}$
becomes a common constant for all $i$. Under imperfect competition
with variable markups, when the markup depends on revenue $r_{it}$,
$\partial f_{t}(x_{it},z_{it}^{s})/\partial m_{it}$ cannot be identified
solely from the first-order condition.

\subsection{Extensions to Alternative Settings}

\label{sec:extensions} We demonstrate in the Appendix that our identification
strategy extends to settings where the baseline assumptions are relaxed.
First, we establish identification with endogenous labor inputs in
Appendix \ref{subsec:Endogenous-Labor-Input} by incorporating labor
adjustment costs, which allow lagged labor to serve as a valid instrument.
Second, we address endogenous firm characteristics in Appendix \ref{subsec:Endogenous-Firm-Characteristics}
by employing a control function approach for triangular systems. Third,
Appendix \ref{subsec:Discrete_z} confirms that our identification
results---particularly the instrumental variable quantile regression
in Step 1---remain valid when firm characteristics are discrete rather
than continuous.

We also address the restrictiveness of assuming limited persistence
in demand shocks through two extensions. First, Appendix \ref{app:persistent_demand}
demonstrates that identification holds under persistent shocks (e.g.,
AR(1)) if a lagged supply shifter, such as supply-driven R\&D ($z_{it-1}^{h}$),
serves as the instrument. Because $z_{it-1}^{h}$ shifts productivity
but is excluded from demand, it satisfies the completeness condition
despite serial correlation. Second, Appendix \ref{app:quality} explicitly
models persistent unobserved quality. This captures the persistent
component of demand heterogeneity, allowing $\epsilon_{it}$ to represent
only the remaining high-frequency, transitory fluctuations.


\subsection{Fixing Normalization across Periods}

\label{subsec:Normalization}

Let $(\varphi_{t}^{-1}(\cdot),f_{t}(\cdot),\mathbb{M}_{t}^{-1}(\cdot))$
be a model structure for period $t$ identified by using Propositions
\ref{P-step1} and \ref{P-step3} under the normalization in Assumption
\ref{A-2}. Let $(\varphi_{t}^{*-1}(\cdot),f_{t}^{*}(\cdot),\mathbb{M}_{t}^{*-1}(\cdot))$
denote the true model structure. Since the structure is identified
up to scale and location normalization, there exist period-specific
location and scale parameters $(a_{1t},a_{2t},b_{t})\in\mathbb{R}^{2}\times\mathbb{R}_{+}$
such as 
\begin{align}
\varphi_{t}^{-1}(r_{it},z_{it}^{d},u_{it}) & =a_{1t}+a_{2t}+b_{t}\varphi_{t}^{*-1}(r_{it},z_{it}^{d},u_{it}),\ f_{t}(x_{it},z_{it}^{s})=a_{1t}+b_{t}f_{t}^{*}(x_{it},z_{it}^{s}),\ \nonumber \\
\mathbb{M}_{t}^{-1}(m_{it},w_{it},u_{it}) & =a_{2t}+b_{t}\mathbb{M}_{t}^{*-1}(m_{it},w_{it},u_{it}).\label{eq:norm-t}
\end{align}
Generally speaking, the location and scale normalization differ across
periods---that is, $(a_{1t},a_{2t},b_{t})\neq(a_{1t+1},a_{2t+1},b_{t+1})$.
For the identified objects to be comparable across periods, we need
to fix normalization across periods by assuming that some object in
the model is time-invariant. The subsection discusses these additional
assumptions.\footnote{\citet{klette1996jae} and \citet{de2011product} identify the levels
of markups and output elasticities from revenue data by using a functional
form property of a demand function. They consider a constant elastic
demand function leading to $\varphi_{t}(y_{it},z_{it})=\alpha y_{it}-(\alpha-1)z_{it}$
where $z_{it}$ is an aggregate demand shifter, which is an weighted
average of revenue across firms, and $\alpha$ is an unknown parameter.
This formulation implies $\varphi_{t}^{-1}(r_{it},z_{it})=(1/\alpha)r_{it}+(1-1/\alpha)z_{it}$
and imposes a linear restriction $\partial\varphi_{t}^{-1}(r_{it},z_{it})/\partial r_{it}+\partial\varphi_{t}^{-1}(r_{it},z_{it})/\partial z_{it}=1$,
which fixes the scale parameter $b_{t}$.}

\subsubsection{Scale Normalization}

From (\ref{eq:norm-t}), the ratio of identified markups across two
periods relates to the ratio of true markups as 
\[
\frac{\partial\varphi_{t+1}^{-1}(r_{it+1},z_{it+1}^{d},u_{it+1})/\partial r}{\partial\varphi_{t}^{-1}(r_{it},z_{it}^{d},u_{it})/\partial r}=\frac{b_{t+1}}{b_{t}}\frac{\partial\varphi_{t+1}^{*-1}(r_{it+1},z_{it+1}^{d},u_{it+1})/\partial r}{\partial\varphi_{t}^{*-1}(r_{it},z_{it}^{d},u_{it})/\partial r}.
\]
Therefore, the ability to identify how true markups change over two
periods requires identification of the ratio of scale parameters,
$b_{t+1}/b_{t}$. Similarly, the ratio of identified output elasticities
across periods and that of identified TFP deviation from the mean
are related to their true values via the ratio of scale parameters:
\begin{align*}
\frac{\partial f_{t+1}(x_{it+1},z_{it+1}^{s})/\partial q}{\partial f_{t}(x_{it},z_{it}^{s})/\partial q} & =\frac{b_{t+1}}{b_{t}}\frac{\partial f_{t+1}^{*}(x_{it+1},z_{it+1}^{s})/\partial q}{\partial f_{t}^{*}(x_{it},z_{it}^{s})/\partial q},\frac{\omega_{it+1}-E\left[\omega_{it+1}\right]}{\omega_{it}-E\left[\omega_{it}\right]}=\frac{b_{t+1}}{b_{t}}\left(\frac{\omega_{it+1}^{*}-E\left[\omega_{it+1}^{*}\right]}{\omega_{it}^{*}-E\left[\omega_{it}^{*}\right]}\right)
\end{align*}
for $q\in\{m,k,l,z^{s}\}$.

To identify $b_{t+1}/b_{t}$, we consider the following assumptions.
\begin{assumption} \label{A-period} At least one of the following
conditions (a)--(c) holds. (a) The unconditional variance of $\eta_{it}$
does not change over time. (b) For some known interval $\mathcal{B}$
of $\mathcal{X}$ and some known point $z^{s}\in\mathcal{Z}_{s}$,
the output elasticity of one of the inputs evaluated at $z_{it}^{s}=z^{s}$
does not change over time for all $x\in\mathcal{B}$. (c) For some
known interval $\mathcal{B}$ of $\mathcal{X}$ and some known point
$z^{s}\in\mathcal{Z}_{s}$, the sum of output elasticities of the
three inputs evaluated at $z_{it}^{s}=z^{s}$ does not change over
time for all $x\in\mathcal{B}$. \end{assumption} Assumption \ref{A-period}(a)
holds, for example, if the productivity shock $\omega_{it}$ follows
a stationary process because stationarity requires that the distribution
of $\eta_{it}$ does not change over time. Assumption \ref{A-period}(b)
assumes that the elasticity of output with respect to one input does
not change over time for some known interval; meanwhile, under Assumption
\ref{A-period}(c), returns to scale in production technology does
not change for some known interval of inputs. 
\begin{prop}
\label{P-3} Suppose that Assumptions \ref{A-0}--\ref{A-period}
hold for time $t$ and $t+1$. Then, we can identify the ratio of
markups between two periods $t$ and $t+1$, the ratio of output elasticities
between $t$ and $t+1$, and the ratio of TFP deviation from the mean
between $t$ and $t+1$. 
\end{prop}

The proof is given in Appendix \ref{subsec:-proof-Pro6}.

\subsubsection{Local Constant Returns to Scale}

\label{subsec:LCRS}

We consider the following local constant returns to scale that strengthens
Assumption \ref{A-period}(c). \begin{assumption} \label{A-LCRS}
(Local Constant Returns to Scale) For some known interval $\mathcal{B}$
of $\mathcal{X}$and some known point $z^{s}\in\mathcal{Z}_{s}$,
the sum of the output elasticities of the three inputs evaluated at
$z_{it}^{s}=z_{t}^{s}$ equals to 1 for all $x\in\mathcal{B}$. \end{assumption}
Assumption \ref{A-LCRS} is stronger than Assumption \ref{A-period}(c)
but weaker than those used in some studies of markup estimation. In
particular, markups are often estimated as the ratio of revenue $\exp(r_{it})$
to total cost $TC_{it}$ under the assumption of a linear cost function
$TC_{it}=MC_{it}y_{it}$ with constant marginal cost $MC_{it}$. Such
a linear cost function requires stronger conditions than Assumption
\ref{A-LCRS}: (i) global constant returns to scale for all $x\in\mathcal{X}$,
(ii) full flexibility of all inputs, and (iii) price-taking behavior
in all input markets. By contrast, under Assumption \ref{A-LCRS},
marginal cost may increase with output, especially in the short run
when dynamic inputs such as capital entail adjustment costs.

With Assumption \ref{A-LCRS}, the scale normalization parameter $b_{t}$
can be identified for all periods as follows. Let $f_{t}(x_{t},z_{t}^{s})$
be the identified production function under Assumption \ref{A-2}
and $f_{t}^{*}(x_{t})$ be the true one where $f_{t}(x_{t},z_{t}^{s})=a_{t}+b_{t}f_{t}^{*}(x_{t},z_{t}^{s})$
from (\ref{eq:norm-t}). For $x\in\mathcal{B}$, we have 
\[
b_{t}=b_{t}\left(\frac{\partial f_{t}^{*}(x,z^{s})}{\partial m}+\frac{\partial f_{t}^{*}(x,z^{s})}{\partial k}+\frac{\partial f_{t}^{*}(x,z^{s})}{\partial l}\right)=\frac{\partial f_{t}(x,z^{s})}{\partial m}+\frac{\partial f_{t}(x,z^{s})}{\partial k}+\frac{\partial f_{t}(x,z^{s})}{\partial l}.
\]
Given that we have identified the scale parameter $b_{t}$ in (\ref{eq:norm-t}),
we have established the following proposition. 
\begin{prop}
\label{P-CRS} Suppose that Assumptions \ref{A-0}--\ref{A-FOC}
and \ref{A-LCRS} hold. Then, $\varphi_{t}(\cdot)$, $f_{t}(\cdot)$,
and $\psi_{t}(\cdot)$ can be identified up to location. The levels
of markup and output elasticities can be identified. Output quantity,
output price, and TFP can be identified up to location. \end{prop}

\subsubsection{Location Normalization}

Suppose that scale normalization $b_{t}$ is already identified---for
example, from Proposition \ref{P-CRS}. Define 
\begin{align}
\tilde{\varphi}_{t}^{-1}(r_{it},z_{it}^{d},u_{it}) & :=\varphi_{t}^{-1}(r_{it},z_{it}^{d},u_{it})/b_{t},\ \tilde{f}_{t}(x_{it},z_{it}^{s}):=f_{t}(x_{it},z_{it}^{s})/b_{t},\ \tilde{\omega}_{it}:=\omega_{it}/b_{t},\nonumber \\
\ \tilde{a}_{1t} & :=a_{1t}/b_{t},\text{\ {and}}\ \tilde{a}_{2t}:=a_{2t}/b_{t}.\label{eq:norm-tilde}
\end{align}
Then, (\ref{eq:norm-t}) is written as 
\begin{align}
\tilde{\varphi}_{t}^{-1}(r_{it},z_{it}^{d},u_{it}) & =\tilde{a}_{1t}+\tilde{a}_{2t}+\varphi_{t}^{*-1}(r_{it},z_{it}^{d},u_{it}),\ \tilde{f}_{t}(x_{it},z_{it}^{s})=\tilde{a}_{1t}+f_{t}^{*}(x_{it},z_{it}^{s}),\ \tilde{\omega}_{it}=\tilde{a}_{2t}+\omega_{it}^{*}.\label{eq:norm-t2}
\end{align}
From (\ref{eq:norm-t}), the growth rates (log differences) of the
identified output and TFP between $t$ and $t+1$ are related to their
true values as follows: 
\begin{align}
\tilde{\varphi}_{t+1}^{-1}(r_{it+1},z_{it+1}^{d},u_{it+1})-\tilde{\varphi}_{t}^{-1}(r_{it},z_{it}^{d},u_{it}) & =\tilde{a}_{1t+1}+\tilde{a}_{2t+1}-\tilde{a}_{1t}-\tilde{a}_{2t}\nonumber \\
 & \quad+\varphi_{t+1}^{*-1}(r_{it+1},z_{it+1}^{d},u_{it+1})-\varphi_{t}^{*-1}(r_{it},z_{it}^{d},u_{it}),\nonumber \\
\tilde{f}_{t+1}(x_{it+1},z_{it+1}^{s})-\tilde{f}_{t}(x_{it},z_{it}^{s}) & =\tilde{a}_{1t+1}-\tilde{a}_{1t}+f_{t}^{*}(x_{it+1},z_{it+1}^{s})-f_{t}^{*}(x_{it},z_{it}^{s}),\nonumber \\
\tilde{\omega}_{it+1}-\tilde{\omega}_{it} & =\tilde{a}_{2t+1}-\tilde{a}_{2t}+\omega_{it+1}^{*}-\omega_{it}^{*}.\label{eq:growth}
\end{align}
Therefore, to identify the growth rates of output and TFP, we need
to identify the changes in the location parameters. To do so, we can
use an industry-level producer price index $P_{t}^{*}$, which is
often available as data, to identify the change in the location parameters.
Suppose that $P_{t}^{*}$ is a Laspeyres index 
\begin{equation}
P_{t}^{*}:=\frac{\sum_{i\in\tilde{N}}\exp(p_{it}^{*}+y_{i0}^{*})}{\sum_{i\in\tilde{N}}\exp(p_{i0}^{*}+y_{i0}^{*})},\label{eq:PPI}
\end{equation}
where $\tilde{N}$ is a known set (or a random sample) of products.
$p_{i0}^{*}$ and $y_{i0}^{*}$ are firm $i$'s log true price and
log true output at the base period, respectively. The following argument
holds for forms of a price index (other than Laspeyres) as long as
the price index is a known function of prices that is homogenous of
degree $1$, which is typically satisfied. \begin{assumption} \label{A-location}(a)
The industry-level producer price index $P_{t}^{*}$ is known as data.
(b) For some known point $\left(\bar{x},\bar{z}^{s}\right)\in\mathcal{X}\times\mathcal{Z}_{s}$
and the true production functions of $t$ and $t+1$, $f_{t}^{*}(\cdot)$
and $f_{t+1}^{*}(\cdot)$, satisfy $f_{t}^{*}(\bar{x},\bar{z}^{s})=f_{t+1}^{*}(\bar{x},\bar{z}^{s})$.
\end{assumption} Assumption \ref{A-location}(b) is innocuous, implying
that any output change between $t$ and $t+1$ when inputs are fixed
at $\bar{x}$ is attributed to a TFP change.

Using the aggregate price index, we can identify the change in the
location parameters and identify the growth of TFP and output. 
\begin{prop}
\label{P-location} Suppose Assumptions \ref{A-0}--\ref{A-FOC},
\ref{A-LCRS}, and \ref{A-location} hold. Then, the true growth rate
of output $\varphi_{t+1}^{*-1}(r_{it+1},z_{it+1}^{d},u_{it+1})-\varphi_{t}^{*-1}(r_{it},z_{it}^{d},u_{it})$
and that of TFP $\omega_{it+1}^{*}-\omega_{it}^{*}$ can be identified
for each firm. 
\end{prop}

The proof is given in Appendix \ref{subsec:-Proof-pro8}.

\subsection{Identification of HSA Demand System, Utility Function, and Counterfactual
Welfare Effects}

\label{subsec:demand}


Given that we have identified each firm's output price and quantity,
it is possible to nonparametrically identify a Homothetic Single Aggregator
(HSA) system of demand functions and the associated utility function
of a representative consumer under an additional homotheticity shape
restriction. Furthermore, we may identify the welfare effects of counterfactual
experiments. \citet{matsuyama2017beyond} further propose two additional
families of homothetic demand systems---the homothetic demand system
with direct implicit additivity (HDIA) and that with indirect implicit
additivity (HIIA). In the Appendix \ref{app:alternative-demand},
we show that the analysis in this subsection extends to the HDIA and
HIIA cases. Notably, the HDIA family nests the Kimball aggregator,
which is widely used in Macroeconomics, as a special case.

\subsubsection{Identification of the HSA demand system and consumer preference}

\paragraph{The HSA demand system}

We adopt the Homothetic Single Aggregator (HSA) demand system proposed
by \citet{matsuyama2017beyond}.\footnote{The HSA system can be expressed as a system of direct demand functions
or of inverse demand functions. These two systems are self-dual, meaning
either can be derived from the other. \citet{matsuyama2023non} and
\citet{matsuyama2025homothetic} provide comprehensive reviews of
flexible extensions to the CES demand system, including the HSA framework.} Let $N_{t}$ denote the number of firms in the industry. The HSA
demand system is characterized by the \textit{structural} budget-share
function $\mathfrak{s}_{t}^{*}(\cdot,z_{it}^{d},u_{it})$, which returns
the log market revenue share of product $i$ as a function of its
log relative output quantity, defined by $y_{it}-q_{t}(\mathbf{y}_{t},\mathbf{z}_{t}^{d},\mathbf{u}_{t})$.
This relationship is expressed as: 
\begin{equation}
\ln\left(\frac{\exp(r_{it})}{\sum_{j=1}^{N_{t}}\exp\left(r_{jt}\right)}\right)=\mathfrak{s}_{t}^{*}\left(y_{it}-q_{t}\left(\mathbf{y}_{t},\mathbf{z}_{t}^{d},\mathbf{u}_{t}\right),z_{it}^{d},u_{it}\right)\quad\text{for }i=1,...,N_{t},\label{eq:structural-bs}
\end{equation}
where $\mathbf{y}_{t}:=(y_{1t},...,y_{N_{t}t})\in\mathcal{Y}^{N_{t}}$
is a vector of consumption, $\mathbf{z}_{t}^{d}:=(z_{1t}^{d},...,z_{N_{t}t}^{d})$
is a vector of observable demand shifters, $\mathbf{u}_{t}:=(u_{1t},...,u_{N_{t}t})$
is a vector of demand shocks, and $q_{t}(\mathbf{y}_{t},\mathbf{z}_{t}^{d},\mathbf{u}_{t})$
is the aggregate quantity index summarizing interactions across products.\footnote{If the utility function is CES, $U(\mathbf{y}_{t})=\left[\sum_{i=1}^{N_{t}}\exp\left(\rho y_{it}\right)\right]^{1/\rho}$,
the inverse demand becomes $p_{it}=\rho\left(y_{it}-\ln U(\mathbf{y}_{t})\right)+\Phi_{t}-y_{it}$.
Here, $\ln U(\mathbf{y}_{t})=q_{t}\left(\mathbf{y}_{t},\mathbf{z}_{t}^{d},\mathbf{u}_{t}\right)$,
meaning the quantity index coincides with the utility function, though
they generally differ in the HSA framework.} In equilibrium, the quantity index $q_{t}(\mathbf{y}_{t},\mathbf{z}_{t}^{d},\mathbf{u}_{t})$
is uniquely determined by the adding-up constraint on market shares:
\begin{align}
1 & =\sum_{i=1}^{N_{t}}\exp\left(\mathfrak{s}_{t}^{*}\left(y_{it}-q_{t}(\mathbf{y}_{t},{\mathbf{z}}_{t}^{d},\mathbf{u}_{t}),z_{it}^{d},u_{it}\right)\right).\label{eq:market_share_constraint}
\end{align}
Because $\mathfrak{s}_{t}^{*}(\cdot)$ is nonparametric, the HSA framework
nests various demand systems established in the literature, including
the CES and symmetric translog demand systems (\citealp{feenstra2003homothetic};
\citealp{feenstra2017globalization}).\footnote{See \citet{matsuyama2020constant} for details on how the HSA nests
translog demand.}

\paragraph{Identification of the HSA demand system}

To identify the demand system, we maintain the following assumptions.
Let $\Phi_{t}$ denote the log of aggregate expenditure for a representative
consumer, where $\Phi_{t}=\ln\left(\sum_{i=1}^{N_{t}}\exp(r_{it})\right)$
in equilibrium.

\begin{assumption} \label{A-demand} (a) The observed data are generated
from the equilibrium corresponding to the baseline aggregate state
$\left(\mathbf{y}_{t},\mathbf{z}_{t}^{d},\mathbf{u}_{t}\right)$ under
the HSA demand system. (b) The product market is characterized by
monopolistic competition (without free entry); specifically, each
firm takes the quantity index $q_{t}\left(\mathbf{y}_{t},\mathbf{z}_{t}^{d},\mathbf{u}_{t}\right)$
and the log of aggregate expenditure $\Phi_{t}$ as given. (c) The
quantity index at the baseline aggregate state is normalized as ${q}_{t}(\mathbf{y}_{t},\mathbf{z}_{t}^{d},\mathbf{u}_{t})=0$.
(d) $\varphi_{t}(\cdot)$ is known over its support. \end{assumption}

Assumption \ref{A-demand}(a) explicitly specifies the baseline aggregate
state $\left(\mathbf{y}_{t},\mathbf{z}_{t}^{d},\mathbf{u}_{t}\right)$
under which the \textit{reduced-form} functions are defined. Assumption
\ref{A-demand}(b) imposes the standard framework of monopolistic
competition \citep[cf.][]{klette1996jae,de2011product}, implying
that $N_{t}$ is sufficiently large such that single-firm strategic
decisions have a negligible impact on aggregate quantities. Assumption
\ref{A-demand}(c) provides the normalization condition required to
uniquely determine the \textit{structural} function $\mathfrak{s}_{t}^{*}(\cdot)$,
as discussed below. Assumption \ref{A-demand}(d) holds when the conditions
of Proposition \ref{P-CRS} hold.

Treating $q_{t}(\mathbf{y}_{t},\mathbf{z}_{t}^{d},\mathbf{u}_{t})$
and $\Phi_{t}$ as fixed, we define the \textit{reduced-form} revenue
function (\ref{eq:reduced-rev}) from the structural budget-share
function: 
\begin{equation}
r_{it}=\Phi_{t}+\mathfrak{s}_{t}^{*}\left(y_{it}-q_{t}\left(\mathbf{y}_{t},\mathbf{z}_{t}^{d},\mathbf{u}_{t}\right),z_{it}^{d},u_{it}\right)=:\varphi_{t}\left(y_{it},z_{it}^{d},u_{it}\right).\label{eq:reduced-rev}
\end{equation}
Correspondingly, the \textit{reduced-form} inverse demand function
is related to $\mathfrak{s}_{t}^{*}\left(\cdot\right)$ by: 
\[
p_{it}=\Phi_{t}+\mathfrak{s}_{t}^{*}\left(y_{it}-q_{t}\left(\mathbf{y}_{t},\mathbf{z}_{t}^{d},\mathbf{u}_{t}\right),z_{it}^{d},u_{it}\right)-y_{it}=:\psi_{t}\left(y_{it},z_{it}^{d},u_{it}\right).
\]
In this context, the function 
\begin{equation}
\mathfrak{s}_{t}\left(y_{it},z_{it}^{d},u_{it}\right):=\varphi_{t}\left(y_{it},z_{it}^{d},u_{it}\right)-\Phi_{t}\label{eq:reduced-bs}
\end{equation}
represents the \textit{reduced-form} budget-share function, which
is known conditional on $\Phi_{t}$ and $\varphi_{t}\left(\cdot\right)$.

Conditional on the aggregate state $(\mathbf{y}_{t},\mathbf{z}_{t}^{d},\mathbf{u}_{t})$,
the \textit{reduced-form} revenue $\varphi_{t}(\cdot)$ and budget-share
function $\mathfrak{s}_{t}(\cdot)$ are identified up to location
normalization via Proposition \ref{P-CRS}. When evaluating these
reduced-form functions (\ref{eq:reduced-rev}) and (\ref{eq:reduced-bs})
across different values of $(y_{it},z_{it}^{d},u_{it})$, the quantity
index $q_{t}(\mathbf{y}_{t},\mathbf{z}_{t}^{d},\mathbf{u}_{t})$ is
held constant at the baseline state $(\mathbf{y}_{t},\mathbf{z}_{t}^{d},\mathbf{u}_{t})$.
Consequently, these reduced-form functions are not invariant to changes
in the aggregate state. As indicated by (\ref{eq:reduced-rev}), $\varphi_{t}(\cdot)$
shifts in response to changes in $\mathbf{y}_{t}$ because the quantity
index $q_{t}(\mathbf{y}_{t},\mathbf{z}_{t}^{d},\mathbf{u}_{t})$ depends
on the consumption vector $\mathbf{y}_{t}$.

Nonetheless, we may identify the structural function $\mathfrak{s}_{t}^{*}\left(\cdot\right)$
from the reduced-form function $\mathfrak{s}_{t}\left(\cdot\right)$
in (\ref{eq:reduced-bs}) under the equilibrium constraint (\ref{eq:market_share_constraint})
as follows. Evaluating (\ref{eq:structural-bs}) and (\ref{eq:reduced-bs})
at the fixed baseline state $\left(\mathbf{y}_{t},\mathbf{z}_{t}^{d},\mathbf{u}_{t}\right)$,
the \textit{reduced-form} and \textit{structural} budget share functions
are related as: 
\begin{equation}
\mathfrak{s}_{t}\left(y_{it},z_{it}^{d},u_{it}\right)=\mathfrak{s}_{t}^{*}\left(y_{it}-{q}_{t}(\mathbf{y}_{t},{\mathbf{z}}_{t}^{d},\mathbf{u}_{t}),z_{it}^{d},u_{it}\right).\label{eq:comparison}
\end{equation}

Let $\Delta{q}_{t}(\tilde{\mathbf{y}}_{t},\tilde{\mathbf{z}}_{t}^{d},\tilde{\mathbf{u}}_{t}):={q}_{t}(\tilde{\mathbf{y}}_{t},\tilde{\mathbf{z}}_{t}^{d},\tilde{\mathbf{u}}_{t})-{q}_{t}(\mathbf{y}_{t},\mathbf{z}_{t}^{d},\mathbf{u}_{t})$
denote the change in the quantity index induced by a shift from the
baseline state $(\mathbf{y}_{t},\mathbf{z}_{t}^{d},\mathbf{u}_{t})$
to a new state $(\tilde{\mathbf{y}}_{t},\tilde{\mathbf{z}}_{t}^{d},\tilde{\mathbf{u}}_{t})$.
Evaluating (\ref{eq:comparison}) at the argument $(y_{it},z_{it}^{d},u_{it})=(\tilde{y}_{it}-\Delta{q}_{t}(\tilde{\mathbf{y}}_{t},\tilde{\mathbf{z}}_{t}^{d},\tilde{\mathbf{u}}_{t}),\tilde{z}_{it}^{d},\tilde{u}_{it})$
while keeping ${q}_{t}(\mathbf{y}_{t},{\mathbf{z}}_{t}^{d},\mathbf{u}_{t})$
as constant, we have: 
\begin{equation}
\mathfrak{s}_{t}\left(\tilde{y}_{it}-\Delta{q}_{t}(\tilde{\mathbf{y}}_{t},\tilde{{\mathbf{z}}}_{t}^{d},\tilde{\mathbf{u}}_{t}),\tilde{z}_{it}^{d},\tilde{u}_{it}\right)=\mathfrak{s}_{t}^{*}\left(\tilde{y}_{it}-{q}_{t}(\tilde{\mathbf{y}}_{t},\tilde{{\mathbf{z}}}_{t}^{d},\tilde{\mathbf{u}}_{t}),\tilde{z}_{it}^{d},\tilde{u}_{it}\right),\label{eq: s-id}
\end{equation}
for any alternative state $(\tilde{\mathbf{y}}_{t},\tilde{{\mathbf{z}}}_{t}^{d},\tilde{\mathbf{u}}_{t})$.
By substituting (\ref{eq: s-id}) into the constraint (\ref{eq:market_share_constraint}),
we can uniquely identify $\Delta{q}_{t}(\tilde{\mathbf{y}}_{t},\tilde{\mathbf{z}}_{t}^{d},\tilde{\mathbf{u}}_{t})$
by solving: 
\begin{align}
1 & =\sum_{i=1}^{N_{t}}\exp\left(\mathfrak{s}_{t}\left(\tilde{y}_{it}-\Delta{q}_{t}(\tilde{\mathbf{y}}_{t},\tilde{\mathbf{z}}_{t}^{d},\tilde{\mathbf{u}}_{t}),\tilde{z}_{it}^{d},\tilde{u}_{it}\right)\right).\label{eq:market_share2}
\end{align}
Given the identified quantity index change $\Delta q_{t}(\cdot)$,
we are able to identify the \textit{structural} budget share function
$\mathfrak{s}_{t}^{*}\left(\cdot\right)$ from the \textit{reduced-form}
function $\mathfrak{s}_{t}\left(\cdot\right)$ via (\ref{eq: s-id}),
up to the unknown value of the baseline index ${q}_{t}(\mathbf{y}_{t},{\mathbf{z}}_{t}^{d},\mathbf{u}_{t})$.

The preceding discussion demonstrates that $\mathfrak{s}_{t}^{*}(\cdot)$
cannot be separately identified from the baseline aggregate quantity
index ${q}_{t}(\mathbf{y}_{t},\mathbf{z}_{t}^{d},\mathbf{u}_{t})$.
Specifically, given the identified reduced-form function $\mathfrak{s}_{t}(\cdot)$
and the change in the index $\Delta{q}_{t}(\cdot)$, there exist multiple
pairs $\{\mathfrak{s}_{t}^{*}(\cdot),{q}_{t}(\mathbf{y}_{t},\mathbf{z}_{t}^{d},\mathbf{u}_{t})\}$
that satisfy (\ref{eq: s-id}). For this reason, we impose Assumption
\ref{A-demand}(c) to uniquely determine $\mathfrak{s}_{t}^{*}(\cdot)$
but the fundamental properties of the preference structure to be characterized
in Proposition \ref{P-demand} are invariant to this specific normalization.

Applying the result of \citet{matsuyama2017beyond} (Proposition 1
and Remark 3), the following proposition establishes that the HSA
demand system constructed above can be derived from a unique representative
consumer preference, and that it is possible to identify an associated
utility function. Appendix \ref{subsec:Demand} supplies the proof. 
\begin{prop}
\label{P-demand} Suppose that the conditions of Proposition \ref{P-CRS}
and Assumption \ref{A-demand} hold. Then: (a) There exists a unique
monotone, convex, and homothetic rational preference $\succsim$ over
$\mathcal{Y}^{N_{t}}$ that generates an HSA demand system $\left\{ \mathfrak{s}_{t}^{*}(\cdot),{q}_{t}(\cdot)\right\} $.
(b) This preference $\succsim$ is represented by a homothetic utility
function defined by 
\[
\ln U_{t}(\tilde{\mathbf{y}}_{t},\tilde{{\mathbf{z}}}_{t}^{d},\tilde{\mathbf{u}}_{t})=\Delta{q}_{t}(\tilde{\mathbf{y}}_{t},\tilde{{\mathbf{z}}}_{t}^{d},\tilde{\mathbf{u}}_{t})+\sum_{i=1}^{N_{t}}\int_{b_{i}}^{\tilde{y}_{it}-\Delta{q}_{t}(\tilde{\mathbf{y}}_{t},\tilde{{\mathbf{z}}}_{t}^{d},\tilde{\mathbf{u}}_{t})}\exp\left(\mathfrak{s}_{t}\left(\zeta,\tilde{z}_{it}^{d},\tilde{u}_{it}\right)\right)d\zeta
\]
for some constant vector $\mathbf{b}=(b_{1},...,b_{N_{t}})$, which
is identified from $\{\mathfrak{s}_{t}\left(\cdot\right),\Delta q_{t}(\cdot)\}$.
(c) The identified demand system and preference $\succsim$ do not
depend on the location normalization of $\varphi_{t}^{-1}(r_{t},z_{t}^{d},u_{t})$. 
\end{prop}

\subsubsection{Counterfactual analysis}

We conduct a short-run partial equilibrium counterfactual analysis
by maintaining the following assumptions. \begin{assumption}\label{A-short}
(a) The pre-determined factor inputs $\left(k_{it},l_{it}\right)$,
factor prices, and exogenous variables $(z_{it}^{d},z_{it}^{s},u_{it},\omega_{it},p_{t}^{m})$
are fixed before and after counterfactual interventions. (b) The log
of aggregate expenditure $\Phi_{t}$ is fixed before and after counterfactual
interventions. \end{assumption} Under Assumption \ref{A-short}(a),
firms respond to counterfactual interventions by adjusting their material
inputs, outputs, and prices in the short run. Assumption \ref{A-short}(b)
implies that the aggregate income of the representative consumer remains
invariant to the counterfactual interventions. This is a reasonable
approximation when analyzing a specific industry that is small relative
to the aggregate economy. For large-scale interventions, however,
this assumption can be relaxed to allow aggregate income to adjust
endogenously in response to changes in aggregate profits induced by
the intervention.

\paragraph{Monopolistic Competition Equilibrium}

Using the identified HSA demand system $\{\mathfrak{s}_{t}(\cdot),\Delta q_{t}(\cdot)\}$,
we calculate a monopolistic competition equilibrium (MCE) as follows.
Define the inverse production function $m_{it}=\chi_{it}\left(y_{it}\right)$
such that $y_{it}=f_{t}\left(\chi_{it}\left(y_{it}\right),k_{it},l_{it},z_{it}^{s}\right)+\omega_{it}$
for given $(k_{it},l_{it},\omega_{it},z_{it}^{s})$; namely, $\chi_{it}(y_{it}):=f_{t}^{-1}\left(y_{it}-\omega_{it},k_{it},l_{it},z_{it}^{s}\right)$.

Equilibrium outputs and the quantity index $(\mathbf{y}_{t}^{m},\Delta q_{t}^{m})$
in an MCE are obtained from the first order condition for profit maximization
in (\ref{eq:profit_maximization}), equivalently written in view of
(\ref{eq:reduced-bs})-(\ref{eq: s-id}), and the market share condition
(\ref{eq:market_share2}): 
\begin{align}
\underbrace{\exp\left(\mathfrak{s}_{t}\left(y_{it}^{m}-\Delta q_{t}^{m},z_{it}^{d},u_{it}\right)+\Phi_{t}-y_{it}^{m}\right)\frac{\partial\mathfrak{s}_{t}\left(y_{it}^{m}-\Delta q_{t}^{m},z_{it}^{d},u_{it}\right)}{\partial y_{it}}}_{\text{Marginal Revenue}} & =\underbrace{\frac{\exp\left(p_{t}^{m}+\chi_{it}\left(y_{it}^{m}\right)\right)}{\exp\left(y_{it}^{m}\right)}\frac{\partial\chi_{it}(y_{it}^{m})}{\partial y_{it}}}_{\text{Marginal Cost}},\nonumber \\
\sum_{i=1}^{N_{t}}\exp\left(\mathfrak{s}_{t}\left(y_{it}^{m}-\Delta q_{t}^{m},z_{it}^{d},u_{it}\right)\right) & =1,\label{eq:MCE}
\end{align}
for $i=1,...,N_{t}$, where $\frac{\partial\chi_{it}(y_{it}^{m})}{\partial y_{it}}=\left(\frac{\partial f_{t}}{\partial m}\right)^{-1}$.
The above system can be extended to incorporate policies such as taxes
and subsidies to investigate their effects.

\paragraph{Welfare Costs of Firm's Market Power}

In the empirical section below, we quantify the deadweight loss attributable
to firm's market power by considering the transition to a counterfactual
Marginal Cost Pricing Equilibrium (MCPE), in which firms set their
prices equal to their marginal costs. Specifically, equilibrium outputs
and quantity index $(\mathbf{y}_{t}^{c},\Delta q_{t}^{c})$ are obtained
by solving: 
\begin{align}
\underbrace{\exp\left(\mathfrak{s}_{t}\left(y_{it}^{c}-\Delta q_{t}^{c},z_{it}^{d},u_{it}\right)+\Phi_{t}-y_{it}^{c}\right)}_{\text{Price}} & =\underbrace{\frac{\exp\left(p_{t}^{m}+\chi_{it}\left(y_{it}^{c}\right)\right)}{\exp\left(y_{it}^{c}\right)}\frac{\partial\chi_{it}(y_{it}^{c})}{\partial y_{it}}}_{\text{Marginal Cost}},\nonumber \\
\sum_{i=1}^{N_{t}}\exp\left(\mathfrak{s}_{t}\left(y_{it}^{c}-\Delta q_{t}^{c},z_{it}^{d},u_{it}\right)\right) & =1,\label{eq:MCPE}
\end{align}

for $i=1,...,N_{t}$. 

The consumer welfare cost of firm's market power can be calculated
as the utility change: 
\begin{align*}
\ln U_{t}(\mathbf{y}_{t}^{c},\mathbf{z}_{t}^{d},\mathbf{u}_{t})-\ln U_{t}(\mathbf{y}_{t}^{m},\mathbf{z}_{t}^{d},\mathbf{u}_{t}) & =\Delta q_{t}^{c}-\Delta q_{t}^{m}+\sum_{i=1}^{N_{t}}\int_{y_{it}^{m}-\Delta q_{t}^{m}}^{y_{it}^{c}-\Delta q_{t}^{c}}\exp\left(\mathfrak{s}_{t}\left(\zeta,z_{it}^{d},u_{it}\right)\right)d\zeta.
\end{align*}

An alternative welfare measure is the Compensating Variation (CV),
which is constructed in monetary terms as follows. For a given counterfactual
log income $\Phi_{t}^{c}$, define the counterfactual output and quantity
index $(\mathbf{y}_{t}^{c*}(\Phi_{t}^{c}),\Delta q^{c*}(\Phi_{t}^{c}))$
that solve the price-taking condition: 
\[
\underbrace{\exp\left(\mathfrak{s}_{t}\left(y_{it}^{c*}-\Delta q_{t}^{c*},z_{it}^{d},u_{it}\right)+\Phi_{t}^{c}-y_{it}^{c*}\right)}_{\text{Price}}=\underbrace{\frac{\exp\left(p_{t}^{m}+\chi_{it}\left(y_{it}^{c*}\right)\right)}{\exp(y_{it}^{c*})}\frac{\partial\chi_{it}(y_{it}^{c*})}{\partial y_{it}}}_{\text{Marginal Cost}}
\]
for $i=1,...,N_{t}$ with $\sum_{i=1}^{N_{t}}\exp\left(\mathfrak{s}_{t}\left(y_{it}^{c*}-\Delta q_{t}^{c*},z_{it}^{d},u_{it}\right)\right)=1$.
We then find the counterfactual income $\Phi_{t}^{c*}$ that achieves
the same utility as in the benchmark MCE: 
\begin{equation}
\ln U_{t}(\mathbf{y}_{t}^{c*}(\Phi_{t}^{c*}),\mathbf{z}_{t}^{d},\mathbf{u}_{t})-\ln U_{t}(\mathbf{y}_{t}^{m},\mathbf{z}_{t}^{d},\mathbf{u}_{t})=0.\label{eq:CV}
\end{equation}
The compensating variation is defined as $CV_{t}:=\exp\left(\Phi_{t}^{c*}\right)-\exp\left(\Phi_{t}\right)$.
It measures the change in consumer welfare when moving from an MCPE
to an MCE and quantifies the consumer's loss due to firms' market
power.  Since consumers require less income under competitive pricing to
achieve the same utility, $CV_{t}<0$: the magnitude $|CV_{t}|$ is the
monetary gain from eliminating market power, and the consumer's welfare
improvement from transitioning to an MCPE is $-CV_{t}>0$.

To evaluate the overall welfare change from an MCE to an MCPE, the
consumer gain, measured by $-CV_{t}$, can be compared with firms'
profit loss. Under Assumption \ref{A-short}(b), aggregate revenue
is fixed. Thus, the change in total profits is solely the negative
change in total material costs: 
\begin{equation}
\Delta\Pi_{t}:=\Pi_{t}^{c}-\Pi_{t}^{m}=-\sum_{i=1}^{N_{t}}\exp(p_{t}^{m})\left\{ \exp\left(\chi_{it}\left(y_{it}^{c}\right)\right)-\exp\left(\chi_{it}\left(y_{it}^{m}\right)\right)\right\} .\label{eq:profit change}
\end{equation}
Note that if $\Pi_{t}^{c}<0$, this counterfactual implies firms operate
at a loss, representing a short-run equilibrium before firm exit occurs.
The overall welfare change associated with a transition from an MCE
to an MCPE is therefore given by $\Delta\Pi_{t}-CV_{t}$. 

\section{Semiparametric Estimator}

\label{sec:estimator}


 The identification results in Section~\ref{sec:Identification} are nonparametric: revenue data alone suffice to recover the production function, demand system, and welfare-relevant objects without functional-form restrictions. While a fully nonparametric sieve estimator could in principle be constructed from the identified equations, it would demand sample sizes far exceeding typical manufacturing datasets. We therefore impose Cobb-Douglas production and AR(1) productivity dynamics, following the standard practice of pairing nonparametric identification with parametric specifications \citep[see, e.g.,][]{olley1996dynamics,levinsohn2003estimating,ackerberg2015identification,gandhi2020identification}.

We develop a semiparametric estimator that is applicable for the panel
data with $T\geq4$. We assume the Cobb-Douglas production function:
\begin{align}
f_{t}(m_{it},k_{it},l_{it}) & =\theta_{m}m_{it}+\theta_{k}k_{it}+\theta_{l}l_{it},\label{eq:CobbDouglas}
\end{align}
and TFP follows an AR(1) process: 
\begin{align}
\omega_{it} & =\rho\omega_{it-1}+\eta_{it}.\label{eq:ar-1}
\end{align}

We introduce assumptions on the production function to normalize scale
and location parameters. First, (\ref{eq:CobbDouglas}) normalizes
one location parameter by $f_{t}(0,0,0)=0$. As another location normalization,
we assume $E\left[\omega_{it}\right]=0$, which naturally arises from
a stationary AR1 process (\ref{eq:ar-1}). Finally, we assume the
constant returns to scale, $\theta_{m}+\theta_{k}+\theta_{l}=1$,
which normalizes the scale parameter. 

The control function becomes separable under the Cobb-Douglas assumption:
\begin{equation}
\mathbb{M}_{t}^{-1}(m_{it},k_{it},l_{it},u_{it})=\lambda_{t}(m_{it},u_{it})-\theta_{k}k_{it}-\theta_{l}l_{it}.\label{eq:control}
\end{equation}
Substituting (\ref{eq:control}) into the revenue function, we obtain
\begin{align}
\varphi_{t}\left(f_{t}(m_{it},k_{it},l_{it})+\omega_{it},u_{it}\right)=\varphi_{t}\left(\theta_{m}m_{it}+\lambda_{t}(m_{it},u_{it}),u_{it}\right)=\phi_{t}(m_{it},u_{it}),\label{eq:ivqr}
\end{align}
where $\phi_{t}$ depends only on $(m_{t},u_{it})$ and increases
in $m_{t}$ and $u_{t}$. In Appendix \ref{subsec:Derivation-of-the},
we derive a separable control function similar to (\ref{eq:control}) in a
more general setting and show that (\ref{eq:ivqr}) holds when the production function is separable with
respect to $m_{it}$.  \\

\noindent Substituting (\ref{eq:control}) into (\ref{eq:ar-1}),
we obtain the transformation model as 
\begin{align}
\lambda_{t}(m_{it},u_{it}) & =\theta_{k}(k_{it}-\rho k_{it-1})+\theta_{l}(l_{it}-\rho l_{it-1})+\rho\lambda_{t-1}(m_{it-1},u_{it-1})+\eta_{it}.\label{eq:transformation}
\end{align}

\paragraph{Step 1: Estimation of the Quantile of Demand Shocks}

The first step estimates $\phi_{t}\left(m_{it},u_{it}\right)$ and
$u_{it}$ by IV quantile regression. A traditional approach to IV
quantile regression estimates $\phi_{t}\left(\cdot,u\right)$ from
the moment condition (\ref{eq: IVQR moment condition}) for a fixed
quantile point $u$. This approach often yields a non-monotonic and
non-smooth function in $u$, which is problematic for our identification
using uniquely identified $u_{it}$ and derivatives of $\phi_{t}$.
To overcome this, we use the smoothed GMM quantile regression of \citet*{firpo2022gmm}.
Their approach stacks moment conditions over all quantile points so
we can estimate the smooth sieve function and impose $\partial\phi_{t}/\partial u_{it}>0$.
For the approximation of $\phi_{t}(m_{it},u_{it})$, we employ the
basis $B_{\phi}(m_{it},\tau)$ that consists of a constant term, a
B-spline basis of degree 3 with 2 interior knots in $m_{it}$, a cubic
polynomial in $u_{it}$, and interactions of the B-spline in $m_{it}$
with $u_{it}$ and $u_{it}^{2}$. \citet{firpo2022gmm} also replace
the indicator in (\ref{eq: IVQR moment condition}) with a smooth
kernel CDF to ease computation.

We partition $[0,1]$ into $L=100$ equal parts and let $\mathbb{T}=\{0.01,\ldots,0.99\}$.
The moment condition is 
\begin{equation}
E\left[\left(\mathcal{K}_{1}\left(\frac{B_{\phi}(m_{it},\tau)^{T}\alpha_{t}-r_{it}}{b_{n1}}\right)-\tau\right)B_{IV}(m_{it-2})\right]=0\quad\text{for }\tau\in\mathbb{T},\label{eq:GMM moment}
\end{equation}
where $\mathcal{K}_{1}(\cdot)$ is a smooth kernel CDF with bandwidth
$b_{n1}$ and $B_{IV}(m_{it-2})$ is the B-spline basis of degree
$S_{1}=3$ with $K_{1}=2$ interior knots in $m_{it-2}$ as instruments.
Following \citet{firpo2022gmm}, we use the rule-of-thumb bandwidth
and the kernel CDF of \citet{horowitz1998bootstrap}: 
\[
\mathcal{K}_{1}(s):=\left[\frac{1}{2}+\frac{105}{64}\left(s-\frac{5}{3}s^{3}+\frac{7}{5}s^{5}-\frac{3}{7}s^{7}\right)\right]1\{s\in[-1,1]\}+1\{s>1\}.
\]
The number of moment conditions (\ref{eq:GMM moment}) is the number
of IVs times the number of quantile $(S_{1}+K_{1}+1)\times(L-1)$.
As $L$ is usually a large number, the moment condition (\ref{eq:GMM moment})
typically overidentifies $\alpha_{t}$ so that we use GMM. \citet{firpo2022gmm}
derive a expression of the optimal GMM weight matrix and showed it
does not depend on the parameter $\alpha_{t}$ so that its estimation
completes in one step. Monotonicity in $m_{it}$ and $u_{it}$ is
imposed via linear constraints on the derivatives of the basis functions.
The demand shocks $\hat{u}_{it}$ are then estimated by numerically
inverting $\hat{\phi}_{t}(m_{it},\hat{u}_{it})=r_{it}$. The same
procedure is applied to $t-1$ to estimate $\hat{u}_{it-1}$.

\paragraph{Step 2: Estimation of the control function}

The second step estimates the transformation model (\ref{eq:transformation}).
We use the Profile Likelihood (PL) estimator developed by \citet*{linton2008estimation}.
From (\ref{eq:dG2}) for $q_{it}=m_{it}$ and (\ref{eq:control}),
the conditional density of $m_{it}$ given $v_{it}$ is written as
\begin{align*}
g_{m_{t}|v_{t}}(m_{it}|v_{it}) & =g_{\eta_{t}}\left(\eta_{it}\right)\frac{\partial\lambda_{t}(m_{it},u_{it})}{\partial m_{it}}.
\end{align*}
To approximate $\lambda_{t}(m_{it},u_{it})$, we use the basis $B_{\lambda}(m_{it},u_{it})$
that is the Kronecker product of B-spline bases of degree 3 with 1
interior knot in $m_{it}$ and $u_{it}$. We do not assume a parametric
distribution on $\eta_{it}$. Let $\partial_{m}B_{\lambda}(m_{it},u_{it})$
be the derivatives of the B-spline bases $B_{\lambda}(m_{it},u_{it})$
so that $\partial_{m}B_{\lambda}(m_{it},u_{it})^{T}c_t$ approximates
$\frac{\partial\lambda_{t}(m_{it},u_{it})}{\partial m_{it}}$. Thus,
the log-likelihood function is written as 
\[
\sum_{i=1}^{n}\left\{ \ln g_{m_{t}|v_{t}}\left(m_{it}|v_{it}\right)\right\} =\sum_{i=1}^{n}\left\{ \ln g_{\eta_{t}}\left(\eta_{it}\right)+\ln\partial_{m}B_{\lambda}(m_{it},u_{it})^{T}c_t\right\} .
\]
where ${g}_{\eta_{t}}(\eta)$ is the corresponding (Gaussian) kernel
density  with the bandwidth chosen by Silverman's Rule. We obtain estimates of $\eta_{it}$ as follows. {For given
$(c_t,\rho)$, define $\lambda_{it}(c_t):=B_{\lambda}(m_{it},u_{it})^{T}c_t$, $\tilde k_{it}(\rho):=k_{it}-\rho k_{it-1}$, $\tilde l_{it}(\rho):=l_{it}-\rho l_{it-1}$, and $\tilde B_{\lambda}(m_{it-1},u_{it-1},\rho):=\rho B_{\lambda}(m_{it-1},u_{it-1})$, and  (\ref{eq:transformation}) implies that
\begin{align}
\lambda_{it}(c_t) &  = \theta_{k}\tilde k_{it}(\rho)+\theta_{l}\tilde l_{it}(\rho)+\tilde B_{\lambda}(m_{it-1},u_{it-1},\rho)^{T} c_{t-1}+\eta_{it}. \label{eq:TRmodel}
\end{align}
Therefore, for each $(c_t,\rho)$, let $\eta_{it}(c_t,\rho)$ be the residual from 
regressing $\lambda_{it}(c_t)$ on $\tilde k_{it}(\rho)$, $\tilde l_{it}(\rho)$, and $\tilde B_{\lambda}(m_{it-1},u_{it-1},\rho)$. 
Then, the PL estimator $(\tilde{c}_{t},\tilde\rho)$ is defined as 
\begin{align*}
(\tilde{c}_{t},\tilde{\rho}) & \in\arg\max_{c_t,\rho}\sum_{i=1}^{n}\left\{ \ln{g}_{\eta,{t}}\left(\eta_{it}(c_t,\rho)\right)+\ln\partial_{m}B_{\lambda}(m_{it},\hat{u}_{it})^{T}c_t\right\} \text{ subject to }\partial_{m}B_{\lambda}(m_{it},\hat{u}_{it})^{T}c_t>0.
\end{align*}}

\paragraph{Step 3: Estimation of production function, markup, TFP, and output}

With estimated $(\tilde{c}_{t},\tilde\rho)$, 
we estimate $(\theta_{k},\theta_{l},c_{t-1})$ by regressing $B_{\lambda}(m_{it},\hat{u}_{it})^{T}\tilde{c}_{t}$ on $\tilde k_{it}(\tilde\rho)$, $\tilde l_{it}(\tilde\rho)$, and $\tilde B_{\lambda}(m_{it-1},u_{it-1},\tilde\rho)$:
\begin{align*}
(\tilde{\theta}_{k},\tilde{\theta}_{l},\tilde c_{t-1})\in\arg\min_{\theta_{k},\theta_{l},c_{t-1}}\sum_{i=1}^{n}\left\{B_{\lambda}(m_{it},\hat{u}_{it})^{T}\tilde{c}_{t}-\theta_{k}\tilde k_{it}(\tilde\rho)-\theta_{l}\tilde l_{it}(\tilde\rho)-\tilde B_{\lambda}(m_{it-1},u_{it-1},\rho)^{T} c_{t-1}\right\} ^{2}.
\end{align*}
From (\ref{eq:markup_identified}) and (\ref{eq:output_elasticities_pro2}),
the material elasticity is identified as 
$\theta_{m}=\frac{\partial\lambda_{t}\left(m_{it},u_{it}\right)}{\partial m_{it}}\left(\frac{s_{it}^{m}}{\partial\phi_{t}\left(m_{it},u_{it}\right)/\partial m_{it}-s_{it}^{m}}\right)$,
where $s_{it}^{m}=\exp(p_{t}^{m}+m_{it})/\exp\left(r_{it}\right).$
Following this equation, we estimate $\theta_{m}$ as follows: 
\begin{align*}
\tilde{\theta}_{m}=\text{median}\left\{ \frac{\left(\partial_{m}B_{\lambda}(m_{it},\hat{u}_{it})^{T}\tilde{c}_{t}\right)s_{it}^{m}}{\partial_{m}B_{\phi}(m_{it},\hat{u}_{it})^{T}\hat{\alpha}_{t}-s_{it}^{m}}\right\} ,
\end{align*}
where $\partial_{m}B_{\phi}(m_{it},\hat{u}_{it})$ is the derivatives
of bases $B_{\phi}(m_{it},\tau)$ so that $\partial_{m}B_{\phi}(m_{it},\hat{u}_{it})^{T}\hat{\alpha}_{t}$
estimates $\partial\phi_{t}\left(m_{it},u_{it}\right)/\partial m_{it}$.
Using the constant returns to scale, we obtain the normalized production
parameters as $\hat{\theta}_{j}=\tilde{\theta}_{j}/\tilde{b}_{t}$
for $j\in\{m,k,l\}$ where $\tilde{b}_{t}=\tilde{\theta}_{m}+\tilde{\theta}_{k}+\tilde{\theta}_{l}$.
Then, we estimate markups as follows: 
\begin{align*}
\hat{\mu}_{it}=\frac{\hat{\theta}_{m}}{s_{it}^{m}}.
\end{align*}
With the mean-zero restriction on $\omega_{it}$, the location parameter
$\hat{a}_{2t}=n^{-1}\sum_{i=1}^{n}\left[\tilde{\lambda}_{it}-\tilde{\theta}_{k}k_{it}-\tilde{\theta}_{l}l_{it}\right]$
is estimated. The estimated TFP, output, and price are 
\begin{align*}
\hat{\omega}_{it} & =\frac{\tilde{\lambda}_{it}-\tilde{\theta}_{k}k_{it}-\tilde{\theta}_{l}l_{it}-\hat{a}_{2t}}{\tilde{b}_{t}},\quad\hat{y}_{it}=\hat{\omega}_{it}+\hat{\theta}_{m}m_{it}+\hat{\theta}_{l}l_{it}+\hat{\theta}_{k}k_{it},\quad\text{and}\quad\hat{p}_{it}=r_{it}-\hat{y}_{it}.
\end{align*}

\paragraph{Step 4: Estimation of parametric CoPaTh-HSA demand system}

Our estimation steps of production function above do not assume
any parametric demand system. Thus, in theory, one can estimate a
fully nonparametric HSA demand system as described in Section \ref{subsec:demand}.
However, in our empirical application, we estimate a parametric HSA
demand system to obtain more stable estimates from a dataset with
a moderate sample size. In particular, we consider a HSA demand system
with the CoPaTh (constant pass-through) demand function with incomplete
pass-through by \citet{matsuyama2020constant}: 
\begin{align}
\mathfrak{s}_{t}^{*}(y_{it}-q_{t}(\mathbf{y}_{t},\epsilon_{t}),\epsilon_{it}):=\delta_{t}-\frac{1}{\beta_{t}}\ln\left(\frac{\exp(-\beta_{t}\left(y_{it}-q_{t}(\mathbf{y}_{t},\mathbf{\epsilon}_{t})\right)+\gamma_{t})+\epsilon_{it}}{1+\epsilon_{it}}\right)\label{eq:revenue1}
\end{align}
where the quantity index $q_{t}(\mathbf{y}_{t},\epsilon_{t})$ is
implicitly defined by the market share constraint (\ref{eq:market_share_constraint})
for a given output vector $\mathbf{y}_{t}=\left(y_{1t},...,y_{Nt}\right)$
and a given demand shock vector $\mathbf{\epsilon}_{t}=\left(\epsilon_{1t},...,\epsilon_{Nt}\right)$.
Appendix \ref{subsec:CoPaTh-HSA-Demand-System} derives the log-variable
version of the CoPaTh demand (\ref{eq:revenue1}) from \citet{matsuyama2020constant}'s
original formulation.


As explained in Section~\ref{subsec:demand}, we estimate the following
reduced-form revenue function instead of the structural form~(\ref{eq:revenue1}):
\begin{equation}
\varphi_{t}\!\left(y_{it},\epsilon_{it}\right)=\Phi_{t}+\delta_{t}-\frac{1}{\beta_{t}}\ln\!\left(\frac{\exp(-\beta_{t}y_{it}+\gamma_{t})+\epsilon_{it}}{1+\epsilon_{it}}\right),\label{eq:revenue-copath}
\end{equation}
where we impose the normalization $q_{t}(\mathbf{y}_{t},\boldsymbol{\epsilon}_{t})=0$,
as stated in Assumption \ref{A-demand}(c), by assuming that the observed
data correspond to the baseline aggregate state. Consequently, the
implied markup is 
\[
\frac{P_{it}}{MC_{it}}=\left(\frac{\partial\varphi_{t}(y_{it},\epsilon_{it})}{\partial y_{it}}\right)^{-1}=1+\epsilon_{it}\exp(\beta_{t}y_{it}-\gamma_{t}).
\]

Taking the limit as $\beta_{t}\to0$ on the right-hand side of~(\ref{eq:revenue-copath})
yields a linear revenue function $\varphi_{t}(y_{it},\epsilon_{it})=\alpha_{t}+\rho_{it}y_{it}$,
where $\alpha_{t}=\Phi_{t}+\delta_{t}$ and $\rho_{it}=\frac{1}{1+\epsilon_{it}}$.
Thus, the CoPaTh demand system in~(\ref{eq:revenue1}) nests the
generalized CES demand with heterogeneous markups defined in~(\ref{eq:demand}).
Under this special case, the markup is independent of $y_{it}$ and
pass-through is complete. Moreover, if $\epsilon_{it}=\epsilon$ is
common across all firms, then the system collapses to the standard
CES demand framework with a homogeneous markup.

With the estimated outputs $\hat{y}_{it}$ and markups $\hat{\mu}_{it}$
from Step 3, the composite nonlinear least square estimator of demand
parameters $(\beta_{t}$, $\gamma_{t}$, $\delta_{t})$ is defined
as, 
\[
{(\hat{\beta}_{t},\hat{\gamma}_{t},\hat{\delta}_{t})'}\in\arg\min_{\beta_{t},\delta_{t},\gamma_{t}}\sum_{i}\left(r_{it}-\left(\Phi_{t}+\delta_{t}-\frac{1}{\beta_{t}}\ln\left(\frac{\exp(-\beta_{t}\hat{y}_{it}+\gamma_{t})+\epsilon_{it}}{1+\epsilon_{it}}\right)\right)\right)^{2}
\]
\[
+\sum_{i}\left(\hat{u}_{it}-\left(\text{quantile}(\epsilon_{it})\right)\right)^{2},
\]
subject to 
\begin{align*}
\epsilon_{it} & =\frac{\hat{\mu}_{it}-1}{\exp(\beta_{t}\hat{y}_{it}-\gamma_{t})}\quad\text{and}\quad1=\sum_{i}\exp\left(\delta_{t}-\frac{1}{\beta_{t}}\ln\left(\frac{\exp(-\beta_{t}\hat{y}_{it}+\gamma_{t})+\epsilon_{it}}{1+\epsilon_{it}}\right)\right),
\end{align*}
where $\text{quantile}(\epsilon_{it})$ is the empirical quantile
of $\epsilon_{it}$ among all firms and $\Phi_{t}=\ln\left(\sum_{i}\exp(r_{it})\right)$
is the industry revenue. The estimation utilizes the theoretical relationship
between markups and demand shocks and the market share restriction.
\footnote{For model-consistency, we estimate the HSA demand system only using
the firms with estimated markups $\hat{\mu}_{it}>1$.}


\section{Simulation}

\label{sec:simulation}

This section presents the finite sample performance of our proposed
estimator, comparing to that of the conventional method, when firms
charge small but heterogeneous markups under the HSA demand system.
We consider a simple data-generating process (DGP) in which firms
set variable markups, calibrated to Chilean manufacturing data. Further
details of the DGP and simulation design are provided in the Appendix
\ref{app:simulation}.

Consider $N$ firms in a market and $t\in\{1,2,...,T\}$ period. Each
firm produces one variety of differentiated goods and faces the HSA-CoPaTh
demand function (\ref{eq:revenue1}). The demand shock $\epsilon_{it}$
follows an MA(1) process: $\epsilon_{it}=0.5\zeta_{it-1}+\zeta_{it}$,
where $\zeta_{it}\sim\text{Unif}[0,0.3]$. The production function
takes the Cobb-Douglas form (\ref{eq:CobbDouglas}) where $\omega_{it}$
follows an AR1 process $\omega_{it}=0.8\omega_{it-1}+\eta_{it},\eta_{it}\sim N\left(0,\left(0.05\right)^{2}\right)$.
Capital $k_{it}$ and labor $l_{it}$ are predetermined and follow
exogenous laws of motion explained in Appendix \ref{subsec:Data-Generating-Process}.

The production function parameters are set to $\left(\theta_{m},\theta_{k},\theta_{l}\right)=\left(0.4,0.3,0.3\right)$.
We choose the structural parameters of the HSA demand system, $(\Phi_{t},\delta_{t},\beta_{t})=(20,-6.5,0.21)$,
and allow parameter $\gamma_{t}$ to vary across simulations to satisfy
the normalization $q_{t}(\mathbf{y}_{t},\boldsymbol{\epsilon}_{t})=0$
in each simulation. Specifically, for each period, we find equilibrium
outputs $\{y_{it}^{m}\}_{i=1}^{N_{t}}$ and $\gamma_{t}$ in an MCE
by solving the first order conditions and the market share condition
analogous to (\ref{eq:MCE}): 
\begin{align*}
\Phi_{t}+\delta_{t}-\beta_{t}y_{it}^{m}+\gamma_{t}+\Xi_{it}-\frac{y_{it}^{m}}{\theta_{m}}+\frac{1}{\beta_{t}}\ln(1+\epsilon_{it})\\
-\left(1+\frac{1}{\beta_{t}}\right)\ln\left(\exp(-\beta_{t}y_{it}^{m}+\gamma_{t})+\epsilon_{it}\right) & =0\text{ for }i=1,..,N_{t}\\
\sum_{i=1}^{N_{t}}\exp\left(\delta_{t}-\frac{1}{\beta_{t}}\ln\left(\frac{\exp(-\beta_{t}y_{it}^{m}+\gamma_{t})+\epsilon_{it}}{1+\epsilon_{it}}\right)\right) & =1.
\end{align*}
where $\Xi_{it}=\ln\theta_{m}+(\theta_{k}k_{it}+\theta_{l}l_{it}+\omega_{it})/\theta_{m}$
and $p_{t}^{m}=0$. Appendix \ref{subsec:Derivations-of-Equilibrium}
show its derivation. We simulate 100 replications of $N=600$ firms
and $T=5$ periods, with the following summary statistics of the resulting
markups:

\begin{table}[H]
\begin{centering}
\begin{tabular}[t]{ccccccc}
\toprule 
 & Min.  & 1st Qu.  & Median  & Mean  & 3rd Qu.  & Max \tabularnewline
\midrule 
Markup  & 1.001  & 1.147  & 1.219  & 1.223  & 1.295  & 1.617 \tabularnewline
\bottomrule
\end{tabular}
\par\end{centering}
\caption{Summary statistics of Markups in simulated data ($t=5$)}
\end{table}


In addition to our proposed estimator, we also consider the widely-used
estimator proposed by \citet{ackerberg2015identification}(ACF) applied
with quantity data and the ACF estimator applied with revenue data.
\citet{gandhi2020identification}(GNR) showed the difficulty of identification
in the DGP that ACF assumed where a firm-level unobserved shock is
a scalar, TFP. The GNR criticism is not applicable for the current
DGP with two unobserved shocks. However, to show our point is different
from the GNR critique, we employ the ACF method with constant returns
to scale (CRS) restriction (i.e., $\theta_{m}+\theta_{k}+\theta_{l}=1$)
that \citet*{flynn2019measuring} proposed to address the GNR criticism.


Figure \ref{fig:ACF-Sim} shows the histograms of 100 estimates of
$(\theta_{m},\theta_{k},\theta_{l})$ from the ACF method with revenue
data and quantity data. While using quantity data yields estimates
that are tightly clustered around the true values, using revenue data
substantially biases the estimation of the production function. The
simulation result confirms the long criticism in the literature against
the ad hoc use of revenue data as output. 


\begin{figure}[H]
\begin{centering}
\caption{Production Function Estimation with ACF on Revenue Data and Quantity
Data }
\label{fig:ACF-Sim} \includegraphics[scale=0.3]{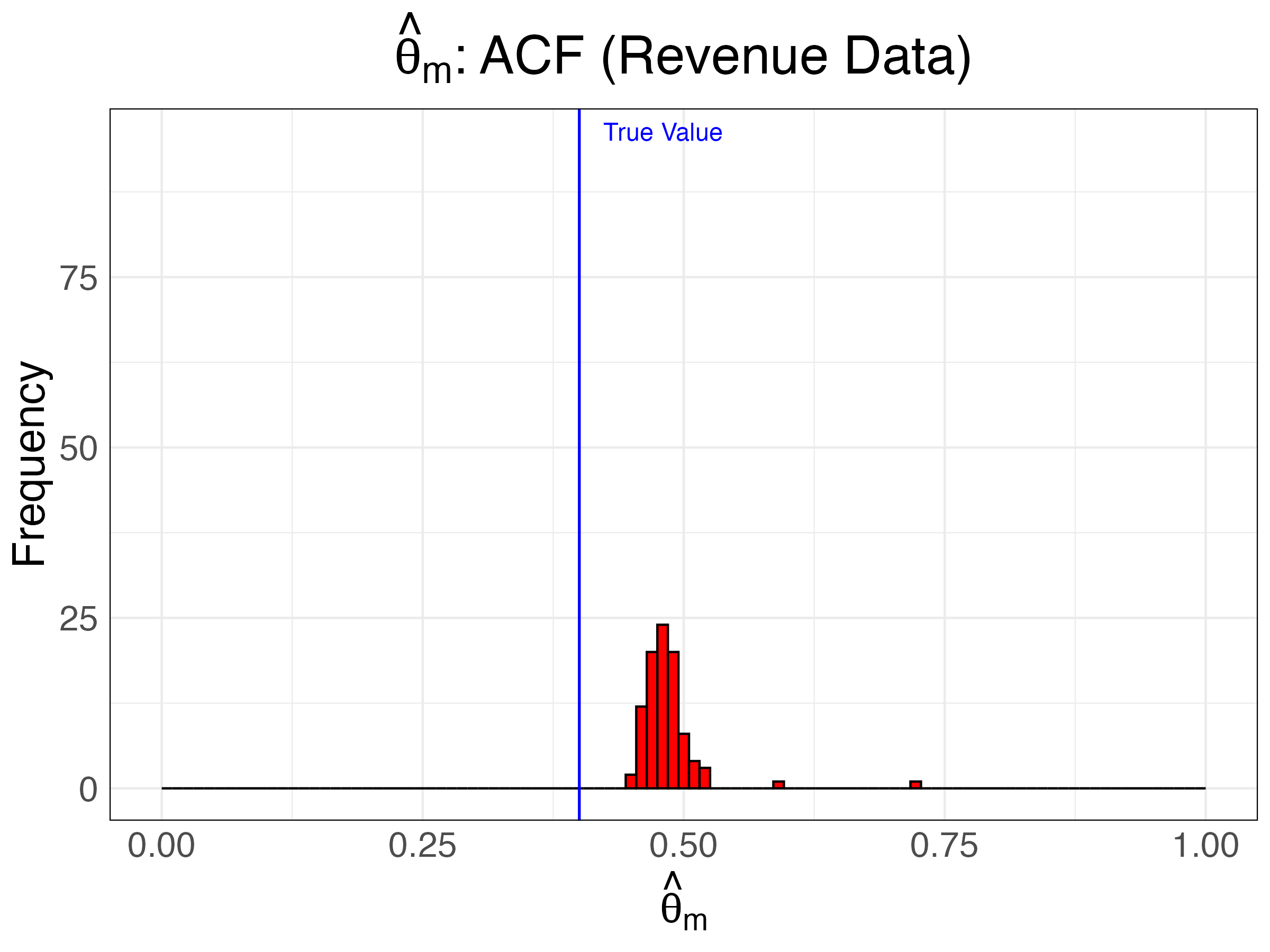}\includegraphics[scale=0.3]{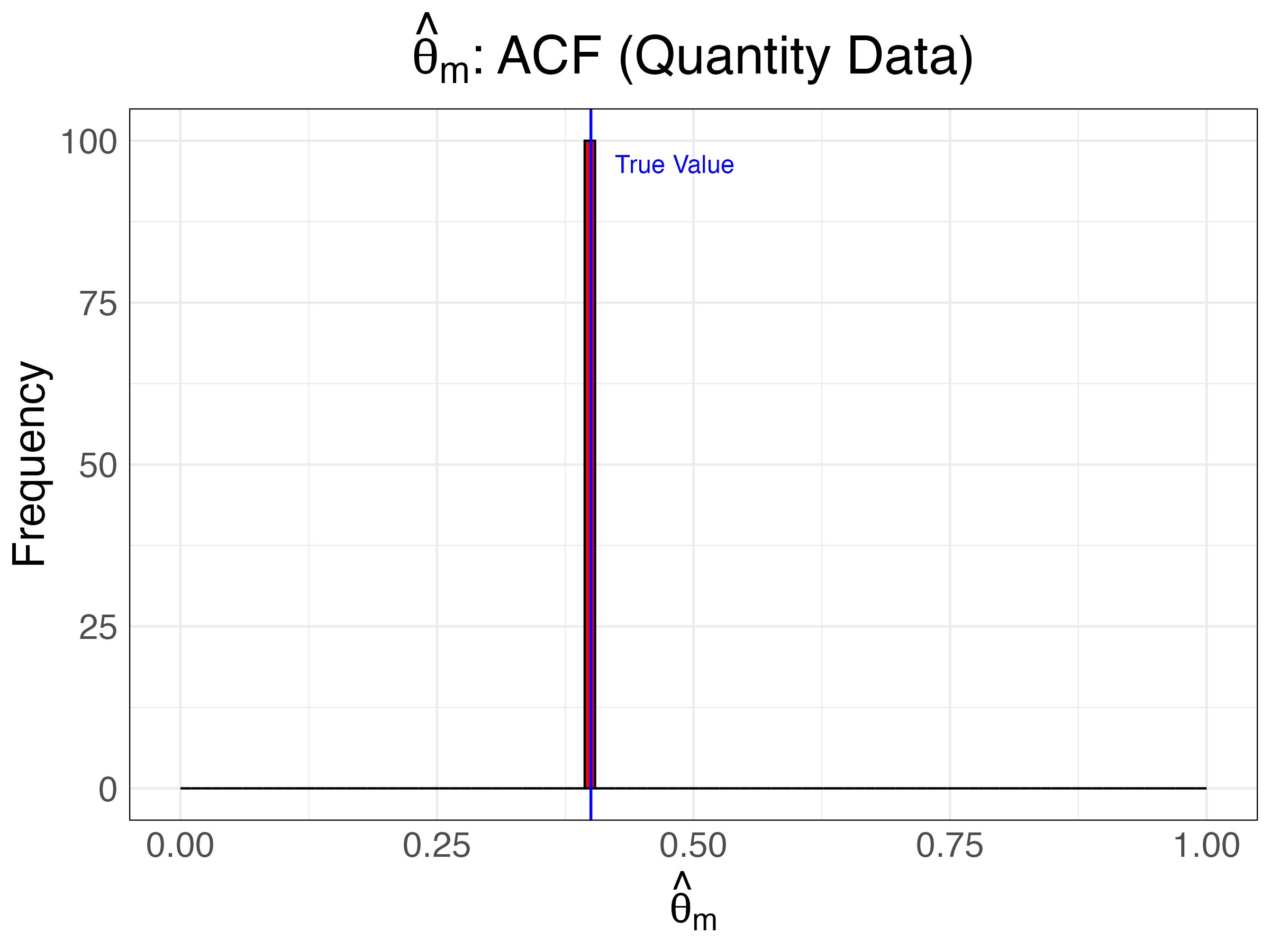} 
\par\end{centering}
\centering{}\includegraphics[scale=0.3]{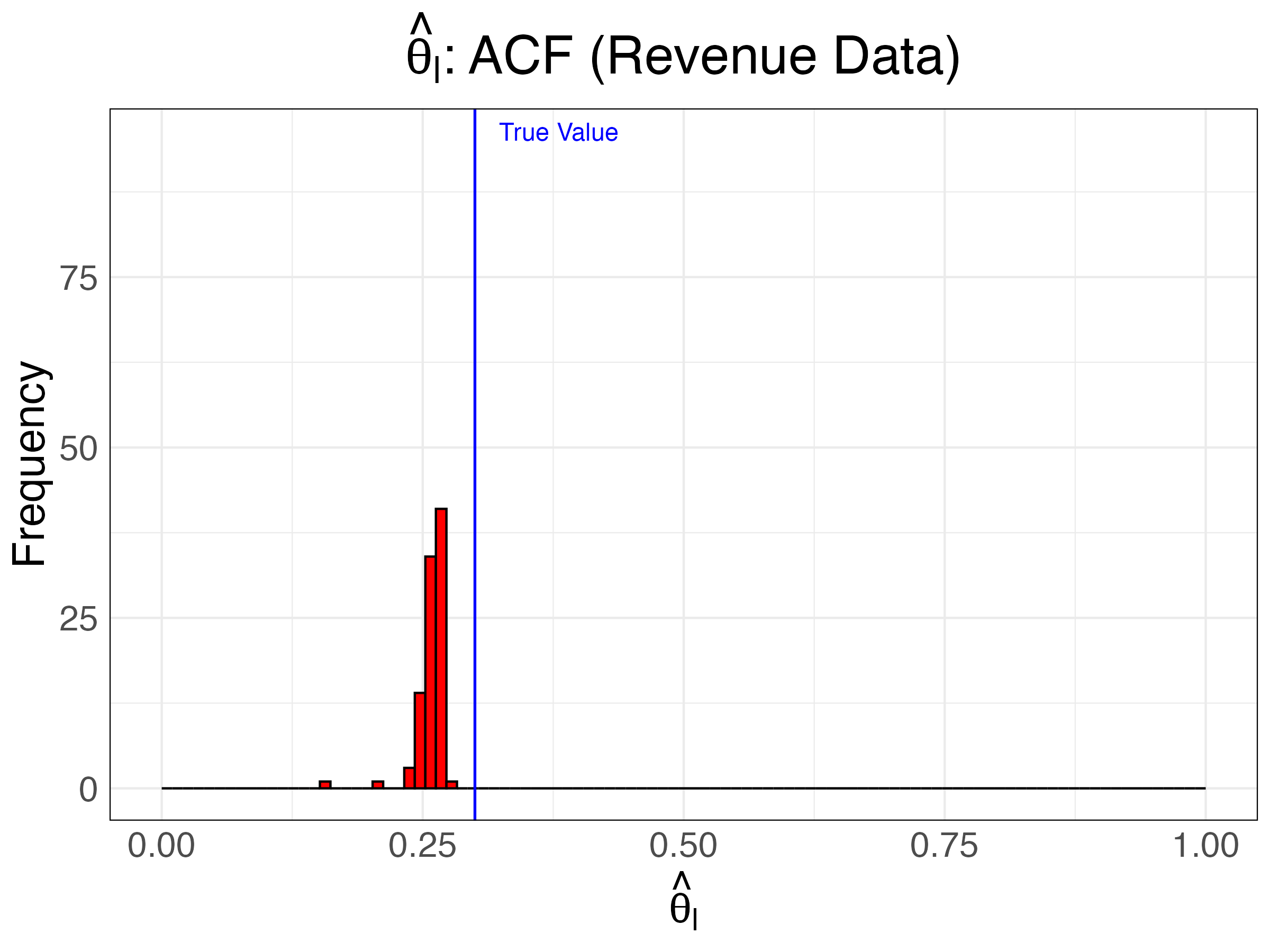}\includegraphics[scale=0.3]{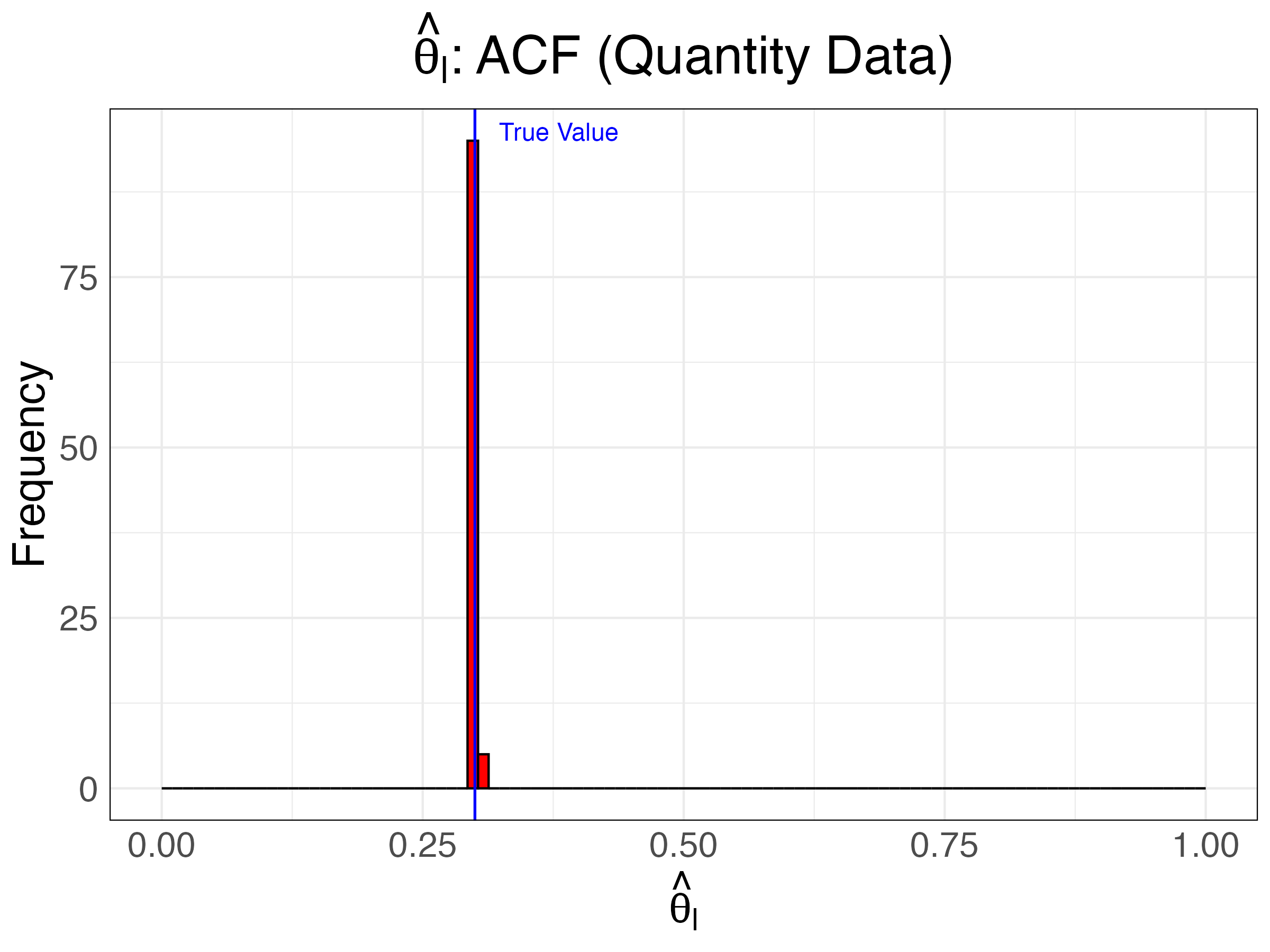}
\centering{}\includegraphics[scale=0.3]{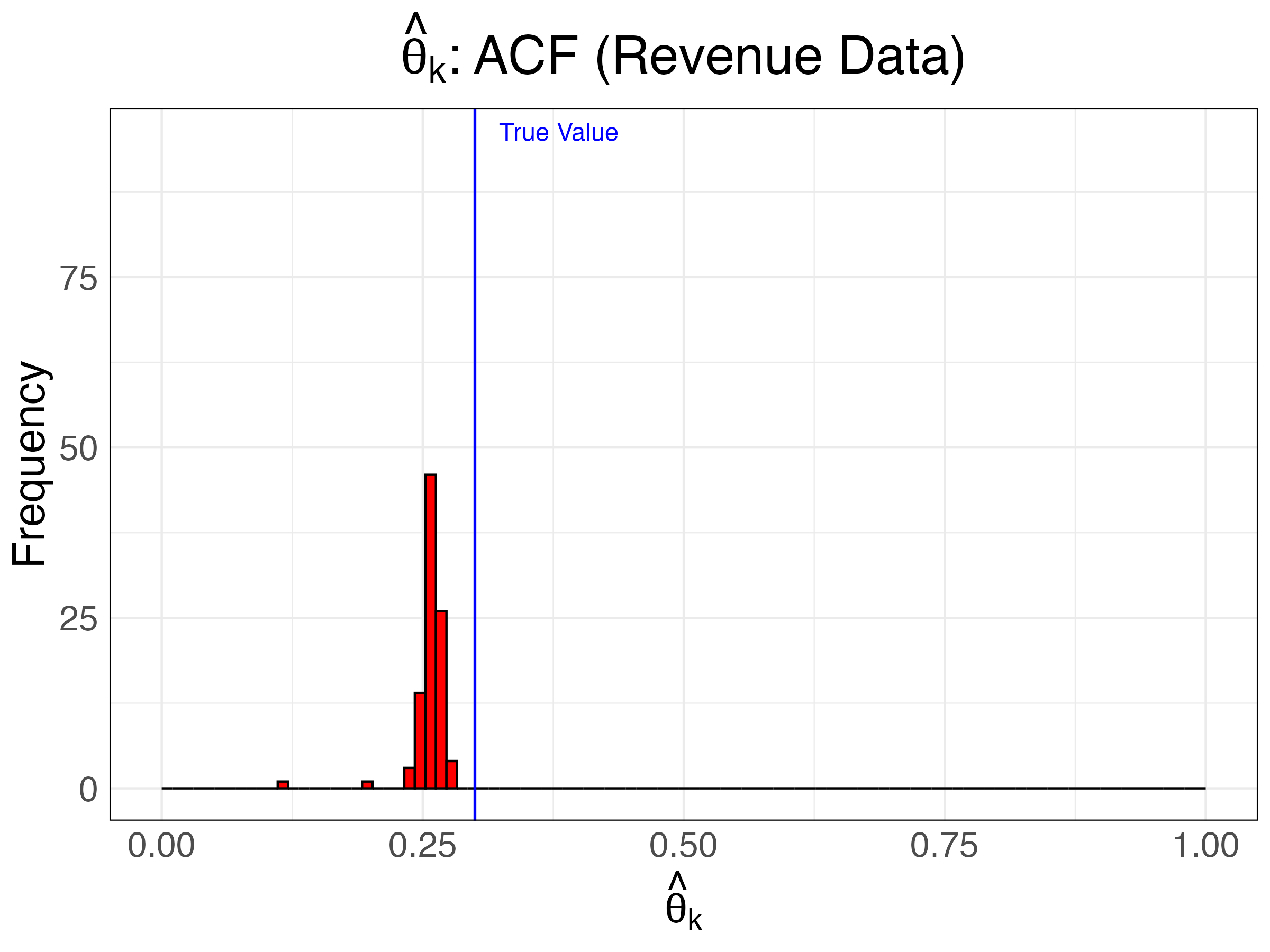}\includegraphics[scale=0.3]{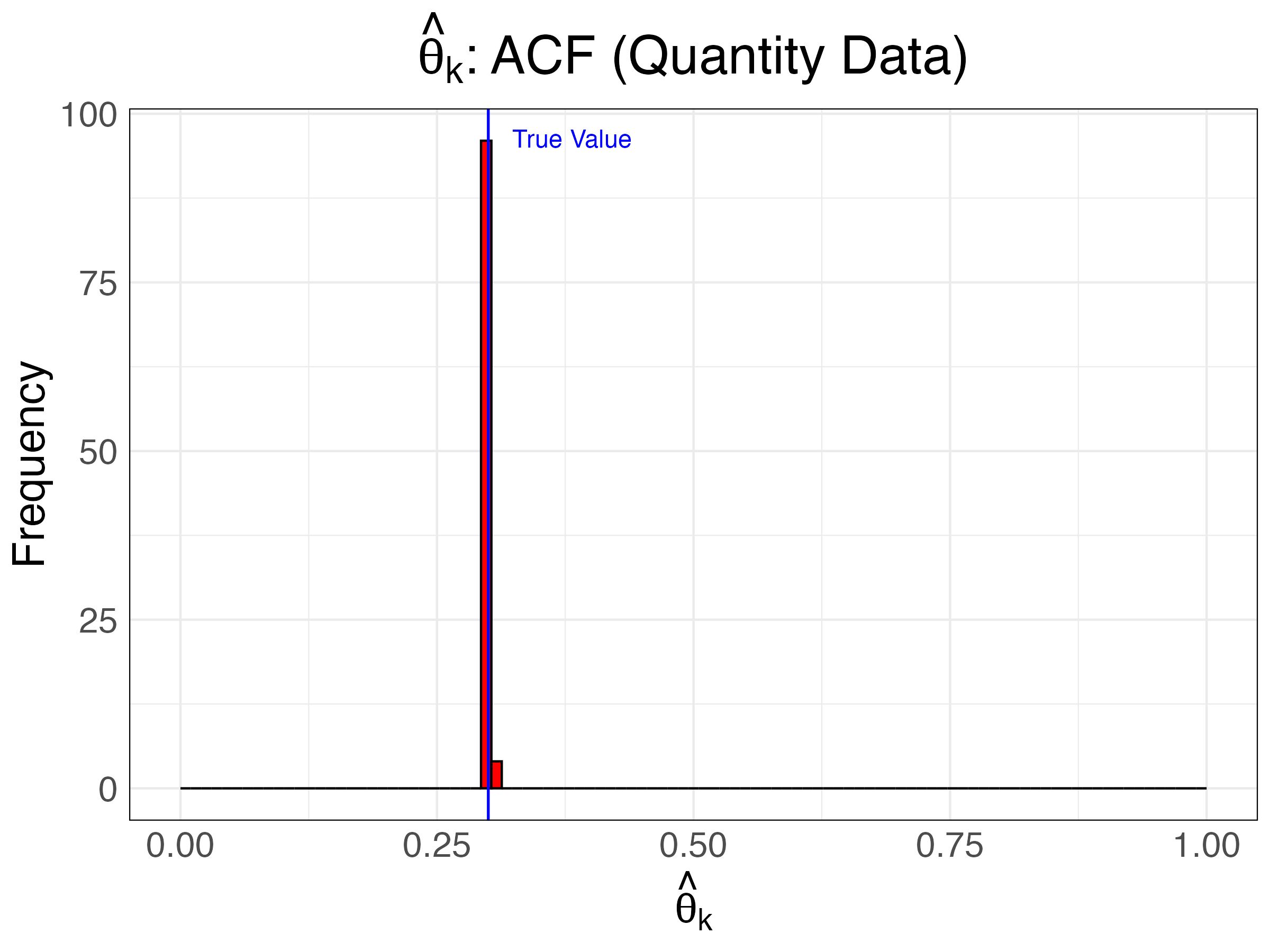} 
\end{figure}

Figure \ref{fig:CKS} shows the histograms of 100 estimates for $(\theta_{m},\theta_{k},\theta_{l})$
from our proposed estimator and $(\beta_{t},\delta_{t})$ at $t=T$
on the HSA demand system. Since $\gamma_{t}$ varies across simulation,
we report the histograms of estimation errors $\hat{\gamma}_{t}-\gamma_{t}$.
They are tightly clustered around their true values, suggesting that
our method recovers the structural parameters very well. Figure
\ref{fig:CKS-TFP-scatter} shows the scatter plot of true versus estimated
TFPs and markups for the first 20 Monte Carlo simulations. The strong
alignment of points along the 45-degree lines accompanying with the
low RMSEs and high correlations suggest that our method precisely
estimates TFPs and markups.

\begin{figure}[H]
\begin{centering}
\caption{Production Function and Demand System Estimation with Revenue Data}
\label{fig:CKS} \includegraphics[scale=0.3]{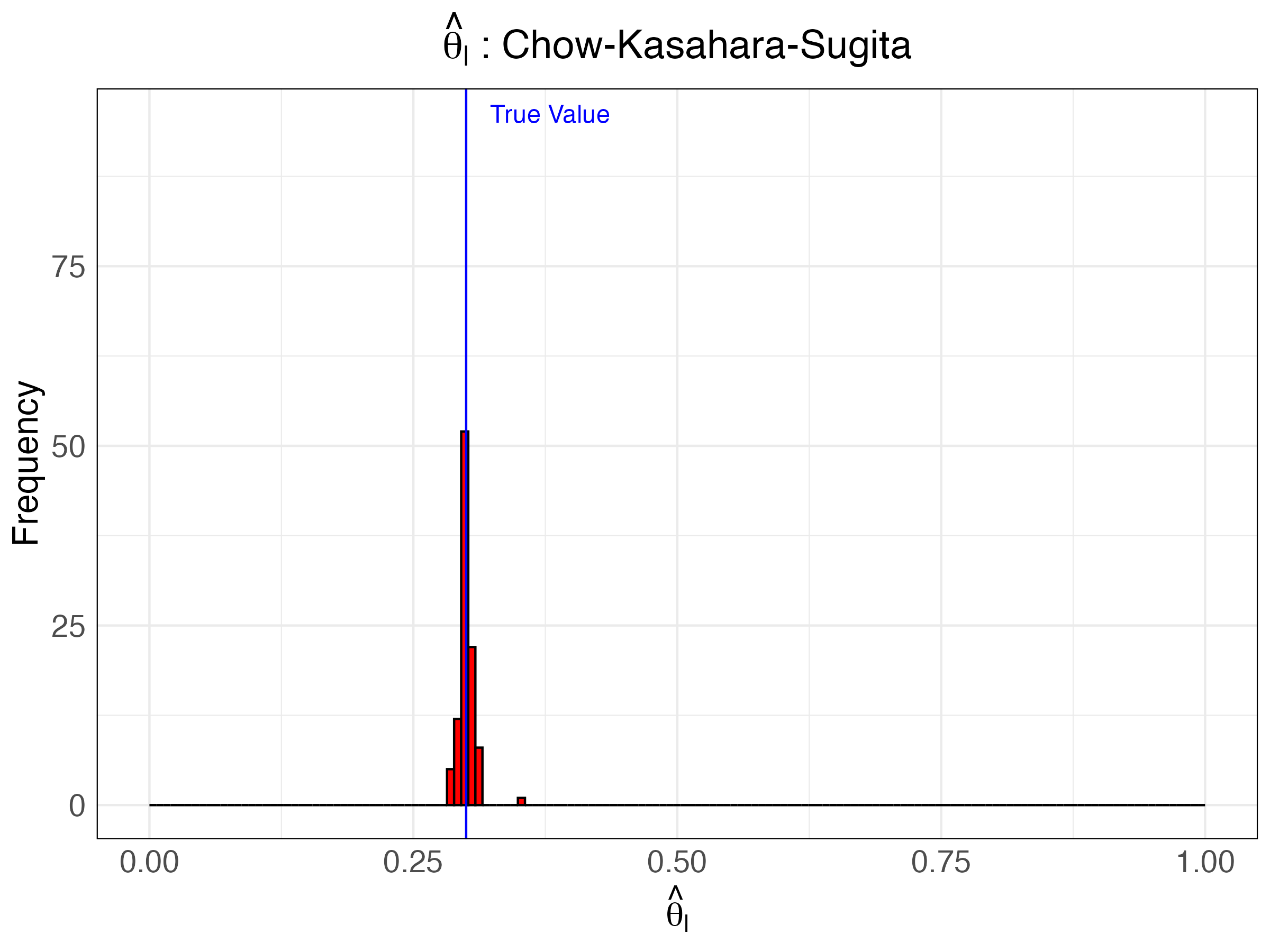}\includegraphics[scale=0.3]{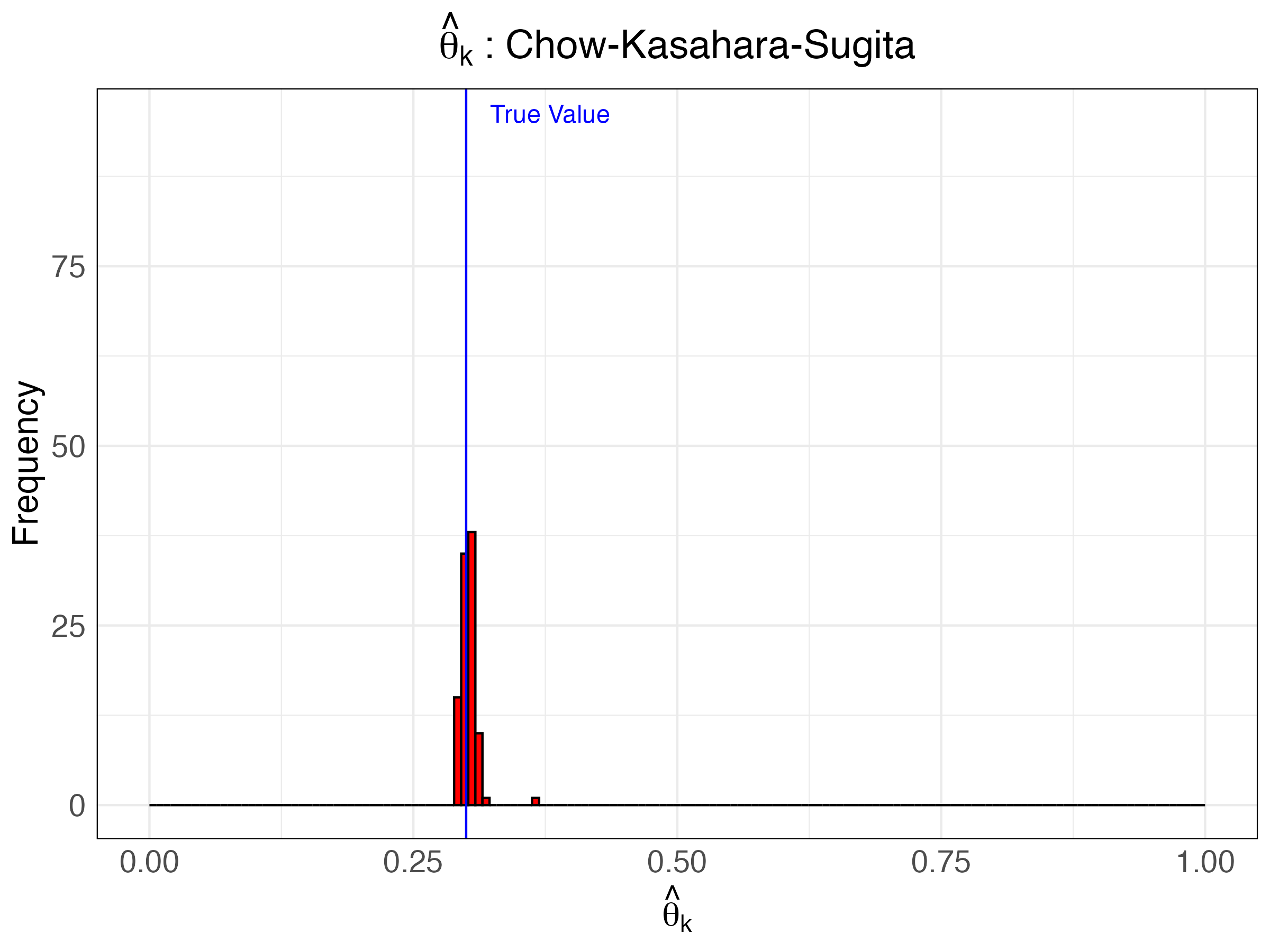} 
\par\end{centering}
\centering{} \includegraphics[scale=0.3]{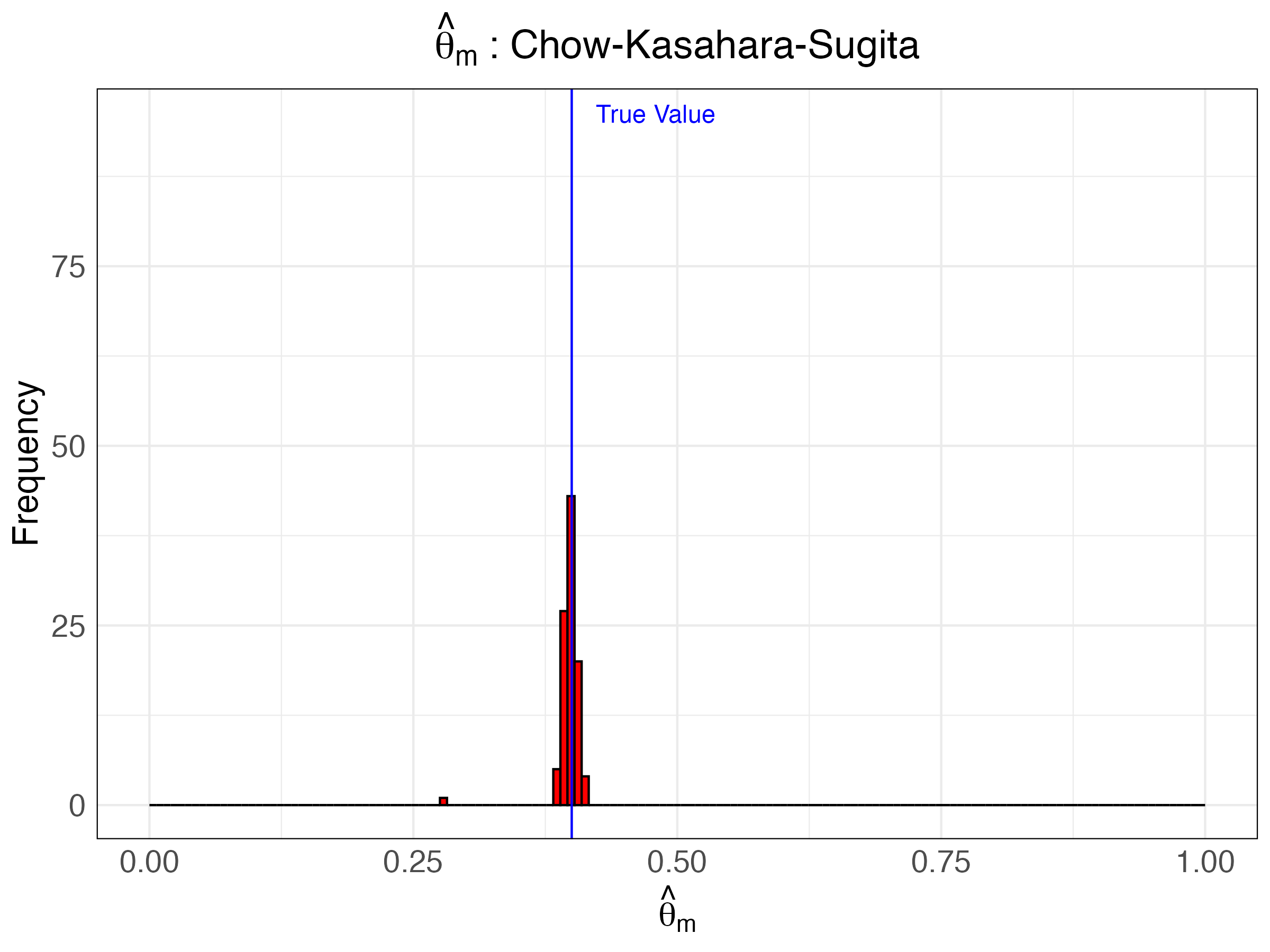}\includegraphics[scale=0.3]{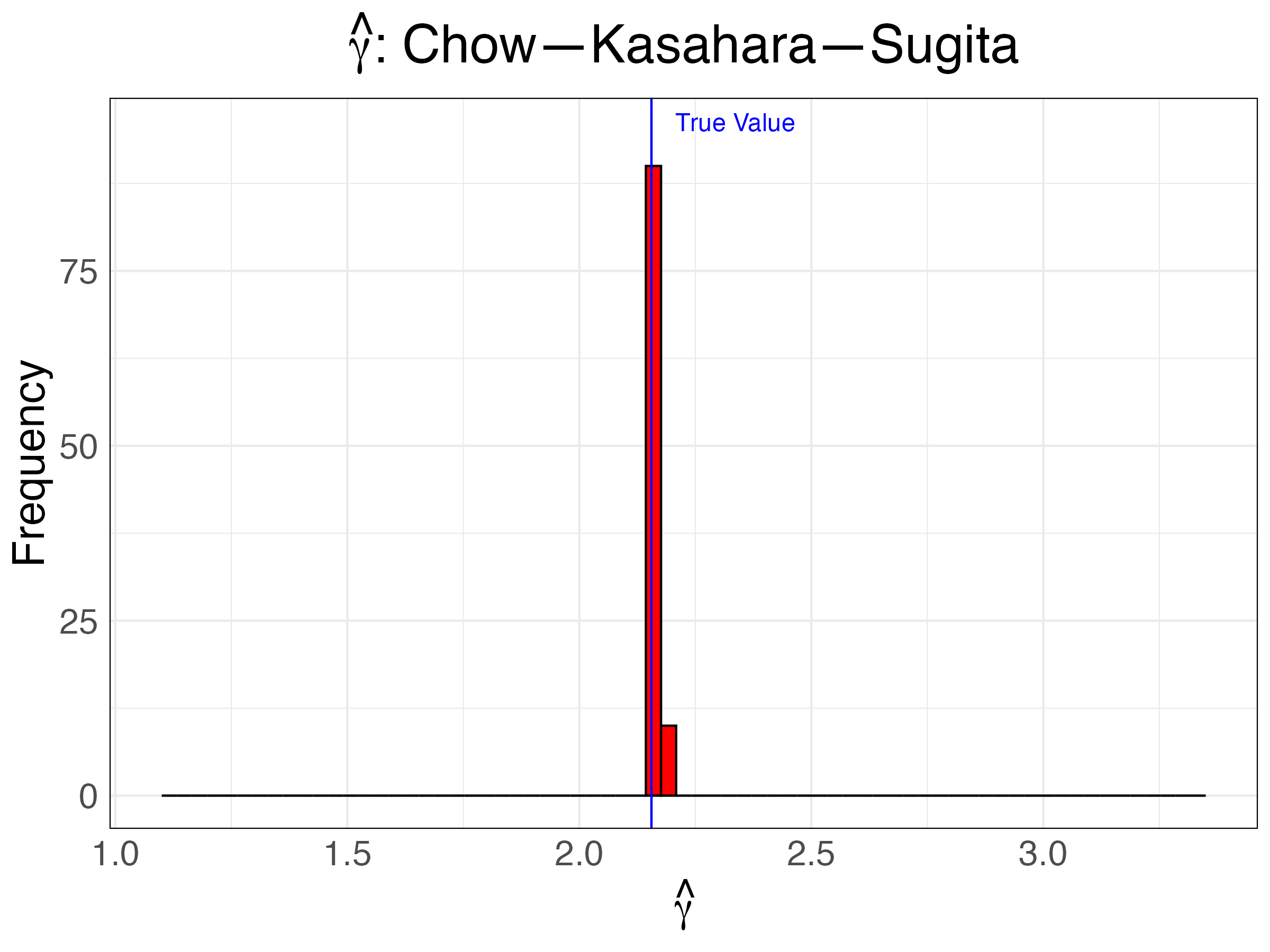}
\centering{}\includegraphics[scale=0.3]{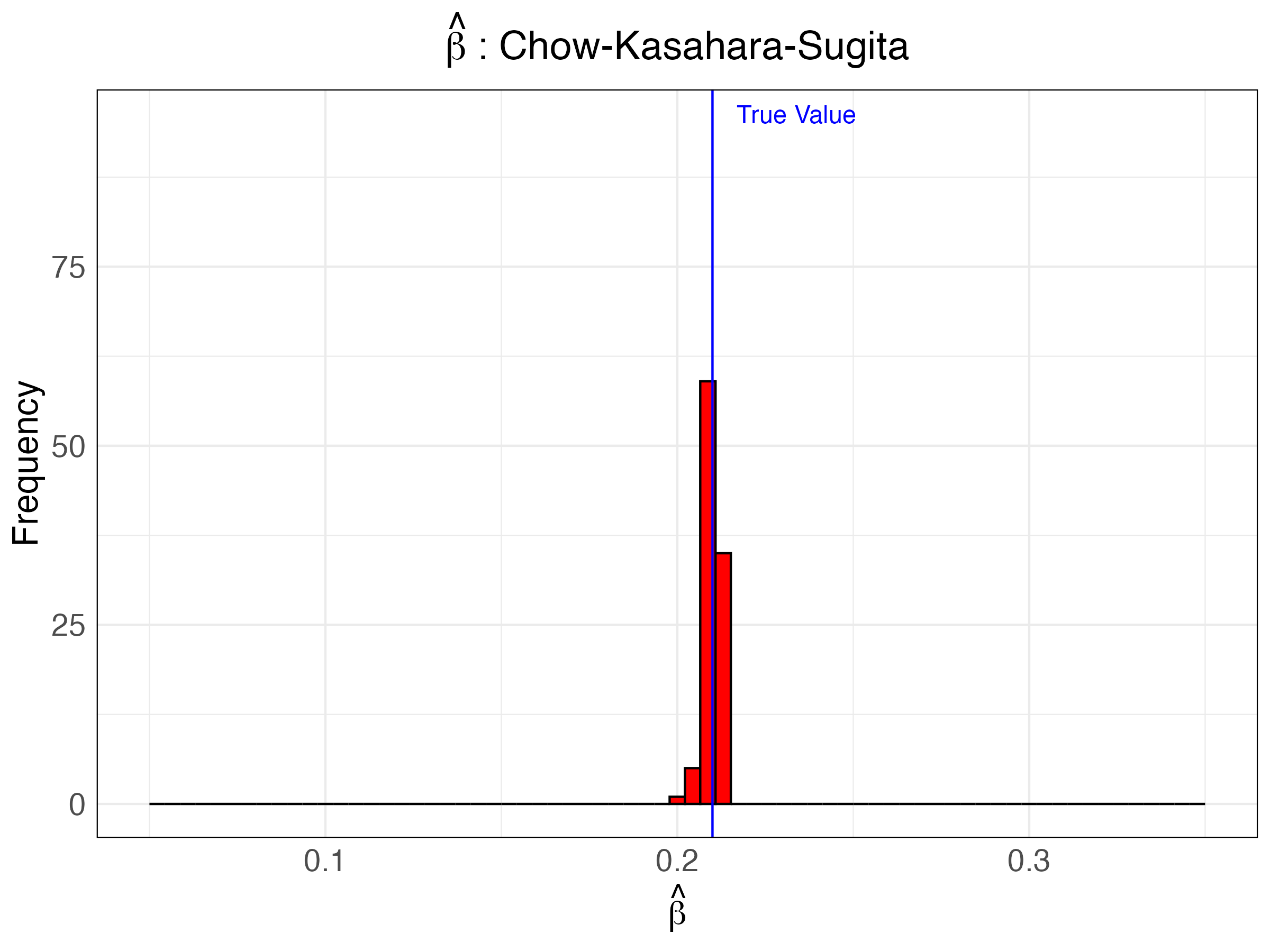}\includegraphics[scale=0.3]{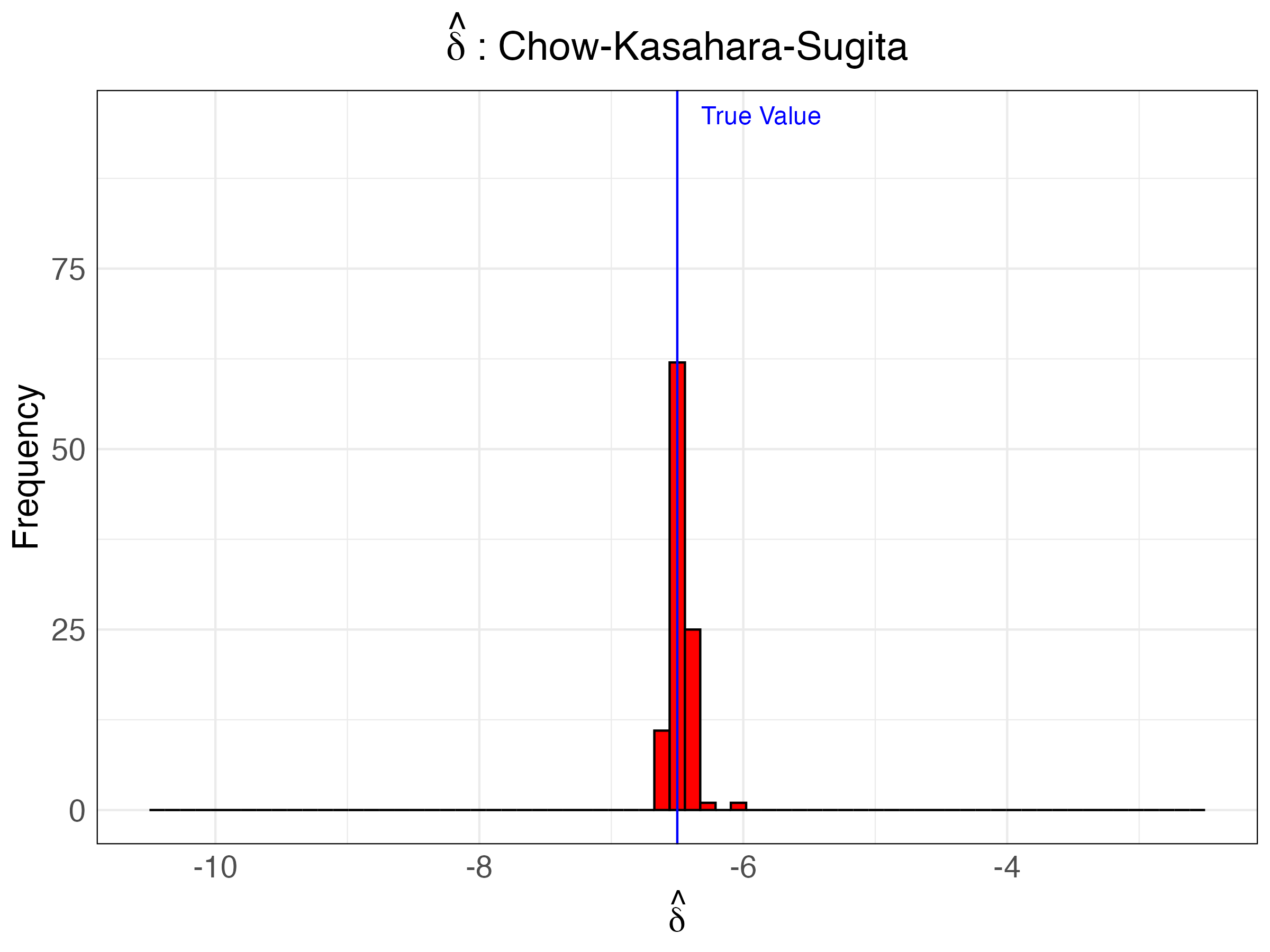}
\end{figure}

\begin{figure}[H]
\centering{}\caption{True and estimated TFPs and Markups for first 20 MC Simulations }
\label{fig:CKS-TFP-scatter} \includegraphics[scale=0.3]{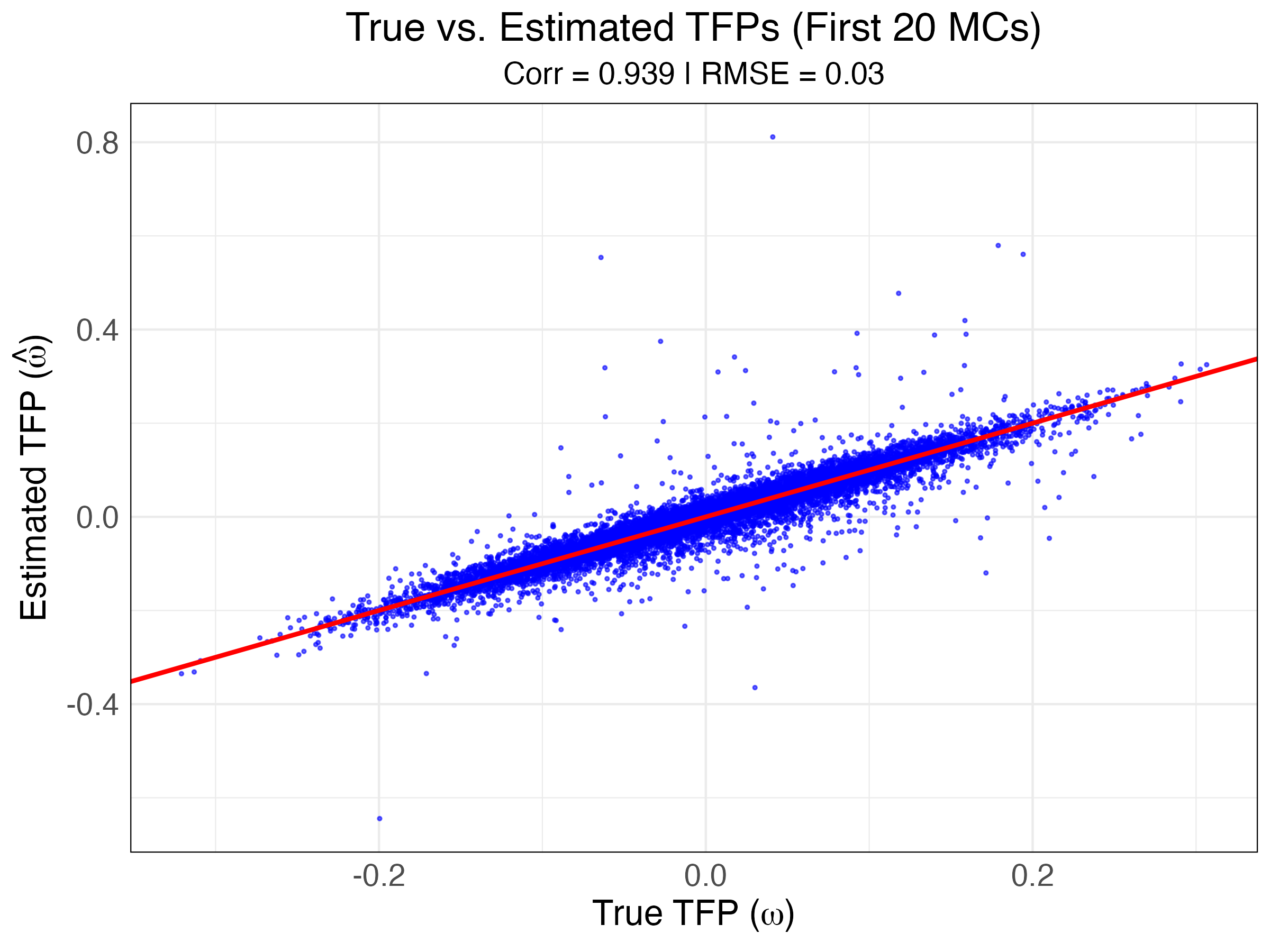}
\includegraphics[scale=0.3]{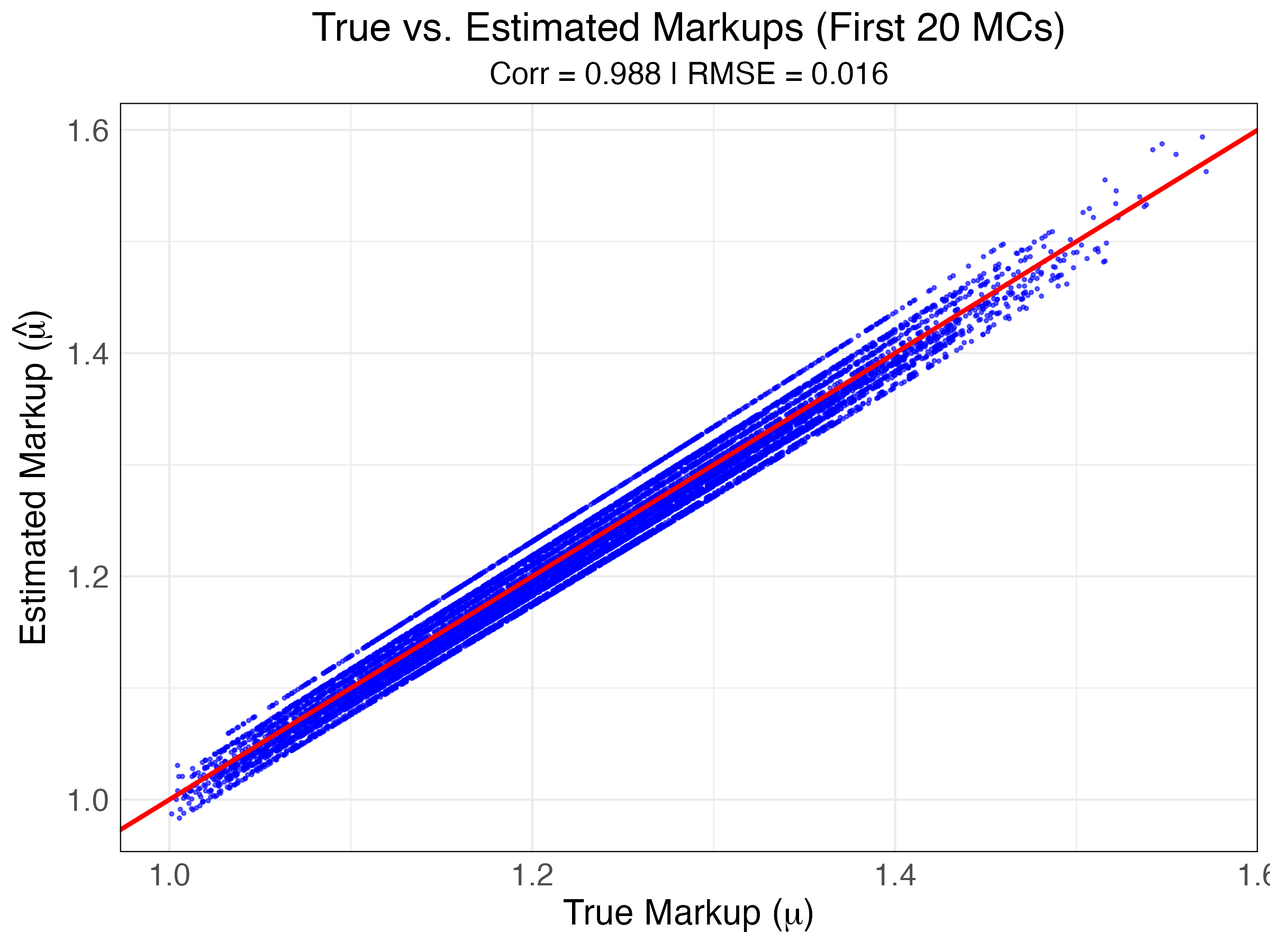} 
\end{figure}
\begin{flushleft}
\vspace{-0.4cm}
Notes: The red line indicates the 45-degree line.
\end{flushleft}

\section{Empirical Application: Chilean Manufacturing Sector}

\label{sec:Chilean}

The semiparametric estimator is applied to the Chilean manufacturing
plant dataset, derived from the census conducted by the Chilean Instituto
Nacional de Estadística, covering all plants with 10 or more employees
from 1993 to 1996. The primary objective of this empirical application
is to illustrate the practical feasibility of our method in a well-studied
empirical environment, rather than to uncover new stylized facts.
The Chilean manufacturing census has been used extensively in the
literature, making it a natural benchmark for assessing how our approach
performs relative to existing empirical frameworks. In addition, we
use this setting to examine whether the commonly imposed CES demand
specification is rejected by the data in favor of the more flexible
HSA demand system, and to quantify the welfare losses associated with
markups through the counterfactual experiments discussed in Section
\ref{subsec:demand}.

Following the standard approach in this literature, we define labor
input as the number of workers, material input as materials cost,
and revenue as income plus the value of capital produced for own use,
with all nominal values deflated using industry-specific deflators.
Capital input is constructed as the sum of deflated values for buildings,
machinery, and vehicles using the perpetual inventory method. Our
analysis focuses on the three largest manufacturing industries in
1996, corresponding to 2-digit SIC codes 31 (Food, Beverage, and Tobacco),
32 (Textiles, Apparel, and Leather Products), and 38 (Metal Products,
Electric and Non-electric Machinery, Transport Equipment, and Professional
Equipment). We exclude plants with non-positive capital, as well as
those with material cost-to-revenue ratios below zero, above one,
or in the bottom and top two percentiles of the distribution, in order
to remove observations that are inconsistent with the production model
or likely to reflect reporting errors and extreme measurement noise.

\subsection{Result}

\begin{table}[H]
\begin{centering}
\begin{tabular}[t]{lccccc}
\toprule 
Industry  & $n$  & $\hat{\theta}_{m}$  & $\hat{\theta}_{k}$  & $\hat{\theta}_{l}$  & $\hat{\bar{\mu}}$ \tabularnewline
\midrule 
31  & 736  & 0.852  & 0.012  & 0.136  & 1.392 \tabularnewline
 &  & (0.035)  & (0.009)  & (0.034)  & (0.059) \tabularnewline
 &  &  &  &  & \tabularnewline
32  & 463  & 0.755  & 0.070  & 0.174  & 1.501 \tabularnewline
 &  & (0.064)  & (0.039)  & (0.041)  & (0.126) \tabularnewline
 &  &  &  &  & \tabularnewline
38  & 391  & 0.685  & 0.043  & 0.272  & 1.661 \tabularnewline
 &  & (0.063)  & (0.035)  & (0.056)  & (0.155) \tabularnewline
\bottomrule
\end{tabular}
\par\end{centering}
\caption{Chilean Manufacturing plant estimation: Step 1, Step 2, and Step 3
(Industries 31, 32, and 38 in 1996). Standard errors in parentheses
with 100 non-parametric bootstrap iterations. }
\label{tab:chilean_estimates1} 
\end{table}

\begin{table}[H]
\begin{centering}
\begin{tabular}[t]{lcccc}
\toprule 
Industry  & $n$  & $\hat{\beta}$  & $\hat{\gamma}$  & $\hat{\delta}$ \tabularnewline
\midrule 
31  & 700  & 0.173  & 1.973  & -8.085 \tabularnewline
 &  & (0.005)  & (0.113)  & (0.397) \tabularnewline
 &  &  &  & \tabularnewline
32  & 408  & 0.087  & 0.979  & -6.913 \tabularnewline
 &  & (0.015)  & (0.236)  & (0.897) \tabularnewline
 &  &  &  & \tabularnewline
38  & 351  & 0.093  & 1.272  & -5.237 \tabularnewline
 &  & (0.014)  & (0.288)  & (0.670) \tabularnewline
\bottomrule
\end{tabular}
\par\end{centering}
\caption{Chilean Manufacturing plant estimation: Step 4 (Industries 31, 32,
and 38 in 1996). Standard errors in parentheses with 100 non-parametric
bootstrap iterations. }
\label{tab:chilean_estimates2} 
\end{table}

Tables \ref{tab:chilean_estimates1} and \ref{tab:chilean_estimates2}
report estimates of the structural parameters for the Chilean manufacturing
industries obtained using our method. Notably, the estimates of $\beta$
is statistically significantly different from zero. Since the HSA
demand system nests the CES demand system as the special case $\beta=0$,
this provides evidence against the CES specification in favor of the
more flexible HSA demand system in Chilean manufacturing industries. Appendix~\ref{app:robustness_rts} shows that these findings are robust
under decreasing returns to scale.



Figures \ref{fig:rvsfittedr} and \ref{fig:uvseps} show scatter plots
of observed revenue $r_{it}$ against fitted revenue from the HSA
demand system in Step 4, along with quantile--quantile plots comparing
the estimated demand shocks $\epsilon_{it}$ from Step 4 with the
corresponding nonparametric estimates $u_{it}$ from Step 1. Observed
revenue aligns closely with fitted revenue under the parametric HSA
restriction, and the quantiles of $\epsilon_{it}$ generally track
those of $u_{it}$ along the 45-degree line. While the fit is not
perfect, these patterns provide supportive evidence for the HSA demand
specification, particularly for Industries 31 and 32. Appendix~\ref{app:robustness_rts} shows that these patterns are preserved
under decreasing returns to scale.

\begin{figure}[H]
\caption{Observed revenue vs. fitted revenue from Step 4}
\label{fig:rvsfittedr} \centering{}\includegraphics[scale=0.4]{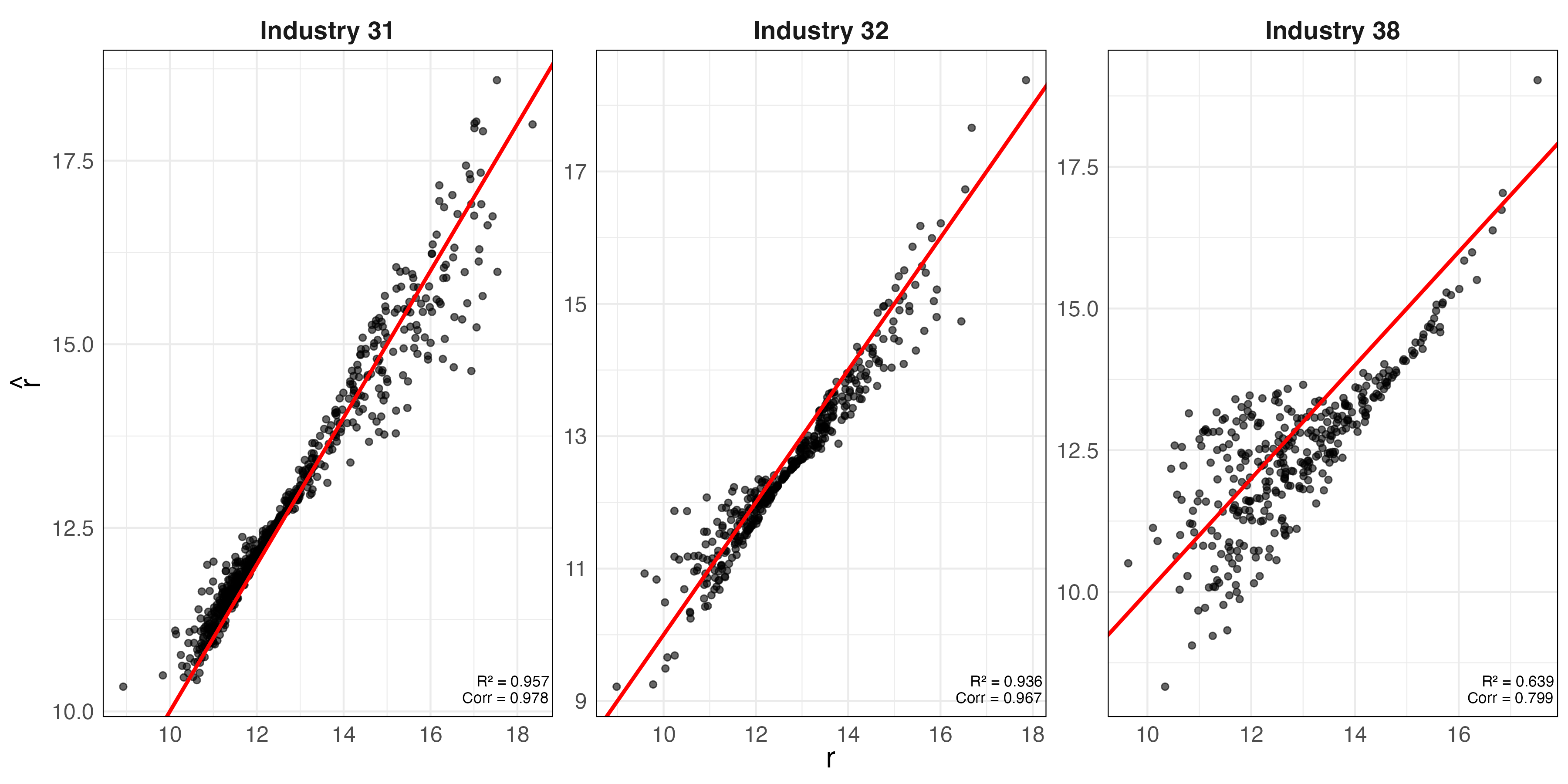} 
\begin{flushleft}
\vspace{-0.2cm}
Notes: The red line indicates the 45-degree line.
\end{flushleft}
\end{figure}

\begin{figure}[H]
\caption{Rank of demand shock from Step 1 vs. Step 4}
\label{fig:uvseps} \centering{}\includegraphics[scale=0.4]{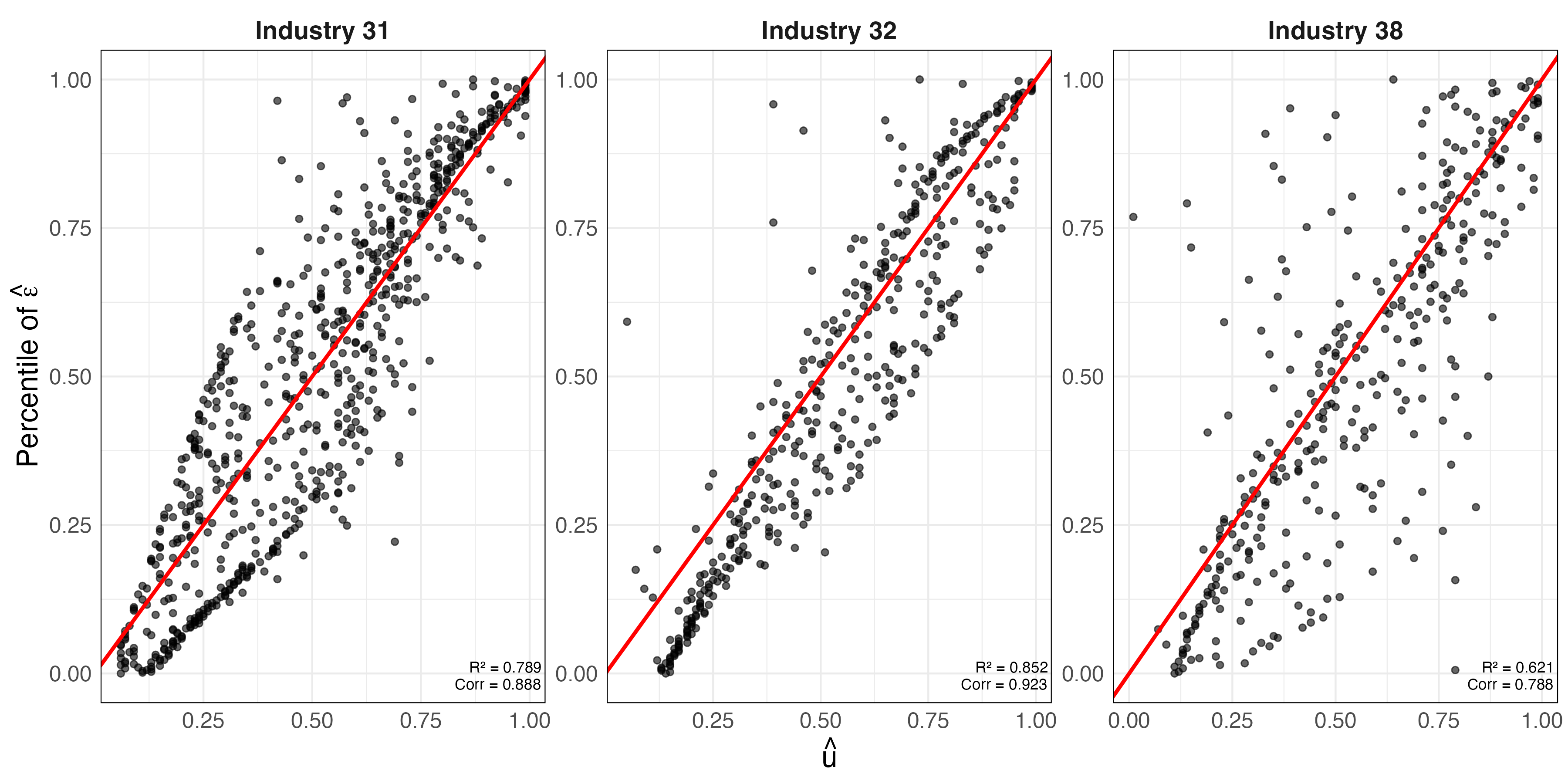} 
\begin{flushleft}
\vspace{-0.2cm}
Notes: The red line indicates the 45-degree line.
\end{flushleft}
\end{figure}




\subsection{Counterfactual Welfare Analysis}

Using the estimated HSA demand system, we quantify the consumer utility
loss attributable to firm's market power by calculating the compensation
variation for a counterfactual marginal cost pricing equilibrium (MCPE)
as described in Section \ref{subsec:demand}.

In our counterfactual analysis, we hold fixed the structural primitives:
the production function $f_{t}$, the demand shock distribution $G_{u}$,
and the structural taste parameter $\gamma_{t}$. The counterfactual
equilibrium under a marginal cost pricing equilibrium (MCPE) recomputes
the aggregate index  $\Delta q_{t}^c$ and firm-level outcomes,
ensuring that the equilibrium condition (\ref{eq:MCPE}). Specifically,
we implement the following procedure. 


First, we calculate a monopolistic competition equilibrium (MCE),
utilizing the estimated structural parameters from Tables \ref{tab:chilean_estimates1}
and \ref{tab:chilean_estimates2} and the firm-level states $(k_{it},l_{it},\hat{\omega}_{it},\hat{\epsilon}_{it})$.
Since the observed output values do not exactly satisfy the equilibrium
conditions (\ref{eq:MCE})---due to model misspecification and estimation
error for $\gamma_{t}$---we recompute the equilibrium output vector
$\{y_{it}^{m}\}_{i=1}^{n}$ and parameter $\gamma_{t}^{new}$ under
the normalization $q_{t}^{m}=0$ by jointly solving the conditions
in (\ref{eq:MCE}), which are simplified as: 
\begin{align*}
\Phi_{t}+\hat{\delta}_{t}-\hat{\beta}_{t}y_{it}^{m}+\gamma_{t}^{new}+\Xi_{it}-\frac{y_{it}^{m}}{\hat{\theta}_{m}}+\frac{1}{\hat{\beta}_{t}}\ln(1+\hat{\epsilon}_{it})\\
-\left(1+\frac{1}{\hat{\beta}_{t}}\right)\ln\left(\exp(-\hat{\beta}_{t}y_{it}^{m}+\gamma_{t}^{new})+\hat{\epsilon}_{it}\right) & =0\text{ for }i=1,..,N_{t}\\
\sum_{i=1}^{N_{t}}\exp\left(\hat{\delta}_{t}-\frac{1}{\hat{\beta}_{t}}\ln\left(\frac{\exp(-\hat{\beta}_{t}y_{it}^{m}+\gamma_{t}^{new})+\hat{\epsilon}_{it}}{1+\hat{\epsilon}_{it}}\right)\right) & =1
\end{align*}
where $\Xi_{it}=\ln\hat{\theta}_{m}+(\hat{\theta}_{k}k_{it}+\hat{\theta}_{l}l_{it}+\hat{\omega}_{it})/\hat{\theta}_{m}$
and $p_{t}^{m}$ is normalized to zero. The resulting output vector and
parameters are fully consistent with our HSA demand system, which
mitigates model misspecification bias and addresses estimation errors
in the counterfactual analysis.

Second, we consider a marginal cost pricing equilibrium (MCPE) for
given counterfactual log income $\Phi_{t}^{c}$. We find an output
vector $\mathbf{y}_{t}^{c}(\Phi_{t}^{c})$ and quantity index $\Delta{q}_{t}^{c}(\Phi_{t}^{c})$
by solving (\ref{eq:MCPE}), which are simplified as: 
\begin{align*}
\Phi_{t}^{c}+\hat{\delta}_{t}+\Xi_{it}-\frac{y_{it}^{c}(\Phi_{t}^{c})}{\hat{\theta}_{m}}-\frac{1}{\hat{\beta}_{t}}\ln\left(\frac{\exp(-\hat{\beta}_{t}\left(y_{it}^{c}(\Phi_{t}^{c})-\Delta{q}_{t}^{c}(\Phi_{t}^{c})\right)+\gamma_{t}^{new})+\hat{\epsilon}_{it}}{1+\hat{\epsilon}_{it}}\right) & =0\text{ for }i=1,..,N_{t}\\
\sum_{i=1}^{N_{t}}\exp\left(\hat{\delta}_{t}-\frac{1}{\hat{\beta}_{t}}\ln\left(\frac{\exp(-\hat{\beta}_{t}\left(y_{it}^{c}(\Phi_{t}^{c})-\Delta{q}_{t}^{c}(\Phi_{t}^{c})\right)+\gamma_{t}^{new})+\hat{\epsilon}_{it}}{1+\hat{\epsilon}_{it}}\right)\right) & =1.
\end{align*}
We assume the same income case $\Phi_{t}^{c}=\Phi_{t}$ for our benchmark,
while also analyzing how potential changes in income might affect
the equilibrium outcome.


To calculate the compensating variation for transitioning from an
MCE to an MCPE, we solve for the counterfactual income $\Phi_{t}^{c*}$
that results in a zero utility change, as defined in (\ref{eq:CV}):
\begin{align*}
\Delta\ln U^{c}(\Phi_{t}^{c*}) & =\Delta{q}_{t}^{c}(\Phi_{t}^{c*})\\
 & +\sum_{i}\int_{y_{it}^{m}}^{y_{it}^{c}(\Phi_{t}^{c*})-\Delta{q}_{t}^{c}(\Phi_{t}^{c*})}\exp\left(\hat{\delta}_{t}-\frac{1}{\hat{\beta}_{t}}\ln\left(\frac{\exp(-\hat{\beta}_{t}(\zeta-q_{t}^{m})+\gamma_{t}^{new})+\hat{\epsilon}_{it}}{1+\hat{\epsilon}_{it}}\right)\right)d\zeta=0.
\end{align*}
The compensating variation is then calculated as $CV_{t}:=\exp(\Phi_{t}^{c*})-\exp(\Phi_{t})$,
which quantifies the consumer's welfare change from firm's market
power.

\begin{table}[h]
\centering %
\begin{tabular}[t]{lccc}
\toprule 
Industry  & CV  & $\Delta\Pi$  & Overall \tabularnewline
\midrule 
31  & -14.7  & -11.5  & 3.20 \tabularnewline
 & (2.65)  & (2.43)  & (0.648) \tabularnewline
 &  &  & \tabularnewline
32  & -18.9  & -11.9  & 6.93 \tabularnewline
 & (7.74)  & (7.14)  & (1.95) \tabularnewline
 &  &  & \tabularnewline
38  & -13.3  & -6.12  & 7.14 \tabularnewline
 & (11.3)  & (7.82)  & (4.28) \tabularnewline
\bottomrule
\end{tabular}\caption{Compensating Variation, profit loss, and overall welfare change in
percentage of industry revenue $\exp(\Phi_{t})$ in the transition
from original equilibrium to MCPE of Chilean Industries 31, 32, and
38 in 1996 under HSA demand system.  Sign convention: $CV_{t}<0$
indicates consumers require less income under competition (gain from
eliminating market power); $\Delta\Pi_{t}<0$ is firms' profit loss;
Overall $=\Delta\Pi_{t}-CV_{t}>0$ is the net welfare gain. Standard errors in parentheses
with 100 non-parametric bootstrap iterations.}
\label{tab:welfare} 
\end{table}

Finally, we calculate firms' profit loss. In the case of $\Phi_{t}^{c}=\Phi_{t}$,
the total profit change (\ref{eq:profit change}) is expressed as
\begin{align}
\Delta\Pi:=\Pi^{c}-\Pi^{m} & =\sum_{i=1}^{N_{t}}\exp\left(-\frac{\hat{\theta}_{k}k_{it}+\hat{\theta}_{l}l_{it}+\hat{\omega}_{it}}{\hat{\theta}_{m}}\right)\left\{ \exp\left(\frac{y_{it}^{m}}{\hat{\theta}_{m}}\right)-\exp\left(\frac{y_{it}^{c}(\Phi_{t})}{\hat{\theta}_{m}}\right)\right\} \label{eq:}
\end{align}
since $\chi_{it}\left(y_{it}\right)=(y_{it}-\hat{\theta}_{k}k_{it}-\hat{\theta}_{l}l_{it}-\hat{\omega}_{it})/\hat{\theta}_{m}$.

From Table \ref{tab:welfare}, we found empirical evidence that under
our HSA demand system market power in these industries results in
consumer's welfare losses of approximately 10\%--15\% and profit
gains of approximately 4\%--11\%, with overall welfare losses of
3\%--6\% of industry revenue in the three largest Chilean manufacturing
industries in 1996. To put these findings in context, our estimates far exceed the classic Harberger-triangle calculation of less than 0.1\% of GDP \citep{harberger1954}. They are more consistent with recent general-equilibrium analyses that find substantially larger welfare costs. \citet{baqaee2020productivity} estimate that eliminating misallocation would raise output by approximately 15\%, while \citet{edmond2023costly} find welfare costs of variable markups of around 7.5\% in consumption-equivalent terms. Our estimates, though derived from a different structural framework, are comparable in magnitude---reinforcing the view that the welfare costs of imperfect competition are economically significant.  Appendix~\ref{app:robustness_rts} shows that these estimates are broadly robust under moderate decreasing returns to scale of 0.9.

\section{Concluding Remarks}

\label{sec:Concluding-Remarks}


This paper develops constructive nonparametric identification of production functions and markups from revenue data, simultaneously addressing the two fundamental challenges in production function estimation since \citet{ma44ecma} when revenue is used as output: input-TFP correlation and markup heterogeneity bias. By modeling revenue as a function of output, observed characteristics, and an unobserved demand shock, we identify various economic objects of interest under standard assumptions. Our semiparametric estimator, implementable with standard datasets, performs well in simulation. Applied to Chilean manufacturing plant data, we find evidence of CES demand misspecification and estimate that market power generates welfare losses of approximately 3\%--6\% of industry revenue in the three largest manufacturing industries in 1996.

While our identification results establish that firm-level revenue data suffice to recover production functions and consumer demand without physical quantity data, the proof relies on structural assumptions that merit discussion. First, we assume perfectly competitive input markets; monopsony power would make input prices firm-specific and endogenous, breaking the link between expenditure shares and marginal products. Second, we assume Hicks-neutral, scalar productivity; factor-augmenting technological change requires additional structure or data for identification. Third, we assume monopolistic competition; in oligopolistic markets, strategic interaction makes the first-order conditions underlying our markup identification insufficient without a fully specified equilibrium model. Relaxing any of these assumptions is an important direction for future work.

 \bibliographystyle{asa}
\bibliography{markup_new_intro_review}

\newpage{}

\appendix


\setcounter{page}{1} 
\global\long\def\thepage{\Alph{section}.\arabic{page}}%
\setcounter{equation}{0} 
\global\long\def\theequation{\Alph{section}.\arabic{equation}}%
\setcounter{assumption}{0} 
\global\long\def\theassumption{\Alph{section}.\arabic{assumption}}%
\setcounter{thm}{0} 
\global\long\def\thethm{\Alph{section}.\arabic{thm}}%
\setcounter{prop}{0} 
\global\long\def\theprop{\Alph{section}.\arabic{prop}}%
\setcounter{lem}{0} 
\global\long\def\thelem{\Alph{section}.\arabic{lem}}%
\setcounter{figure}{0} 
\global\long\def\thefigure{\Alph{section}.\arabic{figure}}%
\setcounter{table}{0} 
\global\long\def\thetable{\Alph{section}.\arabic{table}}%

\begin{center}
{\LARGE\textbf{Online Supplemental Appendix}} 
\par\end{center}

\section{Simulation}
\label{app:simulation}

\subsection{CoPaTh-HSA Demand System }
\label{subsec:CoPaTh-HSA-Demand-System}

We consider the ``Incomplete Constant (and Common) Pass-Through''
formulation of the CoPaTh-HSA demand system in \citet{matsuyama2020constant}.
With their original notions, the budget share function for product
$\omega$ is expressed as: 
\begin{equation}
S_{\omega}^{*}\left(\frac{Y}{Q(\mathbf{Y})}\right)=\gamma_{\omega}\beta_{\omega}\left[\left(1-\frac{1}{\sigma_{\omega}}\right)\left(\frac{Y/Q(\mathbf{Y})}{\gamma_{\omega}}\right)^{-\Delta}+\frac{1}{\sigma_{\omega}}\right]^{-1/\Delta}\label{eq:MU}
\end{equation}
where $Y$ is the level of output and $Q(\mathbf{Y})$ is the quantity
index which is a function of the output vector $\mathbf{Y}$. The
pass-through rate 
\[
\rho=\frac{\partial\ln P}{\partial\ln MC}=1+\frac{\partial\ln\mu}{\partial\ln MC}=\frac{1}{\Delta+1}
\]
is a function of parameter $\Delta\ge0$ where $\mu=P/MC$ is a markup.
When $\Delta=0$ and $\sigma_{\omega}=\sigma$, the demand system
is reduced to the conventional CES system.

We reformulate it in log variables: 
\begin{align*}
s_{\omega}^{*}(y) & =\ln S_{\omega}^{*}(Y)=\gamma_{\omega}\beta_{\omega}\left[\frac{1}{\sigma_{\omega}}+\left(1-\frac{1}{\sigma_{\omega}}\right)\left(\frac{\exp\left(y-q\right)}{\gamma_{\omega}}\right)^{-\Delta}\right]^{-1/\Delta}\\
 & =\ln\left(\gamma_{\omega}\beta_{\omega}\right)-\frac{1}{\Delta}\ln\left[\left(1-\frac{1}{\sigma_{\omega}}\right)\left(\frac{\exp\left(y-q\right)}{\exp\left(\ln\gamma_{\omega}\right)}\right)^{-\Delta}+\frac{1}{\sigma_{\omega}}\right]\\
 & =\ln\left(\gamma_{\omega}\beta_{\omega}\right)-\frac{1}{\Delta}\ln\left[\left(\frac{\sigma_{\omega}-1}{\sigma_{\omega}}\right)\exp\left(-\Delta\left(y-q\right)+\Delta\ln\gamma_{\omega}\right)+\frac{1}{\sigma_{\omega}}\right]\\
 & =\ln\left(\gamma_{\omega}\beta_{\omega}\right)-\frac{1}{\Delta}\ln\left[\frac{\exp\left(-\Delta\left(y-q\right)+\Delta\ln\gamma_{\omega}\right)+1/\left(\sigma_{\omega}-1\right)}{\sigma_{\omega}/\left(\sigma_{\omega}-1\right)}\right]
\end{align*}

Our formulation of the CoPaTh-HSA demand system of inverse demand
functions is 
\[
s^{*}\left(y_{it},\epsilon_{it}\right)=\delta_{t}-\frac{1}{\beta_{t}}\ln\left(\frac{\exp(-\beta_{t}(y_{it}-q_{t}(\mathbf{y}_{t},\mathbf{\epsilon}_{t}))+\gamma_{t})+\epsilon_{it}}{1+\epsilon_{it}}\right).
\]
While the original formulation has three firm specific demand shifters
$(\gamma_{\omega},\beta_{\omega},\sigma_{\omega})$, we only allow
one shifter.

The correspondence between current parameter notations and \citet{matsuyama2020constant}'s
is as follows: 
\begin{center}
\begin{tabular}{|c|c|c|c|c|c|c|}
\hline 
Current notations  & $i$  & $\delta_{t}$  & $\beta_{t}$  & $\epsilon_{it}$  & $1+\epsilon_{it}$  & $\gamma_{t}$\tabularnewline
\hline 
\hline 
\citet{matsuyama2020constant}'s  & $\omega$  & $\ln\gamma_{\omega}\beta_{\omega}$  & $\Delta=\frac{1-\rho}{\rho}$  & $\frac{1}{\sigma_{\omega}-1}$  & $\frac{\sigma_{\omega}}{\sigma_{\omega}-1}$  & $\Delta\ln\gamma_{\omega}$ \tabularnewline
notations  &  &  &  &  &  & \tabularnewline
\hline 
Range  &  & $(-\infty,\infty)$  & $(0,\infty)$  & $(0,\infty)$  & $(1,\infty)$  & $(-\infty,\infty)$\tabularnewline
\hline 
\end{tabular}
\par\end{center}

\subsection{Data Generating Process }

\label{subsec:Data-Generating-Process}

The demand shock $\epsilon_{it}$ follows an MA(1) process:

\[
\epsilon_{it}=\rho_{\epsilon}\xi_{it-1}+\xi_{it}
\]
where $\xi_{it}$and $\xi_{it-1}$ are independent uniform random
variables with supports {[}0,0.3{]}.

Capital and labor are predetermined and follows the following exogenous
laws of motion: 
\begin{align*}
k_{it} & =0.99k_{it-1}+0.11\omega_{it-1}+e_{kit},\,e_{kit}\sim N(0,0.25^{2}),\,k_{i0}\sim N(10,1)\\
l_{it} & =0.99l_{it-1}+0.11\omega_{it-1}+e_{lit},e_{lit}\sim N(0,0.25^{2}),\,l_{i0}\sim N(10,1).
\end{align*}

\paragraph{Summary statistics}

The following table shows the summary statics of endogenous variables
and exogenous variables. 
\begin{center}
\begin{tabular}[t]{cccccccc}
\toprule 
 & \multicolumn{7}{c}{Endogenous variables}\tabularnewline
\midrule 
$t=5$  &  & Mean  & Min  & P25  & Median  & P75  & Max\tabularnewline
Markup  &  & $1.223$  & $1.001$  & $1.147$  & $1.219$  & $1.295$  & $1.617$\tabularnewline
$m_{it}$  &  & $11.322$  & $8.690$  & $10.948$  & $11.318$  & $11.691$  & $14.161$\tabularnewline
$r_{it}$  &  & $13.437$  & $10.712$  & $13.056$  & $13.445$  & $13.823$  & $16.231$\tabularnewline
\midrule 
 & \multicolumn{7}{c}{Exogenous variables}\tabularnewline
\midrule 
$t=5$  &  & Mean  & Min  & P25  & Median  & P75  & Max\tabularnewline
$\omega_{it}$  &  & $0.000$  & $-0.338$  & $-0.056$  & $0.000$  & $0.056$  & $0.323$\tabularnewline
$k_{it}$  &  & $9.511$  & $4.770$  & $8.771$  & $9.513$  & $10.253$  & $15.055$\tabularnewline
$l_{it}$  &  & $9.505$  & $4.949$  & $8.770$  & $9.505$  & $10.246$  & $14.252$\tabularnewline
\bottomrule
\end{tabular}
\par\end{center}

\subsubsection{\citet{ackerberg2015identification} estimation method}

We estimate the production function with ACF using the R package \texttt{prodest}
by \citep{rovigatti2017production}. Scale parameters are normalized
under constant returns to scale (CRS), and location parameters are
normalized via the mean-zero restriction on the AR(1) TFP process
for the estimates. The initial values in optimization are set to the
estimated parameters from our method for empirical application and
the true parameters for simulation. 

\section{Calculations and Proofs}

\label{app:proof}

\subsection{A necessary and sufficient condition for Assumption \ref{assu: revenue monotonicity}}\label{app:proof-5}

We first derive some derivatives for preparation. From $\varphi_{t}(y_{it},z_{it},u_{it})=y_{it}+\psi_{t}(y_{it},z_{it},u_{it})$,
the demand elasticity is expressed as

\[
\frac{\partial\varphi_{t}}{\partial y_{it}}=1-\frac{1}{\sigma_{t}(y_{it},z_{it}^{d},u_{it})}\Leftrightarrow\sigma_{t}(y_{it},z_{it}^{d},u_{it})=\frac{1}{1-\partial\varphi_{t}(y_{it},z_{it},u_{it})/\partial y_{it}}
\]
Their derivatives are 
\[
\frac{\partial\sigma_{t}}{\partial y_{it}}=\frac{\partial^{2}\varphi_{t}/\partial y_{it}^{2}}{\left(1-\partial\varphi_{t}/\partial y_{it}\right)^{2}}\text{ and }\frac{\partial\sigma_{t}}{\partial u_{it}}=\frac{\partial^{2}\varphi_{t}/\partial y_{it}\partial u_{it}}{\left(1-\partial\varphi_{t}/\partial y_{it}\right)^{2}}.
\]

Denote the profit by 
\[
\pi_{t}(m_{it},\omega_{it},u_{it}):=\exp(\varphi_{t}(f_{t}(m,k_{it},l_{it},z_{it}^{s})+\omega_{it},z_{it}^{d},u_{it}),z_{it}^{d},u_{it})-\exp(p_{t}^{m}+m)
\]
The first order condition for (\ref{eq:profit_maximization}) is 
\begin{align*}
\frac{\partial\pi_{t}}{\partial m_{it}}= & \exp(\varphi_{t}(f_{t}(x_{it},z_{it}^{s})+\omega_{it},z_{it}^{d},u_{it}))\frac{\partial\varphi_{t}(f_{t}(x_{it},z_{it}^{s})+\omega_{it},z_{it}^{d},u_{it})}{\partial y_{it}}\frac{\partial f_{t}(x_{it},z_{it}^{s})}{\partial m_{it}}\\
- & \exp(p_{t}^{m}+m_{it})=0.
\end{align*}
Their cross derivatives are 
\begin{align*}
\frac{\partial^{2}\pi_{t}}{\partial m_{it}\partial\omega_{it}} & =\exp(r_{it})\frac{\partial f_{t}}{\partial m_{it}}\left[\left(\frac{\partial\varphi_{t}}{\partial y_{it}}\right)^{2}+\frac{\partial^{2}\varphi_{t}}{\partial y_{it}^{2}}\right]\\
\frac{\partial^{2}\pi_{t}}{\partial m_{it}\partial u_{it}} & =\exp(r_{it})\frac{\partial f_{t}}{\partial m_{it}}\left(\frac{\partial\varphi_{t}}{\partial y_{it}}\frac{\partial\varphi_{t}}{\partial u_{it}}+\frac{\partial^{2}\varphi_{t}}{\partial y_{it}u_{it}}\right).
\end{align*}
From the implicit function theorem, the derivatives of the material
demand function is 
\[
\frac{\partial\mathbb{M}_{t}}{\partial u_{it}}=-\frac{\partial^{2}\pi_{t}/\partial m_{it}\partial u_{it}}{\partial^{2}\pi_{t}/\partial m_{it}^{2}}\text{ and }\frac{\partial\mathbb{M}_{t}}{\partial\omega_{it}}=-\frac{\partial^{2}\pi_{t}/\partial m_{it}\partial\omega_{it}}{\partial^{2}\pi_{t}/\partial m_{it}^{2}}.
\]
Differentiating $m_{it}=\mathbb{M}_{t}\left(\text{\ensuremath{\mathbb{M}_{t}^{-1}}}(m_{it},w_{it},u_{it}),w_{it},u_{it}\right)$
by $u_{it}$, we obtain the derivatives of the inverse function as
\begin{align*}
\frac{\partial\text{\ensuremath{\mathbb{M}_{t}^{-1}}}(m_{it},w_{it},u_{it})}{\partial u_{it}} & =-\frac{\partial\mathbb{M}_{t}/\partial u_{it}}{\partial\mathbb{M}_{t}/\partial\omega_{it}}\\
 & =-\frac{\partial^{2}\pi_{t}/\partial m_{it}\partial u_{it}}{\partial^{2}\pi_{t}/\partial m_{it}\partial\omega_{it}}\\
 & =-\frac{\frac{\partial\varphi_{t}}{\partial y_{it}}\frac{\partial\varphi_{t}}{\partial u_{it}}+\frac{\partial^{2}\varphi_{t}}{\partial y_{it}u_{it}}}{\left(\frac{\partial\varphi_{t}}{\partial y_{it}}\right)^{2}+\frac{\partial^{2}\varphi_{t}}{\partial y_{it}^{2}}}
\end{align*}

Finally, we derive the derivative of $\phi_{t}$ with respect to $u_{it}$:
\begin{align*}
\frac{\partial\text{\ensuremath{\phi}}_{t}\left(x_{it},z_{it},u_{it}\right)}{\partial u_{it}} & =\frac{\partial\varphi_{t}}{\partial y_{it}}\frac{\partial\mathbb{M}_{t}^{-1}}{\partial u_{it}}+\frac{\partial\varphi_{t}}{\partial u_{it}}\\
 & =\left(\left(\frac{\partial\varphi_{t}}{\partial y_{it}}\right)^{2}+\frac{\partial^{2}\varphi_{t}}{\partial y_{it}^{2}}\right)^{-1}\left(\frac{\partial\varphi_{t}}{\partial u_{it}}\frac{\partial^{2}\varphi_{t}}{\partial y_{it}^{2}}-\frac{\partial\varphi_{t}}{\partial y_{it}}\frac{\partial^{2}\varphi_{t}}{\partial y_{it}\partial u_{it}}\right)\\
 & =\left(\left(\frac{\partial\varphi_{t}}{\partial y_{it}}\right)^{2}+\frac{\partial^{2}\varphi_{t}}{\partial y_{it}^{2}}\right)^{-1}\left(1-\frac{\partial\varphi_{t}}{\partial y_{it}}\right)^{2}\left(\frac{\partial\varphi_{t}}{\partial u_{it}}\frac{\partial\sigma_{t}}{\partial y_{it}}-\frac{\partial\varphi_{t}}{\partial y_{it}}\frac{\partial\sigma_{t}}{\partial u_{it}}\right)
\end{align*}
The assumption$\partial\mathbb{M}_{t}/\partial\omega_{it}>0$ and
$\partial f_{t}/\partial m_{it}>0$ implies $\left(\frac{\partial\varphi_{t}}{\partial y_{it}}\right)^{2}+\frac{\partial^{2}\varphi_{t}}{\partial y_{it}^{2}}>0$.
Therefore, we have 
\[
\frac{\partial\text{\ensuremath{\phi}}_{t}\left(x_{it},z_{it},\epsilon_{it}\right)}{\partial u_{it}}>0\Leftrightarrow\frac{\partial\varphi_{t}}{\partial u_{it}}\frac{\partial\sigma_{t}}{\partial y_{it}}>\frac{\partial\varphi_{t}}{\partial y_{it}}\frac{\partial\sigma_{t}}{\partial u_{it}}.
\]

\subsection{ Proof of Proposition \ref{P-3}}

\label{subsec:-proof-Pro6}

Suppose that Assumption \ref{A-period}(a) holds. Let $\text{var}(\eta_{t})$
and $\text{var}(\eta_{t+1})$ be the variance of $\eta_{t}$ and $\eta_{t+1}$
identified under the period-specific normalization in Assumption \ref{A-2}
for $t$ and $t+1$, respectively. From (\ref{eq:model}) and (\ref{eq:norm-t}),
$\text{var}(\eta_{t})=b_{t}^{2}\text{var}(\eta_{t}^{*})$ and $\text{var}(\eta_{t+1})=b_{t+1}^{2}\text{var}(\eta_{t+1}^{*})$.
From $\text{var}(\eta_{t}^{*})=\text{var}(\eta_{t+1}^{*})$, $b_{t+1}/b_{t}$
is identified as $b_{t+1}/b_{t}=\sqrt{{\text{var}(\eta_{t+1})}/{\text{var}(\eta_{t})}}$.

Let $\partial f_{t}(x_{t},z_{t}^{s})/\partial q$ and $\partial f_{t+1}(x_{t+1},z_{t+1}^{s})/\partial q$
be those elasticities identified under the period-specific normalization
in Assumption \ref{A-2} for $t$ and $t+1$, respectively, and $\partial f_{t}^{*}(x_{t},z_{t}^{s})/\partial q$
and $\partial f_{t+1}^{*}(x_{t+1},z_{t+1}^{s})/\partial q$ be the
true elasticities. From (\ref{eq:norm-t}), $\partial f_{t}(x_{t},z_{t}^{s})/\partial q=b_{t}\partial f_{t}^{*}(x_{t},z_{t}^{s})/\partial q$
and $\partial f_{t+1}(x_{t+1},z_{t+1}^{s})/\partial q=b_{t+1}\partial f_{t+1}^{*}(x_{t+1},z_{t+1}^{s})/\partial q$
hold.

Suppose that Assumption \ref{A-period}(b) holds. Then, $\partial f_{t}^{*}(x,z^{s})/\partial q=\partial f_{t+1}^{*}(x,z^{s})/\partial q$
for some input $q\in\{m,k,l\}$ and $x\in\mathcal{B}$. Then, $b_{t+1}/b_{t}$
is identified as $b_{t+1}/b_{t}=(\partial f_{t+1}(x,z^{s})/\partial q)/(\partial f_{t}(x,z^{s})/\partial q)$
for $x\in\mathcal{B}$.

Suppose that Assumption \ref{A-period}(c) holds, implying 
\[
1=\frac{\partial f_{t+1}^{*}(x,z^{s})/\partial m+\partial f_{t+1}^{*}(x,z^{s})/\partial k+\partial f_{t+1}^{*}(x,z^{s})/\partial l}{\partial f_{t}^{*}(x,z^{s})/\partial m+\partial f_{t}^{*}(x,z^{s})/\partial k+\partial f_{t}^{*}(x,z^{s})/\partial l}\,\text{for }\ensuremath{x\in\mathcal{B}}.
\]
Then, $b_{t+1}/b_{t}$ is identified as 
\[
\frac{b_{t+1}}{b_{t}}=\frac{\partial f_{t+1}(x,z^{s})/\partial m+\partial f_{t+1}(x,z^{s})/\partial k+\partial f_{t+1}(x,z^{s})/\partial l}{\partial f_{t}(x,z^{s})/\partial m+\partial f_{t}(x,z^{s})/\partial k+\partial f_{t}(x,z^{s})/\partial l}\quad\text{for \ensuremath{x\in\mathcal{B}}}.\qed
\]

\subsection{ Proof of Proposition \ref{P-location}}

\label{subsec:-Proof-pro8}

Let $\tilde{p}_{it}:=r_{it}-\tilde{\varphi}_{t}^{-1}(r_{it},z_{it}^{d},u_{it})$
and $\tilde{y}_{it}:=\tilde{\varphi}_{t}^{-1}(r_{it},z_{it}^{d},u_{it})$
be an output price and an output quantity identified under the normalization
in (\ref{eq:norm-tilde}) and Assumption \ref{A-2}, respectively.
Using these, we calculate an industry-level producer price index with
them: 
\[
P_{t}:=\frac{\sum_{i\in\tilde{N}}\exp(\tilde{p}_{it}+\tilde{y}_{i0})}{\sum_{i\in\tilde{N}}\exp(\tilde{p}_{i0}+\tilde{y}_{i0})}.
\]
From (\ref{eq:norm-t2}) and (\ref{eq:PPI}), $P_{t}$ is written
as 
\begin{align*}
P_{t} & =\frac{\sum_{i\in\tilde{N}}\exp(-(\tilde{a}_{1t}+\tilde{a}_{2t})+p_{it}^{*}+\tilde{a}_{1,0}+\tilde{a}_{2,0}+y_{i0}^{*})}{\sum_{i\in\tilde{N}}\exp(p_{i0}^{*}+y_{i0}^{*})}=\exp(\tilde{a}_{1,0}+\tilde{a}_{2,0}-(\tilde{a}_{1t}+\tilde{a}_{2t}))P_{t}^{*}.
\end{align*}
Therefore, $\tilde{a}_{1t+1}+\tilde{a}{}_{2t+1}-\tilde{a}_{1t}-\tilde{a}{}_{2t}$
is identified as: 
\begin{equation}
\tilde{a}_{1t+1}+\tilde{a}{}_{2t+1}-\tilde{a}_{1t}-\tilde{a}{}_{2t}=\ln P_{t+1}^{*}-\ln P_{t+1}-\left(\ln P_{t}^{*}-\ln P_{t}\right)\label{eq:a_diff}
\end{equation}
From (\ref{eq:growth}), we identify the output growth rate $\varphi_{t+1}^{*-1}(r_{it+1},z_{it+1}^{d},u_{it+1})-\varphi_{t}^{*-1}(r_{it},z_{it}^{d},u_{it})$.

Evaluating the second equation in (\ref{eq:growth}) at $x_{t+1}=x_{t}=\bar{x}$
and $z_{t+1}^{s}=z_{t}^{s}=\bar{z}^{s}$in Assumption \ref{A-location}(b),
we identify $\tilde{a}_{1t+1}-\tilde{a}_{1t}$ as: 
\begin{align*}
\tilde{a}_{1t+1}-\tilde{a}_{1t} & =\tilde{a}_{1,t+1}+f_{t+1}^{*}(\bar{x},\bar{z}^{s})-\left(\tilde{a}_{1,t}+f_{t}^{*}(\bar{x},\bar{z}^{s})\right)=\tilde{f}_{t+1}(\bar{x},\bar{z}^{s})-\tilde{f}_{t}(\bar{x},\bar{z}^{s}).
\end{align*}
From (\ref{eq:a_diff}), $\tilde{a}_{2t+1}-\tilde{a}_{2t}$ is also
identified as 
\[
\tilde{a}_{2t+1}-\tilde{a}_{2t}=\ln P_{t+1}^{*}-\ln P_{t+1}-\left(\ln P_{t}^{*}-\ln P_{t}\right)-\left(\tilde{f}_{t+1}(\bar{x},\bar{z}^{s})-\tilde{f}_{t}(\bar{x},\bar{z}^{s})\right).
\]
Therefore, from (\ref{eq:growth}), the true TFP growth rate $\omega_{it+1}^{*}-\omega_{it}^{*}$
is also identified. \qed 

\subsection{Derivations of Equilibrium Conditions for MCE and MCPE}

\label{subsec:Derivations-of-Equilibrium}

\paragraph{General HSA Demand System}

\subparagraph{MCE}

Define the inverse production function $m_{it}=\chi_{it}\left(y_{it}\right)$
such that $y_{it}=f_{t}\left(\chi_{it}\left(y_{it}\right),k_{it},l_{it},z_{it}^{s}\right)+\omega_{it}$
for given $(k_{it},l_{it},\omega_{it},z_{it}^{s})$; namely, $\chi_{it}(y_{it}):=f_{t}^{-1}\left(y_{it}-\omega_{it},k_{it},l_{it},z_{it}^{s}\right)$.
By using this, we rewrite the profit maximization problem with respect
to $m_{it}$ 
\[
\max_{m}\exp\left(\Phi_{t}+\mathfrak{s}_{t}\left(f_{t}\left(m,k_{it},l_{it},z_{it}^{s}\right)+\omega_{it}-\Delta q_{t}^{m},z_{it}^{d},u_{it}\right)\right)-\exp\left(p_{t}^{m}+m\right)
\]
 as the problem with respect to $y_{it}$:
\[
\max_{y}\exp\left(\Phi_{t}+\mathfrak{s}_{t}\left(y-\Delta q_{t}^{m},z_{it}^{d},u_{it}\right)\right)-\exp\left(p_{t}^{m}+\chi_{it}(y)\right)
\]

The first order condition is 
\[
\exp\left(\Phi_{t}+\mathfrak{s}_{t}\left(y_{it}^{m}-\Delta q_{t}^{m},z_{it}^{d},u_{it}\right)\right)\frac{\partial\mathfrak{s}_{t}\left(y_{it}^{m}-\Delta q_{t}^{m},z_{it}^{d},u_{it}\right)}{\partial y_{it}}=\exp\left(p_{t}^{m}+\chi_{it}(y_{it}^{m})\right)\frac{\partial\chi_{it}(y_{it}^{m})}{\partial y_{it}}.
\]
Noting that the marginal cost is 
\[
\frac{\partial\exp\left(p_{t}^{m}+\chi_{it}(y_{it}^{m})\right)}{\partial\exp(y_{it})}=\frac{\exp\left(p_{t}^{m}+\chi_{it}\left(y_{it}^{m}\right)\right)}{\exp\left(y_{it}^{m}\right)}\frac{\partial\chi_{it}(y_{it}^{m})}{\partial y_{it}},
\]
the system of the first order conditions and the market share condition
becomes:
\begin{align*}
\underbrace{\exp\left(\mathfrak{s}_{t}\left(y_{it}^{m}-\Delta q_{t}^{m},z_{it}^{d},u_{it}\right)+\Phi_{t}-y_{it}^{m}\right)\frac{\partial\mathfrak{s}_{t}\left(y_{it}^{m}-\Delta q_{t}^{m},z_{it}^{d},u_{it}\right)}{\partial y_{it}}}_{\text{Marginal Revenue}} & =\underbrace{\frac{\exp\left(p_{t}^{m}+\chi_{it}\left(y_{it}^{m}\right)\right)}{\exp\left(y_{it}^{m}\right)}\frac{\partial\chi_{it}(y_{it}^{m})}{\partial y_{it}}}_{\text{Marginal Cost}},\\
\sum_{i=1}^{N_{t}}\exp\left(\mathfrak{s}_{t}\left(y_{it}^{m}-\Delta q_{t}^{m},z_{it}^{d},u_{it}\right)\right) & =1.
\end{align*}

\subparagraph{MCPE}

The profit maximization problem with respect to $y_{it}$ is 
\[
\max_{y}\exp\left(p_{it}^{c}+y\right)-\exp\left(p_{t}^{m}+\chi_{it}\left(y\right)\right)
\]
where the firm takes $p_{it}^{c}$ as given. The first order condition
is 

\[
\exp\left(p_{it}^{c}+y_{it}^{c}\right)=\exp\left(p_{t}^{m}+\chi_{it}\left(y_{it}^{c}\right)\right)\frac{\partial\chi_{it}(y_{it}^{c})}{\partial y_{it}}
\]

From $p_{it}^{c}=\Phi_{t}^{c}+\mathfrak{s}_{t}\left(y_{it}^{c}-\Delta q_{t}^{c},z_{it}^{d},u_{it}\right)-y_{it}^{c}$
in equilibrium, the system of the first order conditions and the market
share condition for the MCPE is 
\begin{align*}
\underbrace{\exp\left(\mathfrak{s}_{t}\left(y_{it}^{c}-\Delta q_{t}^{c},z_{it}^{d},u_{it}\right)+\Phi_{t}-y_{it}^{c}\right)}_{\text{Price}} & =\underbrace{\frac{\exp\left(p_{t}^{m}+\chi_{it}\left(y_{it}^{c}\right)\right)}{\exp\left(y_{it}^{c}\right)}\frac{\partial\chi_{it}(y_{it}^{c})}{\partial y_{it}}}_{\text{Marginal Cost}},\\
\sum_{i=1}^{N_{t}}\exp\left(\mathfrak{s}_{t}\left(y_{it}^{c}-\Delta q_{t}^{c},z_{it}^{d},u_{it}\right)\right) & =1,
\end{align*}

\paragraph{CoPaTh-HSA Demand System and Cobb-Douglas Production Function}

Consider the CoPaTh-HSA Demand System (\ref{eq:revenue1}):
\[
\mathfrak{s}_{t}\left(y_{it}-q_{t}\left(\mathbf{y}_{t},\mathbf{\boldsymbol{\epsilon}}_{t}\right),\epsilon_{it}\right)=\delta_{t}-\frac{1}{\beta_{t}}\ln\left(\frac{\exp\left(-\beta_{t}\left(y_{it}-q_{t}\left(\mathbf{y}_{t},\mathbf{\boldsymbol{\epsilon}}_{t}\right)\right)+\gamma_{t}\right)+\epsilon_{it}}{1+\epsilon_{it}}\right).
\]

\subparagraph{MCE }

Note that $q_{t}\left(\mathbf{y}_{t}^{m},\mathbf{\boldsymbol{\epsilon}}_{t}\right)=0$
holds as normalization in the initial MCE. Since 
\begin{align*}
\frac{\partial\mathfrak{s}_{t}\left(y_{it}^{m},\epsilon_{it}\right)}{\partial y_{it}^{m}} & =\frac{\exp(-\beta_{t}y_{it}^{m}+\gamma_{t})}{\exp(-\beta_{t}y_{it}^{m}+\gamma_{t})+\epsilon_{it}}\\
\chi_{it}\left(y_{it}^{m}\right) & =\frac{y_{it}^{m}-\theta_{k}k_{it}-\theta_{l}l_{it}-\omega_{it}}{\theta_{m}},
\end{align*}
the first order condition (\ref{eq:MCE}) becomes 
\begin{align*}
 & \exp\left(\delta_{t}-\frac{1}{\beta_{t}}\ln\left(\frac{\exp(-\beta_{t}y_{it}^{m}+\gamma_{t})+\epsilon_{it}}{1+\epsilon_{it}}\right)+\Phi_{t}-y_{it}^{m}\right)\left(\frac{\exp(-\beta_{t}y_{it}^{m}+\gamma_{t})}{\exp(-\beta_{t}y_{it}^{m}+\gamma_{t})+\epsilon_{it}}\right)\\
= & \exp\left(p_{t}^{m}+\frac{y_{it}^{m}-\theta_{k}k_{it}-\theta_{l}l_{it}-\omega_{it}}{\theta_{m}}-y_{it}^{m}\right)\frac{1}{\theta_{m}}.
\end{align*}
Taking the log of both sides yields 
\begin{align*}
\delta_{t}-\frac{1}{\beta_{t}}\ln\left(\frac{\exp(-\beta_{t}y_{it}^{m}+\gamma_{t})+\epsilon_{it}}{1+\epsilon_{it}}\right)+\Phi_{t}-\beta_{t}y_{it}^{m}+\gamma_{t}\\
-\ln\left(\exp(-\beta_{t}y_{it}^{m}+\gamma_{t})+\epsilon_{it}\right)-p_{t}^{m}-\frac{y_{it}^{m}-\theta_{k}k_{it}-\theta_{l}l_{it}-\omega_{it}}{\theta_{m}} & +\ln\theta_{m}=0.
\end{align*}
Letting $\Xi_{it}:=\ln\theta_{m}+(\theta_{k}k_{it}+\theta_{l}l_{it}+\omega_{it})/\theta_{m}$
and $p_{t}^{m}=0$, the system of the first order conditions and the
market share condition is simplified as 
\begin{align*}
\Phi_{t}+\delta_{t}-\beta_{t}y_{it}^{m}+\gamma_{t}+\Xi_{it}-\frac{y_{it}^{m}}{\theta_{m}}+\frac{1}{\beta_{t}}\ln\left(1+\epsilon_{it}\right)\\
-\left(1+\frac{1}{\beta_{t}}\right)\ln\left(\exp(-\beta_{t}y_{it}^{m}+\gamma_{t})+\epsilon_{it}\right) & =0\text{ for }i=1,..,N_{t}\\
\sum_{i=1}^{N_{t}}\exp\left(\delta_{t}-\frac{1}{\beta_{t}}\ln\left(\frac{\exp(-\beta_{t}y_{it}^{m}+\gamma_{t})+\epsilon_{it}}{1+\epsilon_{it}}\right)\right) & =1.
\end{align*}

\subparagraph{MCPE}

The first order condition (\ref{eq:MCPE}) becomes

\begin{align*}
 & \exp\left(\delta_{t}-\frac{1}{\beta_{t}}\ln\left(\frac{\exp(-\beta_{t}\left(y_{it}^{c}-\Delta q_{t}^{c}\right)+\gamma_{t})+\epsilon_{it}}{1+\epsilon_{it}}\right)+\Phi_{t}-y_{it}^{c}\right)\\
= & \exp\left(p_{t}^{m}+\frac{y_{it}^{c}-\theta_{k}k_{it}-\theta_{l}l_{it}-\omega_{it}}{\theta_{m}}-y_{it}^{c}\right)\frac{1}{\theta_{m}}.
\end{align*}
Taking the log of both sides, the system of the first order conditions
and the market share condition is simplified as 

\begin{align*}
\Phi_{t}+\delta_{t}+\Xi_{it}-\frac{y_{it}^{c}}{\theta_{m}}-\frac{1}{\beta_{t}}\ln\left(\frac{\exp\left(-\beta_{t}\left(y_{it}^{c}-\Delta q_{t}^{c}\right)+\gamma_{t}\right)+\epsilon_{it}}{1+\epsilon_{it}}\right) & =0\text{ for }i=1,..,N_{t}\\
\sum_{i=1}^{N_{t}}\exp\left(\delta_{t}-\frac{1}{\beta_{t}}\ln\left(\frac{\exp\left(-\beta_{t}\left(y_{it}^{c}-\Delta q_{t}^{c}\right)+\gamma_{t}\right)+\epsilon_{it}}{1+\epsilon_{it}}\right)\right) & =1.
\end{align*}


\subsection{Proof for Proposition \ref{P-demand}}

\label{subsec:Demand}

The proof for Proposition \ref{P-demand} uses the following result
of \cite{matsuyama2017beyond}. 
\begin{thm}
\label{thm:1}(\citeauthor{matsuyama2017beyond}, 2017, Remark 3
and Proposition 1). Consider a mapping $\mathbf{S(Y)}:=(S_{1}(Y_{1}),...,S_{N}(Y_{N}))'$
from $\mathbb{R}_{+}^{N}$ to $\mathbb{R}_{+}^{N}$, which is differentiable
almost everywhere, is normalized by 
\begin{equation}
\sum_{i=1}^{N}S_{i}(Y_{i}^{*})=1,\label{eq:norm_s}
\end{equation}
for some point $\mathbf{Y}^{*}:=(Y_{1}^{*},...,Y_{N}^{*})$ and satisfies
the following conditions 
\begin{align}
S_{i}'(Y_{i})Y_{i} & <S_{i}(Y_{i})\text{ for }i=1,...,N,\nonumber \\
S_{i}'(Y_{i}) & S_{j}'(Y_{j})\ge0\text{ for }i,j=1,...,N,\label{eq:HSAcondition}
\end{align}
for all $\mathbf{Y}$ such that $\sum_{i=1}^{N}S_{i}(Y_{i})=1$. Then,
(a) for any such mapping, there exists a unique monotone, convex,
continuous, and homothetic rational preference that generates the
HSA demand system described by 
\begin{align}
P_{i} & =\frac{I}{Y_{i}}S_{i}\left(\frac{Y_{i}}{Q\left(\mathbf{Y}\right)}\right)\text{ for }i=1,..,N,\label{eq: demand-A1}
\end{align}
where $I:=\sum_{i=1}^{N}P_{i}Y_{i}$ and $Q(\mathbf{Y})$ is obtained
by solving 
\begin{equation}
\sum_{i=1}^{N}S_{i}\left(\frac{Y_{i}}{Q\left(\mathbf{Y}\right)}\right)=1.\label{eq: budget-A1}
\end{equation}
(b) This homothetic preference is described by a utility function
$U$ which is defined by 
\begin{equation}
\ln U(\mathbf{Y})=\ln Q(\mathbf{Y})+\sum_{i=1}^{N}\int_{c_{i}}^{Y_{i}/Q(Y)}\frac{S_{i}\left(\xi\right)}{\xi}d\xi,\label{eq:utility}
\end{equation}
where $\mathbf{c}=(c_{1},...,c_{N})$ is a vector of constants such
that $U(\mathbf{c})=1$. 
\end{thm}

\paragraph{Proof for Proposition \ref{P-demand}}
\begin{proof}
(a) For any $(\tilde{\mathbf{y}}_{t},\tilde{\mathbf{z}}_{t}^{d},\tilde{\mathbf{u}}_{t})$,
we identify $\Delta{q}_{t}(\tilde{\mathbf{y}}_{t},\tilde{\mathbf{z}}_{t}^{d},\tilde{\mathbf{u}}_{t})$
and $\mathfrak{s}_{t}^{*}(\tilde{y}_{it}-{q}_{t}(\tilde{\mathbf{y}}_{t},\tilde{\mathbf{z}}_{t}^{d},\tilde{\mathbf{u}}_{t}),\tilde{z}_{it}^{d},\tilde{u}_{it})$
from (\ref{eq: s-id})--(\ref{eq:market_share2}) as discussed in
the main text, where the identification of reduced-form functions
$\varphi_{t}\left(\cdot\right)$ and $\mathfrak{s}_{t}(\cdot)$ at
the baseline aggregate state $(\mathbf{y}_{t},\mathbf{z}_{t}^{d},\mathbf{u}_{t})$
in the data follows from Proposition \ref{P-CRS}. 

Fix $\tilde{\mathbf{z}}_{t}^{d}:=(\tilde{z}_{1t}^{d},\ldots,\tilde{z}_{Nt}^{d})$,
$\tilde{\mathbf{u}}_{t}:=(\tilde{u}_{1t},\ldots,\tilde{u}_{Nt})$,
and time $t$. For given $\tilde{\mathbf{Y}}\in\mathscr{Y}_{t}^{N}$,
define $\bs S_{t}(\tilde{\mathbf{Y}}):=\left(S_{1t}(\tilde{Y}_{1}),\ldots,S_{N_{t}t}(\tilde{Y}_{N_{t}})\right)$
with 
\begin{align}
S_{it}(\tilde{Y}_{it}) & :=\exp\left(\mathfrak{s}_{t}^{*}\left(\ln\tilde{Y}_{it}-q_{t}(\mathbf{y}_{t},\mathbf{z}_{t}^{d},\mathbf{u}_{t}),\tilde{z}_{it}^{d},\tilde{u}_{it}\right)\right)\label{eq:s-def}\\
 & =\exp\left(\mathfrak{s}_{t}\left(\ln\tilde{Y}_{it},\tilde{z}_{it}^{d},\tilde{u}_{it}\right)\right)\nonumber \\
 & =\exp\left(\varphi_{t}\left(\ln\tilde{Y}_{it},\tilde{z}_{it}^{d},\tilde{u}_{it}\right)-\Phi_{t}\right)\quad\text{for \ensuremath{i=1,...,N_{t}}},\nonumber 
\end{align}
where the second and third equalities follow from (\ref{eq:comparison})
and (\ref{eq:reduced-bs}), respectively. By definition, $\sum_{i=1}^{N_{t}}S_{it}(\tilde{Y}_{i})=1$
holds. Also, define 
\begin{align}
Q_{t}(\tilde{\mathbf{Y}}) & :=\exp\left(\Delta{q}_{t}(\ln\tilde{\mathbf{Y}},\tilde{\mathbf{z}}_{t}^{d},\tilde{\mathbf{u}}_{t})\right)\nonumber \\
 & =\exp\left({q}_{t}(\ln\tilde{\mathbf{Y}},\tilde{\mathbf{z}}_{t}^{d},\tilde{\mathbf{u}}_{t})\right)\label{eq:q-def}
\end{align}
where the second equality follows from the normalization $q_{t}(\mathbf{y}_{t},\mathbf{z}_{t}^{d},\mathbf{u}_{t})=0$
at the baseline aggregate state $(\mathbf{y}_{t},\mathbf{z}_{t}^{d},\mathbf{u}_{t})$.


From Assumption \ref{A-1}(b) and $\tilde{y}:=\ln\tilde{Y}$, 
\[
0<\frac{\partial\varphi_{t}\left(\ln\tilde{Y},\tilde{z}_{it}^{d},\tilde{u}_{it}\right)}{\partial\ln\tilde{Y}}=1+\frac{\partial\psi_{t}\left(\ln\tilde{Y},\tilde{z}_{it}^{d},\tilde{u}_{it}\right)}{\partial\ln\tilde{Y}}<1
\]
holds for all $i$ and $\tilde{y}$. The above inequality implies
\[
S_{it}'(\tilde{Y})>0\text{ and }S_{it}'(\tilde{Y})\tilde{Y}<S_{it}(\tilde{Y})\text{ for all }i\text{ and }\tilde{Y}
\]
because 
\begin{align*}
S_{it}'(\tilde{Y})\tilde{Y} & =\exp\left(\varphi_{t}\left(\ln\tilde{Y},\tilde{z}_{it}^{d},\tilde{u}_{it}\right)-\Phi_{t}\right)\frac{\partial\varphi_{t}\left(\ln\tilde{Y},\tilde{z}_{it}^{d},\tilde{u}_{it}\right)}{\partial\ln\tilde{Y}}\\
 & =S_{it}(\tilde{Y})\frac{\partial\varphi_{t}\left(\ln\tilde{Y},\tilde{z}_{it}^{d},\tilde{u}_{it}\right)}{\partial\ln\tilde{Y}}.
\end{align*}
Consequently, $\mathbf{S(\tilde{Y})}$ satisfies the inequalities
in (\ref{eq:HSAcondition}) for all $\tilde{\mathbf{Y}}$ satisfying
$\sum_{i=1}^{N_{t}}S_{it}(\tilde{Y}_{i})=1$. 

Therefore, from Theorem \ref{thm:1}(a), there exists a unique monotone,
convex, continuous, and homothetic rational preference that generates
the following HSA demand system:
\begin{align*}
\tilde{P}_{it} & =\frac{I_{t}}{\tilde{Y}_{it}}S_{it}\left(\frac{\tilde{Y}_{it}}{Q_{t}\left(\mathbf{\tilde{Y}}_{t}\right)}\right)\text{ for }i=1,..,N,\\
\sum_{i=1}^{N} & S_{it}\left(\frac{\tilde{Y}_{it}}{Q_{t}\left(\mathbf{\tilde{Y}}_{t}\right)}\right)=1
\end{align*}
where $\tilde{P}_{it}\tilde{Y}_{it}=\varphi_{t}\left(\ln\tilde{Y}_{it},\tilde{z}_{it}^{d},\tilde{u}_{it}\right)$
and $I_{t}=\sum_{i=1}^{N}\tilde{P}_{i}\tilde{Y}_{it}$ is the consumer's
budget. Taking the log of this demand system and using (\ref{eq:s-def})
and (\ref{eq:q-def}) with $\tilde{y}_{it}:=\ln\tilde{Y}_{it}$, we
obtain (\ref{eq:structural-bs}) and (\ref{eq:market_share_constraint})
as
\begin{align*}
\ln\left(\frac{\exp\left(\tilde{r}_{it}\right)}{\sum_{j=1}^{N_{t}}\exp\left(\tilde{r}_{jt}\right)}\right) & =\mathfrak{s}_{t}^{*}\left(\tilde{y}_{it}-q_{t}(\mathbf{\tilde{y}}_{t},\mathbf{\tilde{z}}_{t}^{d},\mathbf{\tilde{u}}_{t}),\tilde{z}_{it}^{d},\tilde{u}_{it}\right)\\
1 & =\sum_{i=1}^{N_{t}}\mathfrak{s}_{t}^{*}\left(\tilde{y}_{it}-q_{t}(\mathbf{\tilde{y}}_{t},\mathbf{\tilde{z}}_{t}^{d},\mathbf{\tilde{u}}_{t}),\tilde{z}_{it}^{d},\tilde{u}_{it}\right)
\end{align*}
where $\tilde{r}_{it}:=\varphi_{t}\left(\ln\tilde{Y}_{it},\tilde{z}_{it}^{d},\tilde{u}_{it}\right)$.

(b) Fix $\tilde{\mathbf{z}}_{t}^{d}:=(\tilde{z}_{1t}^{d},\ldots,\tilde{z}_{Nt}^{d})$,
$\tilde{\mathbf{u}}_{t}:=(\tilde{u}_{1t},\ldots,\tilde{u}_{Nt})$,
and time $t$. For $\tilde{\mathbf{Y}}_{t}\in\mathscr{Y}_{t}^{N}$
and $\tilde{\mathbf{y}}_{t}=\ln\tilde{\mathbf{Y}}_{t}$, let $\bar{U_{t}}(\tilde{\mathbf{Y}}_{t}):=U_{t}(\tilde{\mathbf{y}}_{t},\tilde{\mathbf{z}}_{t}^{d},\tilde{\mathbf{u}}_{t})$
be the utility function of the representative consumer. From Theorem
\ref{thm:1}(b) and using (\ref{eq:s-def}) and (\ref{eq:q-def}),
the homothetic preference is described as 
\begin{align*}
\ln\bar{U}_{t}(\tilde{\mathbf{Y}}_{t}) & =\ln Q_{t}(\tilde{\mathbf{Y}}_{t})+\sum_{i=1}^{N}\int_{c_{i}}^{\tilde{Y}_{it}/Q(\tilde{\mathbf{Y}}_{t})}\frac{S_{it}\left(\xi\right)}{\xi}d\xi.\\
 & =\Delta{q}_{t}(\ln\tilde{\mathbf{Y}}_{t},\tilde{\mathbf{z}}_{t}^{d},\tilde{\mathbf{u}}_{t})+\sum_{i=1}^{N}\int_{c_{i}}^{\tilde{Y}_{i}/\exp\left(\Delta{q}_{t}(\ln\tilde{\mathbf{Y}}_{t},\tilde{\mathbf{z}}_{t}^{d},\tilde{\mathbf{u}}_{t})\right)}\frac{\exp\left(\mathfrak{s}_{t}\left(\ln\xi,\tilde{z}_{it}^{d},\tilde{u}_{it}\right)\right)}{\xi}d\xi.
\end{align*}
 Applying a change in variable $\zeta=\ln\xi$ and $d\zeta=\frac{d\xi}{\xi}$,
we write 
\[
\ln U_{t}(\tilde{\mathbf{y}}_{t},\tilde{\mathbf{z}}_{t}^{d},\tilde{\mathbf{u}}_{t})=\Delta{q}_{t}(\tilde{\mathbf{y}}_{t},\tilde{\mathbf{z}}_{t}^{d},\tilde{\mathbf{u}}_{t})+\sum_{i=1}^{N}\int_{b_{i}}^{\tilde{y}_{it}-\Delta{q}_{t}(\tilde{\mathbf{y}}_{t},\tilde{\mathbf{z}}_{t}^{d},\tilde{\mathbf{u}}_{t})}\exp\left(\mathfrak{s}_{t}\left(\zeta,\tilde{z}_{it}^{d},\tilde{u}_{it}\right)\right)d\zeta,
\]
where $b_{i}=\ln c_{i}$ such that $\ln U_{t}(\mathbf{b},\tilde{\mathbf{z}}_{t}^{d},\tilde{\mathbf{u}}_{t})=0$
for $\mathbf{b}=(b_{1},...,b_{N_{t}})$ from Theorem \ref{thm:1}(b).

(c) The homothetic preference implies that the market share $P_{it}Y_{it}/I_{t}$
depends only on a price vector and is independent of income. Suppose
that $\{\tilde{y}_{it}\}_{i=1}^{N_{t}}$ be a log consumption vector
for given price and income. If the price remains the same and the
income is multiplied by $\exp(\lambda)$, then the log-consumption
vector becomes $\{\tilde{y}_{it}+\gamma\}_{i=1}^{N_{t}}$ to keep
the same market share for each product. This property implies that
\begin{align*}
1 & =\sum_{i=1}^{N_{t}}\exp\left(\mathfrak{s}_{t}(\tilde{y}_{it}-\Delta q_{t}(\tilde{\mathbf{y}}_{t},\tilde{\mathbf{z}}_{t}^{d},\tilde{\mathbf{u}}_{t}),\tilde{z}_{it}^{d},\tilde{u}_{it}\right)\\
 & =\sum_{i=1}^{N_{t}}\exp\left(\mathfrak{s}_{t}(\tilde{y}_{it}+\gamma-\Delta q_{t}(\tilde{\mathbf{y}}_{t}+\gamma,\tilde{\mathbf{z}}_{t}^{d},\tilde{\mathbf{u}}_{t}),\tilde{z}_{it}^{d},\tilde{u}_{it}\right).
\end{align*}
On the other hand, from 
\begin{align*}
1 & =\sum_{i=1}^{N_{t}}\exp\left(\mathfrak{s}_{t}(\tilde{y}_{it}-\Delta q_{t}(\tilde{\mathbf{y}}_{t},\tilde{\mathbf{z}}_{t}^{d},\tilde{\mathbf{u}}_{t}),\tilde{z}_{it}^{d},\tilde{u}_{it}\right)\\
 & =\sum_{i=1}^{N_{t}}\exp\left(\mathfrak{s}_{t}(\tilde{y}_{it}+\gamma-\Delta q_{t}(\tilde{\mathbf{y}}_{t},\tilde{\mathbf{z}}_{t}^{d},\tilde{\mathbf{u}}_{t})-\gamma,\tilde{z}_{it}^{d},\tilde{u}_{it}\right),
\end{align*}
we have $\Delta q_{t}(\tilde{\mathbf{y}}_{t}+\gamma,\tilde{\mathbf{z}}_{t}^{d},\tilde{\mathbf{u}}_{t})=\Delta q_{t}(\tilde{\mathbf{y}}_{t},\tilde{\mathbf{z}}_{t}^{d},\tilde{\mathbf{u}}_{t})+\gamma$
for any constant $\gamma$. Since the output $\tilde{y}_{it}=\varphi_{t}^{-1}(r_{it},\tilde{z}_{it}^{d},\tilde{u}_{it})$
is identified up to location, there is $a\in\mathbb{R}$ such that
$\tilde{y}_{it}=a+\tilde{y}_{it}^{*}$ where $\tilde{y}_{it}^{*}$
is the true output. Note that 
\begin{align*}
\tilde{y}_{it}-\Delta q_{t}(\tilde{\mathbf{y}}_{t},\tilde{\mathbf{z}}_{t}^{d},\tilde{\mathbf{u}}_{t}) & =a+\tilde{y}_{it}^{*}-\Delta q_{t}(a+\mathbf{y}_{t}^{*},\tilde{\mathbf{z}}_{t}^{d},\tilde{\mathbf{u}}_{t})\\
 & =a+\tilde{y}_{it}^{*}-\Delta q_{t}(\mathbf{y}_{t}^{*},\tilde{\mathbf{z}}_{t}^{d},\tilde{\mathbf{u}}_{t})-a\\
 & =\tilde{y}_{it}^{*}-\Delta q_{t}(\mathbf{y}_{t}^{*},\tilde{\mathbf{z}}_{t}^{d},\tilde{\mathbf{u}}_{t}).
\end{align*}
The utility is expressed as: 
\begin{align*}
\ln U_{t}(\tilde{\mathbf{y}}_{t},\tilde{\mathbf{z}}_{t}^{d},\tilde{\mathbf{u}}_{t}) & =\Delta q_{t}(\tilde{\mathbf{y}}_{t},\tilde{\mathbf{z}}_{t}^{d},\tilde{\mathbf{u}}_{t})+\sum_{i=1}^{N}\int_{\ln c_{i}}^{\tilde{y}_{it}-\Delta q_{t}(\tilde{\mathbf{y}}_{t},\tilde{\mathbf{z}}_{t}^{d},\tilde{\mathbf{u}}_{t})}\exp\left(\mathfrak{s}_{t}\left(\zeta,\tilde{z}_{it}^{d},\tilde{u}_{it}\right)\right)d\zeta.\\
 & =\Delta q_{t}(\mathbf{y}_{t}^{*}+a,\tilde{\mathbf{z}}_{t}^{d},\tilde{\mathbf{u}}_{t})+\sum_{i=1}^{N}\int_{\ln c_{i}}^{\tilde{y}_{it}^{*}-\Delta q_{t}(\mathbf{y}_{t}^{*},\tilde{\mathbf{z}}_{t}^{d},\tilde{\mathbf{u}}_{t})}\exp\left(\mathfrak{s}_{t}\left(\zeta,\tilde{z}_{it}^{d},\tilde{u}_{it}\right)\right)d\zeta\\
 & =a+\Delta q_{t}(\mathbf{y}_{t}^{*},\tilde{\mathbf{z}}_{t}^{d},\tilde{\mathbf{u}}_{t})+\sum_{i=1}^{N}\int_{\ln c_{i}}^{\tilde{y}_{it}^{*}-\Delta q_{t}(\mathbf{y}_{t}^{*},\tilde{\mathbf{z}}_{t}^{d},\tilde{\mathbf{u}}_{t})}\exp\left(\mathfrak{s}_{t}\left(\zeta,\tilde{z}_{it}^{d},\tilde{u}_{it}\right)\right)d\zeta\\
 & =a+\ln U(\mathbf{y}_{t}^{*},\tilde{\mathbf{z}}_{t}^{d},\tilde{\mathbf{u}}_{t}).
\end{align*}
Therefore, the log utility function is identified up to the location
normalization of $\varphi_{t}^{-1}(\cdot)$. The identified utility
function is a monotonic transformation of the true utility function,
which implies both utility functions represent the same consumer preference. 
\end{proof}

\subsection{Derivation of the Control Function under a Production Function that is Separable
in Materials}
\label{subsec:Derivation-of-the}

While the main text focuses on the Cobb--Douglas production function,
here we consider a more general production function that is separable
in materials:
\[
f_{t}(m_{it},k_{it},l_{it},z_{it}^{s})
=
f_{1t}\left(m_{it}\right)
+
f_{2t}\left(k_{it},l_{it},z_{it}^{s}\right).
\]

Define a composite variable
\[
\Omega_{it}
\equiv
f_{2t}\left(k_{it},l_{it},z_{it}^{s}\right)
+
\omega_{it}.
\]
The first-order condition for problem~(\ref{eq:profit_maximization}) can
then be written as
\[
\exp\!\left(
\varphi_{t}\!\left(
f_{1t}\left(m_{it}\right)+\Omega_{it},
z_{it}^{d},
u_{it}
\right)
\right)
\frac{
\partial
\varphi_{t}\!\left(
f_{1t}\left(m_{it}\right)+\Omega_{it},
z_{it}^{d},
u_{it}
\right)
}{
\partial y_{it}
}
f_{1t}'\!\left(m_{it}\right)
=
\exp\!\left(p_{t}^{m}+m_{it}\right).
\]

This condition implicitly defines the material demand function
\[
m_{it}
=
\mathbb{M}_{t}\!\left(
\Omega_{it},
z_{it}^{d},
u_{it}
\right).
\]

Since $\mathbb{M}_{t}\!\left(\cdot,z_{it}^{d},u_{it}\right)$ is strictly
increasing in $\Omega_{it}$ and therefore invertible. Its inverse is
given by
\[
\Omega_{it}
=
\lambda_{t}\!\left(
m_{it},
z_{it}^{d},
u_{it}
\right).
\]

The control function for $\omega_{it}$ is expressed as
\[
\omega_{it}
=
\lambda_{t}\!\left(
m_{it},
z_{it}^{d},
u_{it}
\right)
-
f_{2t}\!\left(
k_{it},
l_{it},
z_{it}^{s}
\right).
\]

Substituting this expression into the revenue function yields
\begin{align*}
\varphi_{t}\!\left(
f_{t}(m_{it},k_{it},l_{it},z_{it}^{s})
+
\omega_{it},
z_{it}^{d},
u_{it}
\right)
&=
\varphi_{t}\!\left(
f_{1t}(m_{it})+
\lambda_{t}\!\left(
m_{it},
z_{it}^{d},
u_{it}
\right),
z_{it}^{d},
u_{it}
\right) \\
&\equiv
\phi_{t}\!\left(
m_{it},
z_{it}^{d},
u_{it}
\right).
\end{align*}

\section{Observational Equivalence and Identification up to Location and Scale}
\label{app:equivalence_class}

\begin{prop}[Observational Equivalence and Identification up to Location and Scale]
\label{prop:equivalence}
Suppose Assumptions \ref{A-0}--\ref{A-data} hold. Let $\{\varphi_{t}^{*-1}(\cdot), f_{t}^{*}(\cdot), \mathbb{M}_{t}^{*-1}(\cdot)\}$ denote the true model structure satisfying 
\begin{equation}
\varphi_{t}^{-1}(r_{it}, z_{it}^{d}, u_{it}) = f_{t}(x_{it}, z_{it}^{s}) + \mathbb{M}_{t}^{-1}(m_{it}, w_{it}, u_{it}).
\label{eq:model_step2_app}
\end{equation}
Define the equivalence class
\begin{align*}
\mathcal{E}_t := \Big\{ &\left(\varphi_t^{-1}, f_t, \mathbb{M}_t^{-1}\right) : \exists\, (a_{1t}, a_{2t}, b_t) \in \mathbb{R}^2 \times \mathbb{R}_{++} \text{ such that} \\
&\varphi_{t}^{-1}(r, z^d, u) = (a_{1t} + a_{2t}) + b_t \varphi_{t}^{*-1}(r, z^d, u), \\
&f_{t}(x, z^s) = a_{1t} + b_t f_{t}^{*}(x, z^s), \\
&\mathbb{M}_{t}^{-1}(m, w, u) = a_{2t} + b_t \mathbb{M}_{t}^{*-1}(m, w, u) \Big\}.
\end{align*}
Then:
\begin{enumerate}
    \item[(i)] Every element $(\varphi_t^{-1}, f_t, \mathbb{M}_t^{-1}) \in \mathcal{E}_t$ satisfies \eqref{eq:model_step2_app} and generates the same population joint distribution of $\{r_{is}, m_{is}, v_{is}\}_{s=t-\upsilon-2}^{t}$ as the true structure.
    
    \item[(ii)] Conversely, if a structure $(\tilde{\varphi}_t^{-1}, \tilde{f}_t, \tilde{\mathbb{M}}_t^{-1})$ satisfies \eqref{eq:model_step2_app} and generates the same population joint distribution of $\{r_{is}, m_{is}, v_{is}\}_{s=t-\upsilon-2}^{t}$, then $(\tilde{\varphi}_t^{-1}, \tilde{f}_t, \tilde{\mathbb{M}}_t^{-1}) \in \mathcal{E}_t$.
    
    \item[(iii)] The equivalence class $\mathcal{E}_t$ is exactly identified by the population joint distribution of $\{r_{is}, m_{is}, v_{is}\}_{s=t-\upsilon-2}^{t}$.
\end{enumerate}
\end{prop}

\begin{proof}
We prove each part in turn.

\medskip
\noindent\textbf{Proof of (i).} 
Let $(\varphi_t^{-1}, f_t, \mathbb{M}_t^{-1}) \in \mathcal{E}_t$ with associated constants $(a_{1t}, a_{2t}, b_t)$. We verify that this structure satisfies \eqref{eq:model_step2_app} and generates the same joint distribution of observables.

First, we verify that \eqref{eq:model_step2_app} holds. By definition of $\mathcal{E}_t$,
\begin{align*}
f_t(x_{it}, z_{it}^s) + \mathbb{M}_t^{-1}(m_{it}, w_{it}, u_{it}) 
&= \left[a_{1t} + b_t f_t^*(x_{it}, z_{it}^s)\right] + \left[a_{2t} + b_t \mathbb{M}_t^{*-1}(m_{it}, w_{it}, u_{it})\right] \\
&= (a_{1t} + a_{2t}) + b_t \left[f_t^*(x_{it}, z_{it}^s) + \mathbb{M}_t^{*-1}(m_{it}, w_{it}, u_{it})\right] \\
&= (a_{1t} + a_{2t}) + b_t \varphi_t^{*-1}(r_{it}, z_{it}^d, u_{it}) \\
&= \varphi_t^{-1}(r_{it}, z_{it}^d, u_{it}),
\end{align*}
where the third equality uses the fact that the true structure satisfies \eqref{eq:model_step2_app}.

Second, we show that the joint distribution of observables is unchanged. Note that $(r_{it}, m_{it}, v_{it})$ are directly observed. The equivalence class transformation redefines the latent output as $y_{it} = \varphi_t^{-1}(r_{it}, z_{it}^d, u_{it}) = (a_{1t} + a_{2t}) + b_t y_{it}^*$, where $y_{it}^* = \varphi_t^{*-1}(r_{it}, z_{it}^d, u_{it})$ is the latent output under the true structure. Since this is a one-to-one transformation of the latent variable $y_{it}^*$ that does not alter the relationship between observables $(r_{it}, z_{it}^d)$ and the normalized residual $u_{it}$, the conditional distribution $F_{r|z^d,m,w}$ remains unchanged. Consequently, the joint distribution of all observed variables $\{r_{is}, m_{is}, v_{is}\}_{s=t-\upsilon-2}^{t}$ is identical under both structures.

\medskip
\noindent\textbf{Proof of (ii).} 
Suppose $(\tilde{\varphi}_t^{-1}, \tilde{f}_t, \tilde{\mathbb{M}}_t^{-1})$ satisfies \eqref{eq:model_step2_app} and generates the same population joint distribution as the true structure. We show that there exist constants $(a_{1t}, a_{2t}, b_t) \in \mathbb{R}^2 \times \mathbb{R}_{++}$ such that $(\tilde{\varphi}_t^{-1}, \tilde{f}_t, \tilde{\mathbb{M}}_t^{-1}) \in \mathcal{E}_t$.

\emph{Step 1: Residual equivalence.}
Since both structures generate the same conditional distribution $F_{r|z^d,m,w}$, and both $\varphi_t^*$ and $\tilde{\varphi}_t$ are strictly monotonic in the demand shock (Assumption~\ref{A-0}), the residuals are identified up to monotone transformation. Under the normalization that $u_{it}$ follows a standard uniform distribution, we have $\tilde{u}_{it} = u_{it}$.

\emph{Step 2: Affine relationship between latent outputs.}
Define the latent outputs $y_{it}^* := \varphi_t^{*-1}(r_{it}, z_{it}^d, u_{it})$ and $\tilde{y}_{it} := \tilde{\varphi}_t^{-1}(r_{it}, z_{it}^d, u_{it})$. Both structures satisfy:
\begin{align}
y_{it}^* &= f_t^*(x_{it}, z_{it}^s) + \mathbb{M}_t^{*-1}(m_{it}, w_{it}, u_{it}), \label{eq:true_struct} \\
\tilde{y}_{it} &= \tilde{f}_t(x_{it}, z_{it}^s) + \tilde{\mathbb{M}}_t^{-1}(m_{it}, w_{it}, u_{it}). \label{eq:alt_struct}
\end{align}
Define the difference $\Delta(r, z^d, u) := \tilde{\varphi}_t^{-1}(r, z^d, u) - \varphi_t^{*-1}(r, z^d, u) = \tilde{y} - y^*$. Subtracting \eqref{eq:true_struct} from \eqref{eq:alt_struct}:
\begin{equation}
\Delta(r_{it}, z_{it}^d, u_{it}) = \underbrace{\left[\tilde{f}_t(x_{it}, z_{it}^s) - f_t^*(x_{it}, z_{it}^s)\right]}_{=: \Delta_f(x_{it}, z_{it}^s)} + \underbrace{\left[\tilde{\mathbb{M}}_t^{-1}(m_{it}, w_{it}, u_{it}) - \mathbb{M}_t^{*-1}(m_{it}, w_{it}, u_{it})\right]}_{=: \Delta_{\mathbb{M}}(m_{it}, w_{it}, u_{it})}.
\label{eq:delta_decomp}
\end{equation}
Note that $\Delta_f$ depends only on $(x, z^s)$ while $\Delta_{\mathbb{M}}$ depends only on $(m, w, u)$. 

Under Assumption~\ref{A-data}, there exists variation in $(k_{it}, l_{it})$ that is independent of $(u_{it}, z_{it}^d)$ conditional on $(m_{it}, z_{it}^s)$. Fix $(m, z^s, z^d, u)$ and vary $(k, l)$. Since $x = (m, k, l)$ and $w = (k, l, z^s, z^d)$, the left-hand side $\Delta(r, z^d, u)$ changes only through changes in $r$ induced by changes in $y^*$ (and hence $\tilde{y}$). However, by \eqref{eq:delta_decomp}, this change must equal the change in $\Delta_f(m, k, l, z^s)$, since $\Delta_{\mathbb{M}}(m, w, u)$ varies only through its dependence on $(k, l)$ via $w$.

Taking the partial derivative of \eqref{eq:delta_decomp} with respect to $k$ (holding $m, l, z^s, z^d, u$ fixed):
\[
\frac{\partial \Delta_f}{\partial k} + \frac{\partial \Delta_{\mathbb{M}}}{\partial k} = \frac{\partial \Delta}{\partial r} \cdot \frac{\partial r}{\partial k}.
\]
Since $\Delta_f$ depends on $(x, z^s) = (m, k, l, z^s)$ and $\Delta_{\mathbb{M}}$ depends on $(m, w, u) = (m, k, l, z^s, z^d, u)$, both terms on the left may depend on $k$. However, $\Delta_f$ does not depend on $(z^d, u)$ while $\Delta_{\mathbb{M}}$ does not depend on $x$ except through the overlap in $(k, l)$ via $w$. By varying $(z^d, u)$ while holding $(x, z^s)$ fixed, we deduce that $\partial \Delta_{\mathbb{M}}/\partial k$ cannot depend on $(x, z^s)$ beyond its explicit dependence on $(k, l)$ through $w$. 

Since both structures generate the same joint distribution by assumption, there exists a continuous strictly increasing function $g$ such that $\tilde{y} = g(y^*)$. Using \eqref{eq:true_struct}--\eqref{eq:alt_struct}, this implies
\[
g\bigl(f_t^*(x, z^s) + \mathbb{M}_t^{*-1}(m, w, u)\bigr) = \tilde{f}_t(x, z^s) + \tilde{\mathbb{M}}_t^{-1}(m, w, u).
\]
This is a Pexider functional equation: a continuous function of a sum equals a sum of functions of the separate arguments. Under Assumption~\ref{A-data}, the ranges of $f_t^*(\cdot, z^s)$ and $\mathbb{M}_t^{*-1}(\cdot, w, u)$ are non-degenerate intervals when the other argument is held fixed: the supply-side instruments $z^s$ enter $f_t^*$ but not $\mathbb{M}_t^{*-1}$, while the demand quantile $u$ enters $\mathbb{M}_t^{*-1}$ but not $f_t^*$, so each argument can be varied independently of the other. To see this, set $A = f_t^*(x, z^s)$ and $B = \mathbb{M}_t^{*-1}(m, w, u)$, so the equation reads $g(A + B) = h_1(A) + h_2(B)$ where $h_1 := \tilde{f}_t$ and $h_2 := \tilde{\mathbb{M}}_t^{-1}$, viewed as functions of $A$ and $B$ respectively. By the Pexider theorem  \citep[][Theorem 9]{AczelDhombres1989}, the unique continuous solution is $g(s) = \beta s + c$, $h_1(A) = \beta A + \alpha_1$, $h_2(B) = \beta B + \alpha_2$ with $c = \alpha_1 + \alpha_2$. In particular, the common slope $\beta$ applies to both components:
\[
\tilde{f}_t(x, z^s) = \alpha_1 + \beta f_t^*(x, z^s), \qquad \tilde{\mathbb{M}}_t^{-1}(m, w, u) = \alpha_2 + \beta \mathbb{M}_t^{*-1}(m, w, u).
\]

\emph{Step 3: Common scale parameter.}
The Pexider argument in Step~2 already establishes that $\beta_1 = \beta_2 =: \beta$, so we set $b_t := \beta$. Since both $\varphi_t^*$ and $\tilde{\varphi}_t$ are strictly increasing in $y$ (higher output yields higher revenue), we have $b_t > 0$. Setting $a_{1t} := \alpha_1$ and $a_{2t} := \alpha_2$, we obtain:
\[
\tilde{y}_{it} = (a_{1t} + a_{2t}) + b_t y_{it}^*,
\]
which implies $\tilde{\varphi}_t^{-1}(r, z^d, u) = (a_{1t} + a_{2t}) + b_t \varphi_t^{*-1}(r, z^d, u)$.

Hence $(\tilde{\varphi}_t^{-1}, \tilde{f}_t, \tilde{\mathbb{M}}_t^{-1}) \in \mathcal{E}_t$.

\medskip
\noindent\textbf{Proof of (iii).} 
Part~(i) establishes that every element of $\mathcal{E}_t$ is observationally equivalent to the true structure. Part~(ii) establishes that every observationally equivalent structure belongs to $\mathcal{E}_t$. Together, $\mathcal{E}_t$ is exactly the set of structures consistent with the population distribution, i.e., $\mathcal{E}_t$ is exactly identified.
\end{proof}

\begin{rem}
\label{rem:indeterminacy}
The three-dimensional indeterminacy $(a_{1t}, a_{2t}, b_t)$ arises because the unobserved output level $y_{it}$ enters the revenue function through the unknown nonlinear function $\varphi_t$, which has no natural scale or location. The parameter $b_t \in \mathbb{R}_{++}$ governs the common scale of all three structural functions, while $a_{1t}$ and $a_{2t}$ independently shift the location of $f_t$ and $\mathbb{M}_t^{-1}$, respectively. The constraint that $b_t$ is common across $f_t$ and $\mathbb{M}_t^{-1}$ follows from the additive separability of the model.
\end{rem}

\begin{rem}
\label{rem:normalization}
Under Assumption \ref{A-2}, the normalization conditions
\[
f_{t}(m_{t0}^{*}, k_{t}^{*}, l_{t}^{*}, z_{t}^{s*}) = 0, \quad \mathbb{M}_{t}^{-1}(m_{t0}^{*}, w_{t}^{*}, u_{t}^{*}) = 0, \quad \mathbb{M}_{t}^{-1}(m_{t1}^{*}, w_{t}^{*}, u_{t}^{*}) = 1
\]
uniquely pin down $(a_{1t}, a_{2t}, b_t)$, thereby selecting a unique representative from $\mathcal{E}_t$ and achieving point identification of the structural functions.
\end{rem}

\section{Identification of Consumer Preference and Counterfactual Analysis
under CES Assumption with Heterogeneity}

\label{app:ces}

Under the constant elastic inverse demand function (\ref{eq:demand}),
the \textit{structural} budget-share function $\mathfrak{s}_{t}^{*}(\cdot)$
takes the linear semi-parametric form: 
\begin{equation}
\mathfrak{s}_{t}^{*}\left(y_{it}-q_{t}(\mathbf{y}_{t},\mathbf{u}_{t}),u_{it}\right)=\rho(u_{it})\left(y_{it}-q_{t}(\mathbf{y}_{t},\mathbf{u}_{t})\right)+\delta(u_{it}),\label{eq:ces_hetero_structural}
\end{equation}
where $\delta(u_{it})$ captures structural heterogeneity, and $\alpha_{t}(u_{it})=\Phi_{t}-\rho(u_{it})q_{t}+\delta(u_{it})$
in (\ref{eq:demand}). At the baseline state $(\mathbf{y}_{t},\mathbf{u}_{t})$,
imposing normalization ${q}_{t}=0$ (Assumption \ref{A-demand}(c))
allows the reduced-form budget-share function to identify structural
parameters directly from revenue data: 
$r_{it}-\Phi_{t}=\mathfrak{s}_{t}(y_{it},u_{it})=\rho(u_{it})y_{it}+\delta(u_{it})$.
Estimating this equation identifies the elasticity parameter $\rho(u_{it})$
and heterogeneity $\delta(u_{it})$.

For a counterfactual state $(\tilde{\mathbf{y}}_{t},\tilde{\mathbf{u}}_{t})$,
substituting the identified linear form into the market share constraint
(\ref{eq:market_share2}) yields: 
$1=\sum_{i=1}^{N_{t}}\exp\left(\rho(\tilde{u}_{it})(\tilde{y}_{it}-\Delta q_{t})+\delta(\tilde{u}_{it})\right)$.
Since $\rho(\cdot)>0$, the right-hand side is strictly decreasing
in $\Delta q_{t}$, uniquely identifying the quantity index change.

Applying Proposition \ref{P-demand}(b) with $\mathfrak{s}_{t}(\zeta,\tilde{u}_{it})=\rho(\tilde{u}_{it})\zeta+\delta(\tilde{u}_{it})$
yields the utility function: 
\begin{align}
\ln U_{t}(\tilde{\mathbf{y}}_{t},\tilde{\mathbf{u}}_{t}) & =\Delta{q}_{t}+\sum_{i=1}^{N_{t}}\int_{b_{i}}^{\tilde{y}_{it}-\Delta{q}_{t}}\exp\left(\rho(\tilde{u}_{it})\zeta+\delta(\tilde{u}_{it})\right)d\zeta\nonumber \\
 & =\Delta{q}_{t}+\sum_{i=1}^{N_{t}}\frac{1}{\rho(\tilde{u}_{it})}\exp\left(\rho(\tilde{u}_{it})\left(\tilde{y}_{it}-\Delta{q}_{t}\right)+\delta(\tilde{u}_{it})\right)+\text{const}.\label{eq:ces-utility}
\end{align}
This expression represents utility as a weighted sum of structural
market shares. In the homogeneous case ($\rho\left(\tilde{u}_{it}\right)=\rho$
and $\delta\left(\tilde{u}_{it}\right)=\delta$), the summation simplifies
to $1/\rho$ due to the adding-up constraint, recovering the standard
CES result $\ln U_{t}=\Delta q_{t}+\text{const}=\frac{1}{\rho}\ln\sum_{i=1}^{N_{t}}\exp\left(\rho\tilde{y}_{it}\right)+\text{const}$.

We may also apply the counterfactual framework to the CES specification
(\ref{eq:ces_hetero_structural}). In this case, the derivative of
the structural budget-share function with respect to output is constant,
$\frac{\partial\mathfrak{s}_{t}^{*}}{\partial y}=\rho(u_{it})$. Substituting
this derivative into the equilibrium conditions, we can express the
Monopolistic Competition Equilibrium (MCE, where $\mathbb{I}_{MCE}=1$)
and the Marginal Cost Pricing Equilibrium (MCPE, where $\mathbb{I}_{MCE}=0$)
in a unified system of equations in levels: 
\begin{align}
 & \underbrace{\exp\left(\rho(u_{it})(y_{it}^{\star}-\Delta q_{t}^{\star})+\delta(u_{it})+\Phi_{t}+\mathbb{I}_{MCE}\ln\rho(u_{it})-y_{it}^{\star}\right)}_{\text{Marginal Revenue }(\mathbb{I}_{MCE}=1)\text{ or Price }(\mathbb{I}_{MCE}=0)}\nonumber \\
 & \quad=\underbrace{\frac{\exp\left(p_{t}^{m}+\chi_{it}\left(y_{it}^{\star}\right)\right)}{\exp (y_{it}^{*})}\frac{\partial\chi_{it}(y_{it}^{\star})}{\partial y_{it}}}_{\text{Marginal Cost}}\quad\text{for }i=1,...,N_{t},\nonumber \\
 & \sum_{i=1}^{N_{t}}\exp\left(\rho(u_{it})(y_{it}^{\star}-\Delta q_{t}^{\star})+\delta(u_{it})\right)=1,\label{eq:CES_Equilibria}
\end{align}
where $\mathbb{I}_{MCE}$ is an indicator variable. For the MCE, we
set $\mathbb{I}_{MCE}=1$ to solve for $(\mathbf{y}^{\star},\Delta q^{\star})=(\mathbf{y}^{m},\Delta q^{m})$.
For the MCPE, we set $\mathbb{I}_{MCE}=0$ to solve for $(\mathbf{y}^{\star},\Delta q^{\star})=(\mathbf{y}^{c},\Delta q^{c})$.

Using the utility function derived for the case of CES with heterogeneity,
the consumer welfare cost of market power is calculated as: 
\begin{align*}
\ln U_{t}^{c}-\ln U_{t}^{m} & =\left(\Delta q_{t}^{c}-\Delta q_{t}^{m}\right)\\
 & +\sum_{i=1}^{N_{t}}\frac{1}{\rho(u_{it})}\left[\exp\left(\mathfrak{s}_{t}^{*}(y_{it}^{c}-\Delta q_{t}^{c},u_{it})\right)-\exp\left(\mathfrak{s}_{t}^{*}(y_{it}^{m}-\Delta q_{t}^{m},u_{it})\right)\right].
\end{align*}
In the homogeneous case where $\rho(u_{it})=\rho$, the summation
term vanishes because the sum of structural market shares $\sum\exp(\mathfrak{s}_{t}^{*})$
equals 1 in both equilibria, simplifying the welfare loss to the difference
in quantity indices: $\Delta q_{t}^{c}-\Delta q_{t}^{m}$.

\section{Alternative Demand Systems}

\label{app:alternative-demand}

In addition to the HSA demand system, \citet{matsuyama2017beyond}
further propose two additional families of homothetic demand systems: 
\begin{enumerate}
\item the Homothetic demand system with Direct Implicit Additivity (HDIA),
and 
\item the Homothetic demand system with Indirect Implicit Additivity (HIIA). 
\end{enumerate}
Both systems share the same homotheticity structure but impose implicit
additivity either in quantities (HDIA) or prices (HIIA). Below we
summarize their reduced-form representations and show how the two
quantity (or price) indices are identified.

We consider a shift from the baseline state $(\mathbf{y}_{t},\mathbf{z}_{t}^{d},\mathbf{u}_{t})$
to a new state $(\tilde{\mathbf{y}}_{t},\mathbf{z}_{t}^{d},\mathbf{u}_{t})$
where demand shifters $\left(\mathbf{z}_{t}^{d},\mathbf{u}_{t}\right)$
remain unchanged.

\subsection{HDIA demand system}

The structural budget share function of the HDIA has the form: 
\[
\ln\left(\frac{\exp\left(r_{it}\right)}{\sum_{j=1}^{N_{t}}\exp\left(r_{jt}\right)}\right)=\upsilon_{t}^{*}\left(y_{it}-a_{t}\left(\mathbf{y}_{t},\mathbf{z}_{t}^{d},\mathbf{u}_{t}\right),z_{it}^{d},u_{it}\right)-b_{t}\left(\mathbf{y}_{t},\mathbf{z}_{t}^{d},\mathbf{u}_{t}\right)
\]
There are two aggregators $a_{t}\left(\mathbf{y}_{t},\mathbf{z}_{t}^{d},\mathbf{u}_{t}\right)$
and $b_{t}\left(\mathbf{y}_{t},\mathbf{z}_{t}^{d},\mathbf{u}_{t}\right)$.
Aggregator $a_{t}\left(\mathbf{y}_{t},\mathbf{z}_{t}^{d},\mathbf{u}_{t}\right)=\ln U(\mathbf{y}_{t},\mathbf{z}_{t}^{d},\mathbf{u}_{t})$
is the log of a utility function that generates the demand system
and implicitly determined by the additive restriction: 
\begin{equation}
1=\sum_{i=1}^{N_{t}}\Upsilon_{t}^{*}\left(y_{it}-a_{t}\left(\mathbf{y}_{t},\mathbf{z}_{t}^{d},\mathbf{u}_{t}\right),z_{it}^{d},u_{it}\right)\label{eq:HDIA1}
\end{equation}
where $\partial\Upsilon_{t}^{*}\left(x,z_{it}^{d},u_{it}\right)/\partial x=\exp\left(\upsilon_{t}^{*}\left(x,z_{it}^{d},u_{it}\right)\right)$.
Another aggregator $b_{t}\left(\mathbf{y}_{t},\mathbf{z}_{t}^{d},\mathbf{u}_{t}\right)$
is defined by 
\[
b_{t}\left(\mathbf{y}_{t},\mathbf{z}_{t}^{d},\mathbf{u}_{t}\right):=\ln\left[\sum_{i=1}^{N_{t}}\exp\left(\upsilon_{t}^{*}\left(y_{it}-a_{t}\left(\mathbf{y}_{t},\mathbf{z}_{t}^{d},\mathbf{u}_{t}\right),z_{it}^{d},u_{it}\right)\right)\right]
\]

The reduced form revenue function is related to the structural budget
share function as: 
\begin{align*}
r_{it} & =\Phi_{t}+\upsilon_{t}^{*}\left(y_{it}-a_{t}\left(\mathbf{y}_{t},\mathbf{z}_{t}^{d},\mathbf{u}_{t}\right),z_{it}^{d},u_{it}\right)-b_{t}\left(\mathbf{y}_{t},\mathbf{z}_{t}^{d},\mathbf{u}_{t}\right)\\
 & =:\varphi_{t}\left(y_{it},z_{it}^{d},u_{it}\right).
\end{align*}
We define the reduced form budget share function as 
\[
\upsilon_{t}\left(y_{it},z_{it}^{d},u_{it}\right):=\varphi_{t}\left(y_{it},z_{it}^{d},u_{it}\right)-\Phi_{t},
\]
which is related to the structural form as 
\begin{equation}
\upsilon_{t}\left(y_{it},z_{it}^{d},u_{it}\right)=\upsilon_{t}^{*}\left(y_{it}-a_{t}\left(\mathbf{y}_{t},\mathbf{z}_{t}^{d},\mathbf{u}_{t}\right),z_{it}^{d},u_{it}\right)-b_{t}\left(\mathbf{y}_{t},\mathbf{z}_{t}^{d},\mathbf{u}_{t}\right).\label{eq:HDIA2}
\end{equation}

We consider a shift from the baseline state $(\mathbf{y}_{t},\mathbf{z}_{t}^{d},\mathbf{u}_{t})$
to a new state $(\tilde{\mathbf{y}}_{t},\mathbf{z}_{t}^{d},\mathbf{u}_{t})$
where demand shifters $\left(\mathbf{z}_{t}^{d},\mathbf{u}_{t}\right)$
remain unchanged. Let $\Delta a_{t}\left(\tilde{\mathbf{y}}_{t}\right):={a}_{t}(\tilde{\mathbf{y}}_{t},\mathbf{z}_{t}^{d},\mathbf{u}_{t})-{a}_{t}(\mathbf{y}_{t},\mathbf{z}_{t}^{d},\mathbf{u}_{t})$
and $\Delta b_{t}\left(\tilde{\mathbf{y}}_{t}\right):={b}_{t}(\tilde{\mathbf{y}}_{t},\mathbf{z}_{t}^{d},\mathbf{u}_{t})-{b}_{t}(\mathbf{y}_{t},\mathbf{z}_{t}^{d},\mathbf{u}_{t})$.

The reduced form budget share functions with the changes in the quantity
indices gives the value of the structural budget share function at
the new state as follows: 
\begin{align*}
 & \upsilon_{t}\left(\tilde{y}_{it}-\Delta a_{t}\left(\tilde{\mathbf{y}}_{t}\right),z_{it}^{d},u_{it}\right)-\Delta b_{t}\left(\tilde{\mathbf{y}}_{t}\right)\\
 & =\upsilon_{t}^{*}\left(\tilde{y}_{it}-a_{t}\left(\tilde{\mathbf{y}}_{t},{\mathbf{z}}_{t}^{d},\mathbf{u}_{t}\right),z_{it}^{d},\tilde{u}\right)-b_{t}\left(\tilde{\mathbf{y}}_{t},{\mathbf{z}}_{t}^{d},\mathbf{u}_{t}\right).
\end{align*}
Therefore, as in the HSA case, we can conduct a counterfactual analysis
using the reduced form budget share functions and the changes in the
quantity indices. The following lemma shows identification of $\Delta a_{t}\left(\tilde{\mathbf{y}}_{t}\right)$
and $\Delta b_{t}\left(\tilde{\mathbf{y}}_{t}\right)$. 
\begin{lem}
\label{lem:HDIA_indices}$\Delta a_{t}\left(\tilde{\mathbf{y}}_{t}\right)$
is uniquely identified by
\begin{align*}
\sum_{i=1}^{N_{t}}\int_{y_{it}}^{\tilde{y}_{it}-\Delta a_{t}\left(\tilde{\mathbf{y}}_{t}\right)}\exp\left(\upsilon_{t}\left(\zeta,z_{it}^{d},u_{it}\right)\right)d\zeta & =0.
\end{align*}

$\Delta b_{t}\left(\tilde{\mathbf{y}}_{t}\right)$ is identified by 
\begin{align*}
\Delta b_{t}\left(\tilde{\mathbf{y}}_{t}\right) & =\ln\left[\sum_{i=1}^{N_{t}}\exp\left\{ \upsilon_{t}\left(\tilde{y}_{it}-\Delta a_{t}\left(\tilde{\mathbf{y}}_{t}\right),z_{it}^{d},u_{it}\right)\right\} \right].
\end{align*}
\end{lem}

\begin{proof}
Let $B_{t}:=\exp\left(b_{t}\left(\mathbf{y}_{t},{\mathbf{z}}_{t}^{d},\mathbf{u}_{t}\right)\right)$.
From (\ref{eq:HDIA2}), we have 
\begin{align*}
 & \int_{y_{it}}^{\tilde{y}_{it}-\Delta a_{t}\left(\tilde{\mathbf{y}}_{t}\right)}\exp\left(\upsilon_{t}\left(\zeta,z_{it}^{d},u_{it}\right)\right)d\zeta\\
= & \frac{1}{B_{t}}\int_{y_{it}}^{\tilde{y}_{it}-\Delta a_{t}\left(\tilde{\mathbf{y}}_{t}\right)}\exp\left(\upsilon_{t}^{*}\left(\zeta-a_{t}\left(\mathbf{y}_{t},\mathbf{z}_{t}^{d},\mathbf{u}_{t}\right),z_{it}^{d},u_{it}\right)\right)d\zeta\\
= & \frac{1}{B_{t}}\int_{y_{it}}^{\tilde{y}_{it}-\Delta a_{t}\left(\tilde{\mathbf{y}}_{t}\right)}\exp\left(\frac{\partial\Upsilon_{t}^{*}\left(\zeta-a_{t}\left(\mathbf{y}_{t},\mathbf{z}_{t}^{d},\mathbf{u}_{t}\right),z_{it}^{d},u_{it}\right)}{\partial x}\right)d\zeta\\
= & \frac{1}{B_{t}}\left[\Upsilon_{t}^{*}\left(\tilde{y}_{it}-a_{t}\left(\tilde{\mathbf{y}}_{t},\mathbf{z}_{t}^{d},\mathbf{u}_{t}\right),z_{it}^{d},u_{it}\right)-\Upsilon_{t}^{*}\left(y_{it}-a_{t}\left(\mathbf{y}_{t},\mathbf{z}_{t}^{d},\mathbf{u}_{t}\right),z_{it}^{d},u_{it}\right)\right]
\end{align*}
Thus, from (\ref{eq:HDIA1}), 
\begin{align*}
 & \sum_{i=1}^{N_{t}}\int_{y_{it}}^{\tilde{y}_{it}-\Delta a_{t}\left(\tilde{\mathbf{y}}_{t}\right)}\exp\left(\upsilon_{t}\left(\zeta,z_{it}^{d},u_{it}\right)\right)d\zeta\\
= & \frac{1}{B_{t}}\left[\sum_{i=1}^{N_{t}}\Upsilon_{t}^{*}\left(\tilde{y}_{it}-a_{t}\left(\tilde{\mathbf{y}}_{t},\mathbf{z}_{t}^{d},\mathbf{u}_{t}\right),z_{it}^{d},u_{it}\right)-\sum_{i=1}^{N_{t}}\Upsilon_{t}^{*}\left(y_{it}-a_{t}\left(\mathbf{y}_{t},\mathbf{z}_{t}^{d},\mathbf{u}_{t}\right),z_{it}^{d},u_{it}\right)\right]\\
= & \frac{1}{B_{t}}[1-1]=0.
\end{align*}
Because $\exp\left(\upsilon_{t}\left(\cdot\right)\right)>0$, the left-hand side is strictly decreasing in $\Delta a_{t}$. Thus, the last equation uniquely identifies $\Delta a_{t}(\tilde{\mathbf{y}_{t}})$.

From (\ref{eq:HDIA1}), we have 
\begin{align*}
\Delta b_{t}\left(\tilde{\mathbf{y}}_{t}\right) & =\ln\left[\sum_{i=1}^{N_{t}}\exp\left(\upsilon_{t}^{*}\left(\tilde{y}_{it}-a_{t}\left(\tilde{\mathbf{y}}_{t},\mathbf{z}_{t}^{d},\mathbf{u}_{t}\right),z_{it}^{d},u_{it}\right)\right)\right]-b_{t}\left(\mathbf{y}_{t},\mathbf{z}_{t}^{d},\mathbf{u}_{t}\right)\\
 & =\ln\left[\sum_{i=1}^{N_{t}}\exp\left\{ \upsilon_{t}\left(\tilde{y}_{it}-\Delta a_{t}\left(\tilde{\mathbf{y}}_{t}\right),z_{it}^{d},u_{it}\right)+b_{t}\left(\mathbf{y}_{t},\mathbf{z}_{t}^{d},\mathbf{u}_{t}\right)\right\} \right]-b_{t}\left(\mathbf{y}_{t},\mathbf{z}_{t}^{d},\mathbf{u}_{t}\right)\\
 & =\ln\left[\exp\left(b_{t}\left(\mathbf{y}_{t},\mathbf{z}_{t}^{d},\mathbf{u}_{t}\right)\right)\sum_{i=1}^{N_{t}}\exp\left\{ \upsilon_{t}\left(\tilde{y}_{it}-\Delta a_{t}\left(\tilde{\mathbf{y}}_{t}\right),z_{it}^{d},u_{it}\right)\right\} \right]-b_{t}\left(\mathbf{y}_{t},\mathbf{z}_{t}^{d},\mathbf{u}_{t}\right)\\
 & =\ln\left[\sum_{i=1}^{N_{t}}\exp\left\{ \upsilon_{t}\left(\tilde{y}_{it}-\Delta a_{t}\left(\tilde{\mathbf{y}}_{t}\right),z_{it}^{d},u_{it}\right)\right\} \right]
\end{align*}
\end{proof}

\subsection{HIIA demand system}

We consider the direct demand system instead of the inverse demand
system. The structural budget share function of the HIIA has two aggregators
$a_{t}\left(\mathbf{p}_{t},\mathbf{z}_{t}^{d},\mathbf{u}_{t}\right)$
and $b_{t}\left(\mathbf{p}_{t},\mathbf{z}_{t}^{d},\mathbf{u}_{t}\right)$:
\[
\ln\left(\frac{\exp\left(r_{it}\right)}{\sum_{j=1}^{N_{t}}\exp\left(r_{jt}\right)}\right)=\upsilon_{t}^{*}\left(p_{it}-a_{t}\left(\mathbf{p}_{t},\mathbf{z}_{t}^{d},\mathbf{u}_{t}\right),z_{it}^{d},u_{it}\right)-b_{t}\left(\mathbf{p}_{t},\mathbf{z}_{t}^{d},\mathbf{u}_{t}\right)
\]
where $a_{t}\left(\mathbf{p}_{t},\mathbf{z}_{t}^{d},\mathbf{u}_{t}\right)=\ln P_{t}(\mathbf{p}_{t})$
is the log of the ideal price index and implicitly restricted by the
following additive restriction: 
\begin{equation}
1=\sum_{i=1}^{N_{t}}\Upsilon_{t}^{*}\left(p_{it}-a_{t}\left(\mathbf{p}_{t},\mathbf{z}_{t}^{d},\mathbf{u}_{t}\right),z_{it}^{d},u_{it}\right)\label{eq:HDIA1-1}
\end{equation}
where $\partial\Upsilon_{t}^{*}\left(x,z_{it}^{d},u_{it}\right)/\partial x=\exp\left(\upsilon_{t}^{*}\left(x,z_{it}^{d},u_{it}\right)\right)$.
Another price aggregator$b_{t}\left(\mathbf{p}_{t},\mathbf{z}_{t}^{d},\mathbf{u}_{t}\right)$
is defined by 
\[
b_{t}\left(\mathbf{p}_{t},\mathbf{z}_{t}^{d},\mathbf{u}_{t}\right):=\ln\left[\sum_{i=1}^{N_{t}}\exp\left(\upsilon_{t}^{*}\left(p_{it}-a_{t}\left(\mathbf{p}_{t},\mathbf{z}_{t}^{d},\mathbf{u}_{t}\right),z_{it}^{d},u_{it}\right)\right)\right].
\]

We obtain the reduced form revenue function as a function of price:
\[
\chi_{t}\left(p_{it},z_{it}^{d},u_{it}\right):=\psi_{t}^{-1}\left(p_{it},z_{it}^{d},u_{it}\right)+p_{it}
\]
where $\psi_{t}^{-1}\left(p_{it},z_{it}^{d},u_{it}\right)$ is the
inverse of the reduced form inverse demand $\psi_{t}\left(y_{it},z_{it}^{d},u_{it}\right)$
with respect to $y_{it}$.

The reduced form revenue function is related to the structural budget
share function as follows: 
\begin{align*}
r_{it} & =\Phi_{t}+\upsilon_{t}^{*}\left(p_{it}-a_{t}\left(\mathbf{p}_{t},\mathbf{z}_{t}^{d},\mathbf{u}_{t}\right),z_{it}^{d},u_{it}\right)-b_{t}\left(\mathbf{p}_{t},\mathbf{z}_{t}^{d},\mathbf{u}_{t}\right)\\
 & =:\chi_{t}\left(p_{it},z_{it}^{d},u_{it}\right).
\end{align*}
We define the reduced form budget share function as 
\[
\upsilon_{t}\left(p_{it},z_{it}^{d},u_{it}\right):=\chi_{t}\left(p_{it},z_{it}^{d},u_{it}\right)-\Phi_{t},
\]
which is related to the structural form as 
\begin{equation}
\upsilon_{t}\left(p_{it},z_{it}^{d},u_{it}\right)=\upsilon_{t}^{*}\left(p_{it}-a_{t}\left(\mathbf{p}_{t},\mathbf{z}_{t}^{d},\mathbf{u}_{t}\right),z_{it}^{d},u_{it}\right)-b_{t}\left(\mathbf{p}_{t},\mathbf{z}_{t}^{d},\mathbf{u}_{t}\right).\label{eq:HDIA2-1}
\end{equation}

We consider a shift from the baseline state $(\mathbf{p}_{t},\mathbf{z}_{t}^{d},\mathbf{u}_{t})$
to a new state $(\tilde{\mathbf{p}}_{t},\mathbf{z}_{t}^{d},\mathbf{u}_{t})$
where $\left(\mathbf{z}_{t}^{d},\mathbf{u}_{t}\right)$ remain unchanged.
Let $\Delta a_{t}\left(\tilde{\mathbf{p}}_{t}\right):={a}_{t}(\tilde{\mathbf{p}}_{t},\mathbf{z}_{t}^{d},\mathbf{u}_{t})-{a}_{t}(\mathbf{p}_{t},\mathbf{z}_{t}^{d},\mathbf{u}_{t})$
and $\Delta b_{t}\left(\tilde{\mathbf{p}}_{t}\right):={b}_{t}(\tilde{\mathbf{p}}_{t},\mathbf{z}_{t}^{d},\mathbf{u}_{t})-{b}_{t}(\mathbf{p}_{t},\mathbf{z}_{t}^{d},\mathbf{u}_{t})$
be the changes in the price aggregators.

The reduced form budget share functions with those price aggregators
gives the value of the structural budget share function at the new
state: 
\begin{align*}
 & \upsilon_{t}\left(\tilde{p}_{it}-\Delta a_{t}\left(\tilde{\mathbf{p}}_{t}\right),z_{it}^{d},u_{it}\right)-\Delta b_{t}\left(\tilde{\mathbf{p}}_{t}\right)\\
 & =\upsilon_{t}^{*}\left(\tilde{p}_{it}-a_{t}\left(\tilde{\mathbf{p}}_{t},{\mathbf{z}}_{t}^{d},\mathbf{u}_{t}\right),z_{it}^{d},u_{it}\right)-b_{t}\left(\tilde{\mathbf{p}}_{t},{\mathbf{z}}_{t}^{d},\mathbf{u}_{t}\right).
\end{align*}
The following lemma shows identification of $\Delta a_{t}\left(\tilde{\mathbf{p}}_{t}\right)$
and $\Delta b_{t}\left(\tilde{\mathbf{p}}_{t}\right)$. 
\begin{lem}
$\Delta a_{t}\left(\tilde{\mathbf{p}}_{t}\right)$ is uniquely identified by
\begin{align*}
\sum_{i=1}^{N_{t}}\int_{p_{it}}^{\tilde{p}_{it}-\Delta a_{t}\left(\tilde{\mathbf{p}}_{t}\right)}\exp\left(\upsilon_{t}\left(\zeta,z_{it}^{d},u_{it}\right)\right)d\zeta & =0.
\end{align*}
$\Delta b_{t}\left(\tilde{\mathbf{p}}_{t}\right)$ is identified by 
\begin{align*}
\Delta b_{t}\left(\tilde{\mathbf{p}}_{t}\right) & =\ln\left[\sum_{i=1}^{N_{t}}\exp\left\{ \upsilon_{t}\left(\tilde{p}_{it}-\Delta a_{t}\left(\tilde{\mathbf{p}}_{t}\right),z_{it}^{d},u_{it}\right)\right\} \right].
\end{align*}
\end{lem}

\begin{proof}
The proof is identical to that of Lemma \ref{lem:HDIA_indices}, replacing
quantities with prices. 
\end{proof}

\section{Alternative Settings}

\label{app:Alternative-Settings}

\subsection{Endogenous Labor Input}

\label{subsec:Endogenous-Labor-Input}

Identification is possible when a firm chooses $l_{it}$ after observing
$\omega_{it}$ and $u_{it}$. In the spirits of \citet{ackerberg2015identification}
and the dynamic generalized method of moment approach (e.g., \citealp{arellano1991some,arellano1995another,blundell1998initial,blundell2000gmm}),
we provide identification using lagged labor $l_{it-1}$. 

Specifically,
we assume a firm incurs an adjustment cost of labor input, e.g., costs
of recruiting and training new workers. For given $k_{it}$, a firm's
per-period profit (excluding capital costs) is given by 
\[
\exp(\varphi_{t}(f_{t}(m_{it},l_{it},k_{it},z_{it}^{s})+\omega_{it},z_{it}^{d},u_{it}))-\exp(p_{t}^{m}+m_{it})-\exp(p_{t}^{l}+l_{it})-C_{t}(l_{it},l_{it-1}),
\]
where $p_{t}^{l}$ is the wage and $C_{t}(l_{it},l_{it-1})$ is the
adjustment costs. From the above per-period profit, 
the solution to a firm's dynamic problem provides a material demand
function $m_{it}=\tilde{\mathbb{M}}_{t}\left(\omega_{it},l_{it-1},s_{it},u_{it}\right)$
and a labor demand function $l_{it}=\tilde{\mathbb{L}}_{t}(\omega_{it},l_{it-1},s_{it},u_{it})$,
where $s_{it}:=(k_{it},z_{it}^{s},z_{it}^{d})$. We also consider
a ``conditional'' material demand function $m_{it}=\mathbb{M}_{t}\left(\omega_{it},l_{it},s_{it},u_{it}\right)$
when $l_{it}$ is given, which solves the conditional problem (\ref{eq:profit_maximization}).

We assume both $\tilde{\mathbb{M}}_{t}\left(\cdot,l_{it-1},s_{it},u_{it}\right)$
and $\mathbb{M}_{t}\left(\cdot,l_{it},s_{it},u_{it}\right)$ are monotonically
increasing functions so that there exist their inverse functions 
\[
\omega_{it}=\mathbb{M}_{t}^{-1}\left(m_{it},l_{it},s_{it},u_{it}\right)=\tilde{\mathbb{M}}_{t}^{-1}\left(m_{it},l_{it-1},s_{it},u_{it}\right).
\]
In the first step, we substitute $\omega_{it}=\mathbb{M}_{t}^{-1}(\cdot)$
into the revenue function to obtain 
\begin{align*}
r_{it} & =\varphi_{t}(f_{t}(m_{it},l_{it},k_{it},z_{it}^{s})+\mathbb{M}_{t}^{-1}\left(m_{it},l_{it},s_{it},u_{it}\right),z_{it}^{d},u_{it}))\\
 & =\phi_{t}\left(m_{it},l_{it},s_{it},u_{it}\right).
\end{align*}
The first step identification is 
\[
\Pr\left(\left.r_{it}\le\phi_{t}\left(m_{it},l_{it},s_{it},u\right)\right|m_{it-\upsilon-1},l_{it-\upsilon-1},s_{it-\upsilon}\right)=u.
\]
The IVQR identifies $\phi(\cdot)$ and $u_{it}$.

In the second step, we formulate a transformation model using $\omega_{it}=\tilde{\mathbb{M}}_{t}^{-1}(\cdot)$:
\begin{align*}
\tilde{\mathbb{M}}_{t}^{-1}(m_{it},l_{it-1,}s_{it},u_{it}) & =h_{t}\left(\tilde{\mathbb{M}}_{t-1}^{-1}(m_{it-1},l_{it-2},s_{it-1},u_{it-1}),z_{it}^{h}\right)+\eta_{it}\\
 & =\bar{h}_{t}\left(m_{it-1},l_{it-2},s_{it-1},u_{it-1},z_{it}^{h}\right)+\eta_{it}
\end{align*}
Since $\eta_{it}$ is independent of $v_{it}:=(k_{it},l_{it-1,},s_{it},u_{it},m_{it-1},k_{it-1},l_{it-2},s_{it-1},u_{it-1},z_{it-1}^{h})$,
the conditional CDF of $m_{it}$ on $v_{it-1}$ becomes 
\begin{align*}
G_{m_{t}|v_{t}}(m|v_{t}) & =G_{\eta_{t}|v_{t}}\left(\tilde{\mathbb{M}}_{t}^{-1}(m_{it},l_{it-1,}s_{it},u_{it})-\bar{h}_{t}\left(m_{it-1},l_{it-2},s_{it-1},u_{it-1},z_{it-1}^{h}\right)|v_{t}\right)\\
 & =G_{\eta_{t}}\left(\tilde{\mathbb{M}}_{t}^{-1}(m_{it},l_{it-1,}s_{it},u_{it})-\bar{h}_{t}\left(m_{it-1},l_{it-2},s_{it-1},u_{it-1},z_{it-1}^{h}\right)\right).
\end{align*}
Following the same logic of the main text, we can identify $\tilde{\mathbb{M}}_{t}^{-1}(\cdot)$
and $\omega_{it}$ under scale and location normalization. Once we
identify $u_{it}$ and $\omega_{it}$, we can also identify $\omega_{it}=\mathbb{M}_{t}^{-1}\left(m_{it},l_{it},s_{it},u_{it}\right)$,
e.g., by the conditional expectation $\mathbb{M}_{t}^{-1}\left(m_{it},l_{it},s_{it},u_{it}\right)=E\left[\omega_{it}|m_{it},l_{it},s_{it},u_{it}\right]$.

Differentiating $\varphi_{t}^{-1}(\phi_{t}(m_{it},w_{it},u_{it}),z_{it}^{d},u_{it})=f_{t}(x_{it},z_{it}^{s})+\mathbb{M}_{t}^{-1}\left(m_{it},w_{it},u_{it}\right)$
with respect to $q_{it}^{s}\in\{m_{it},k_{it},l_{it},z_{it}^{s}\}$
and $q_{it}^{d}\in\{z_{it}^{d},u_{it}\}$ gives the same equations
as (\ref{eq:phi_qt}) and (\ref{eq:phi_zt}). Therefore, Proposition
\ref{P-step3} holds with the same proof as before.

\paragraph{Role of Adjustment Costs}

The role of adjustment costs is to create variations in $l_{it}$
for given $(m_{it},k_{it},z_{it},u_{it})$. When $l_{it}$ is a fully
flexible input without adjustment costs and chosen at the same timing
as $m_{it}$, the material demand function and the labor demand function
become $m_{it}=\mathbb{M}_{t}^{F}\left(\omega_{it},k_{it},z_{it},u_{it}\right)$
and $l_{it}=\mathbb{L}_{t}^{F}(\omega_{it},k_{it},z_{it},u_{it})$,
respectively. Once $(m_{it},k_{it},z_{it},u_{it})$ are conditioned,
$\omega_{it}$is also conditioned so that $l_{it}$ loses its variation
and we cannot identify $\phi_{t}\left(m_{it},l_{it},k_{it},z_{it},u_{it}\right)$.

\subsection{Endogenous Firm Characteristics}

\label{subsec:Endogenous-Firm-Characteristics}

Firm characteristics $(z_{it}^{s},z_{it}^{d})$ may correlate with
$\zeta_{it}$ and $\eta_{it}$. In step 1, we can use $(z_{it-\upsilon-1}^{s},z_{it-\upsilon-1}^{d})$
in place of $(z_{it-\upsilon}^{s},z_{it-\upsilon}^{d})$ as instrument
variables to construct the moment condition similar to (\ref{eq: IVQR moment condition}).

In step 2, we consider the nonparametric control function approach
by \citet{imbens2009identification} using instrument variables in
the triangular model setting. We assume there exist instrument variables
$\varsigma_{it}=(\varsigma_{it}^{s},\varsigma_{it}^{d})$ , unknown
functions $\Gamma_{kt},$ and unobservable scalars $\kappa_{it}^{k}$
such that 
\begin{align*}
z_{it}^{k} & =\Gamma_{kt}(\varsigma_{it},\varpi_{it},\kappa_{it}^{k}),(k=s,d)
\end{align*}
where $\varpi_{it}:=(l_{it},k_{it},u_{it},w_{it-1},u_{it-1},z_{it-1}^{h})$.
\begin{assumption} \label{assu:IN}For $k=s,d$, (i) $\varsigma_{it}^{k}\Perp(\eta_{it},\kappa_{it}^{k})$
(ii) $\kappa_{it}^{k}$ is a scalar and $\kappa_{it}^{k}\Perp\left(\varsigma_{it},\varpi_{it}\right)$.(iii)
$\Gamma_{kt}$ is strictly increasing in $\kappa_{it}^{k}$ (iv) The
CDF of $\kappa_{it}^{k}$, $F_{\kappa_{t}^{k}}\left(\kappa_{it}^{k}\right)$,
is strictly increasing on the support of $\kappa_{it}^{k}$. \end{assumption}
Let $F_{z_{t}^{k}|\varsigma_{t},\varpi_{t}}(z_{t}^{k}|\varsigma_{t},\varpi_{t})$
be the CDF of $z_{it}^{k}$ conditional on $(\varsigma_{it},\varpi_{it})=(\varsigma_{t},\varpi_{t})$.
Define $\xi_{it}^{k}:=F_{z_{t}^{k}|\varsigma_{t},\varpi_{t}}(z_{it}^{k}|\varsigma_{it},\varpi_{it})$
and $\xi_{it}:=(\xi_{it}^{s},\xi_{it}^{d})$. \citet{imbens2009identification}
showed $\xi_{it}$can be used as control variables, that is, $\eta_{it}$
becomes independent of $v_{it}$ conditional on $\xi_{it}$. 
\begin{lem}
\label{lem:control variable}(\citealp{imbens2009identification},Theorem
1) $\eta_{it}\Perp v_{it}|\xi_{it}$. 
\end{lem}

\begin{proof}
The following proof follows \citet{imbens2009identification}. From
the monotonicity of $\Gamma_{kt}$ in Assumption \ref{assu:IN} (iii),
we can define the inverse function of $\Gamma_{kt}$such that $\kappa_{it}^{k}=\Gamma_{kt}^{-1}\left(\varsigma_{it},\varpi_{it},z_{it}^{k}\right)$.For
given $(z_{t}^{k},\varsigma_{t},\varpi_{t})$ 
\begin{align*}
F_{z_{t}^{k}|\varsigma_{t},\varpi_{t}}(z_{t}^{k}|\varsigma_{t},\varpi_{t}) & =\Pr\left(\Gamma_{kt}(\varsigma_{it},\varpi_{it},\kappa_{it}^{k})\le z_{t}^{k}|\varsigma_{it}=\varsigma_{t},\varpi_{it}=\varpi_{t}\right)\\
 & =\Pr\left(\kappa_{it}^{k}\le\Gamma_{kt}^{-1}\left(\varsigma_{t},\varpi_{t},z_{t}^{k}\right)|\varsigma_{it}=\varsigma_{t},\varpi_{it}=\varpi_{t}\right)\text{ }\\
 & =F_{\kappa_{t}^{k}}\left(\Gamma_{kt}^{-1}\left(\varsigma_{t},\varpi_{t},z_{t}^{k}\right)\right).\,\,\,\text{(from \ensuremath{\kappa_{it}^{k}\Perp\left(\varsigma_{it},\varpi_{it}\right)})}
\end{align*}
Therefore, we have 
\begin{align*}
\xi_{it}^{k} & =F_{z_{t}^{k}|\varsigma_{t},\varpi_{t}}(z_{it}^{k}|\varsigma_{it},\varpi_{it})=F_{\kappa_{t}^{k}}\left(\Gamma_{kt}^{-1}\left(\varsigma_{it},\varpi_{it},z_{it}^{k}\right)\right)=F_{\kappa_{t}^{k}}\left(\kappa_{it}^{k}\right).
\end{align*}

Consider an arbitrary point $(\xi_{t},\eta_{t})$ on the support of
$(\xi_{it},\eta_{it})$. Let $(\kappa_{t}^{s},\kappa_{t}^{d})=\left(F_{\kappa_{t}^{s}}^{-1}(\xi_{t}^{s}),F_{\kappa_{t}^{d}}^{-1}(\xi_{t}^{d})\right)$.
Since $F_{\kappa_{t}^{k}}$ is strictly increasing, the conditional
expectations given $\xi_{it}=\xi_{t}$ are identical to those given
$(\kappa_{it}^{s},\kappa_{it}^{d})=(\kappa_{t}^{s},\kappa_{t}^{d})$.
For any bounded function $a(v_{it})$ of $v_{it}$, the independence
of $(\varsigma_{it},\varpi_{it})$ and $\left(\kappa_{it}^{s},\kappa_{it}^{d},\eta_{it}\right)$
implies 
\begin{align*}
 & E\left[a(v_{it})|\xi_{it}=\xi_{t},\eta_{it}=\eta_{t}\right]\\
 & =E\left[a(v_{it})|\kappa_{it}^{s}=\kappa_{t}^{s},\kappa_{it}^{d}=\kappa_{t}^{d},\eta_{it}=\eta_{t}\right]\\
 & =\int a\left(\Gamma_{st}(\varsigma_{it}\varpi_{it},\kappa_{t}^{s}),\Gamma_{dt}(\varsigma_{it},\varpi_{it},\kappa_{t}^{d}),\varpi_{it}\right)F_{\varsigma_{t},\varpi_{t}}\left(d(\varsigma_{it},\varpi_{it})\right)\\
 & =E\left[a(v_{it})|\kappa_{it}^{s}=\kappa_{t}^{s},\kappa_{it}^{d}=\kappa_{t}^{d}\right]\\
 & =E\left[a(v_{it})|\xi_{it}=\xi_{t}\right].
\end{align*}
For any bounded functions $a(v_{it})$ and $b(\eta_{it})$, we have
\begin{align*}
E\left[a(v_{it})b(\eta_{it})|\xi_{it}=\xi_{t}\right] & =E\left[E\left[a(v_{it})b(\eta_{it})|\xi_{it}=\xi_{t},\eta_{it}=\eta_{t}\right]|\xi_{it}=\xi_{t}\right]\\
 & =E\left[b(\eta_{it})E\left[a(v_{it})|\xi_{it}=\xi_{t},\eta_{it}=\eta_{t}\right]|\xi_{it}=\xi_{t}\right]\\
 & =E\left[b(\eta_{it})E\left[a(v_{it})|\xi_{it}=\xi_{t}\right]|\xi_{it}=\xi_{t}\right]\\
 & =E\left[b(\eta_{it})|\xi_{it}=\xi_{t}\right]E\left[a(v_{it})|\xi_{it}=\xi_{t}\right]
\end{align*}
Thus, $\eta_{it}\Perp v_{it}|\xi_{it}$. 
\end{proof}
Let $\mathfrak{X}:=\{\xi_{it}\}_{i=1}^{n}$ be the set of $\xi_{it}$
constructed for each observation. We update Assumption \ref{A-3}
so that for every $\xi_{it}\in\mathfrak{X}$, the conditional distribution
$G_{\eta_{t}|\xi_{t}}(\cdot|\xi_{it})$ of $\eta_{it}$ on $\xi_{it}$can
satisfy the requirement for $G_{\eta_{t}}(\cdot)$ in Assumption \ref{A-3}.

\begin{assumption} \label{A-3-endogenous} (a) For every $\xi_{it}\in\mathfrak{X}$,
the conditional distribution $G_{\eta_{t}|\xi}(\cdot|\xi_{it})$ of
$\eta_{it}$ conditional on $\xi_{it}$ is absolutely continuous with
a density function $g_{\eta_{t}|\xi_{t}}(\cdot|\xi_{it})$ that is
continuous on its support. (b) $E[\eta_{it}|m_{it-1},w_{it-1},u_{it-1},z_{it-1}^{h}]=0$.
(c) $v_{it}$ is continuously distributed on $\mathcal{V}$. (d) Support
$\varOmega$ of $\omega_{it}$ is an interval $[\text{\underbar{\ensuremath{\omega}}},\bar{\omega}]\subset\mathbb{R}$,
where $\text{\underbar{\ensuremath{\omega}}}<0$ and $1<\bar{\omega}$.
(e) $h(\cdot)$ is continuously differentiable with respect to $\left(\omega,z_{h}\right)$
on $\Omega\times\mathcal{Z}_{h}$. (f) The set $\mathcal{A}_{q_{t-1}}:=\{(m_{it-1},w_{it-1},u_{it-1},z_{it-1}^{h})\in\mathcal{M}\times\mathcal{W}\times[0,1]\times\mathcal{Z}_{h}:\partial G_{m_{t}|v_{t},\xi_{t}}(m_{it}|v_{it},\xi_{it})/\partial q_{it-1}\neq0\text{ for all }(m_{it},w_{it},u_{t})\in\mathcal{M}\times\mathcal{W}\times[0,1]\text{ and for every}\xi_{it}\in\mathfrak{X}\}$
is nonempty for some $q_{it-1}\in\{m_{it-1},k_{it-1},l_{it-1},z_{it-1}^{s},z_{it-1}^{d},u_{it-1},z_{it-1}^{h}\}$.
\end{assumption}

From Lemma \ref{lem:control variable}, the conditional distribution
of $m_{it}$ given $(v_{it},\xi_{it})$ satisfies 
\begin{align*}
G_{m_{t}|v_{t},\xi_{t}}(m_{it}|v_{it},\xi_{it}) & =G_{\eta_{t}|v_{t},\xi_{t}}\left(\mathbb{M}_{t}^{-1}(m_{it},w_{it},u_{it})-\bar{h}_{t}\left(m_{it-1},w_{it-1},u_{it-1},z_{it-1}^{h}\right)|v_{it},\xi_{it}\right)\\
 & =G_{\eta_{t}|\xi_{t}}\left(\mathbb{M}_{t}^{-1}(m_{it},w_{it},u_{it})-\bar{h}_{t}\left(m_{it-1},w_{it-1},u_{it-1},z_{it-1}^{h}\right)|\xi_{it}\right),
\end{align*}
Taking the derivatives of both sides with respect to $q_{it}\in\{m_{it},k_{it},l_{it},z_{it}^{s},z_{it}^{d},u_{it}\}$
and $q_{it-1}\in\{m_{it-1},k_{it-1},l_{it-1},z_{it-1}^{s},z_{it-1}^{d},z_{it-1}^{h},u_{it-1}\}$,
we obtain 
\begin{align}
\frac{\partial G_{m_{t}|v_{t},\xi_{t}}\left(m_{it}|v_{it},\xi_{it}\right)}{\partial q_{it}} & =\frac{\partial\mathbb{M}_{t}^{-1}(m_{it},w_{it},u_{it})}{\partial q_{it}}g_{\eta_{t}|\xi_{t}}\left(\eta_{it}|\xi_{it}\right),\label{eq:dG2-1}\\
\frac{\partial G_{m_{t}|v_{t},\xi_{t}}\left(m_{it}|v_{it},\xi_{it}\right)}{\partial q_{it-1}} & =-\frac{\partial\bar{h}_{t}\left(m_{it-1},w_{it-1},u_{it-1},z_{it-1}^{h}\right)}{\partial q_{it-1}}g_{\eta_{t}|\xi_{t}}\left(\eta_{it}|\xi_{it}\right),\label{eq:dG3-1}
\end{align}
where $\eta_{it}=\mathbb{M}_{t}^{-1}(m_{it},w_{it},u_{it})-\bar{h}_{t}\left(m_{it-1},w_{it-1},u_{it-1},z_{it-1}^{h}\right)$.Using
Assumption \ref{A-3-endogenous}(f), we can choose $q_{it-1}\in\{m_{it-1},k_{it-1},l_{it-1},z_{it-1}^{s},z_{it-1}^{d},z_{it-1}^{h},u_{it-1}\}$
and $(\tilde{m}_{it-1},\tilde{w}_{it-1},\tilde{u}_{it-1},\tilde{z}_{it-1}^{h})\in\mathcal{A}_{q_{t-1}}$
such that $\partial G_{m_{t}|v_{t},\xi_{t}}\left(m_{t}|k_{t},l_{t},z_{t},u_{t},\tilde{m}_{it-1},\tilde{w}_{it-1},\tilde{u}_{it-1},\tilde{z}_{it-1}^{h},\xi_{it}\right)/\partial q_{t-1}\neq0$
for every $\xi_{it}\in\mathfrak{X}$ and for all $(m_{t},k_{t},l_{t},z_{t},u_{t})\in\mathcal{M}\times\mathcal{K}\times\mathcal{L}\times\mathcal{Z}\times[0,1]$.

Dividing (\ref{eq:dG2-1}) by (\ref{eq:dG3-1}), we derive

\begin{align}
\frac{\partial\mathbb{M}_{t}^{-1}(m_{it},w_{it},u_{it})}{\partial q_{it}} & =-\frac{\partial\bar{h}_{t}\left(\tilde{m}_{it-1},\tilde{w}_{it-1},\tilde{u}_{it-1},\tilde{z}_{it-1}^{h}\right)}{\partial q_{it-1}}\nonumber \\
 & \times\frac{\partial G_{m_{t}|v_{t},\xi_{t}}\left(m_{it}|w_{it},u_{it},\tilde{m}_{it-1},\tilde{w}_{it-1},\tilde{u}_{it-1},\tilde{z}_{it-1}^{h},\xi_{it}\right)/\partial q_{t}}{\partial G_{m_{t}|v_{t},\xi_{t}}\left(m_{it}|w_{it},u_{it},\tilde{m}_{it-1},\tilde{w}_{it-1},\tilde{u}_{it-1},\tilde{z}_{it-1}^{h},\xi_{it}\right)/\partial q_{t-1}}.\label{eq:dM/dq-1}
\end{align}
We have obtained equation (\ref{eq:dM/dq-1}) similar to (\ref{eq:dM/dq})
again. Therefore, following the same steps in the proof for Proposition
\ref{P-step1} by replacing $G_{m_{t}|v_{t}}(\cdot|\cdot)$ with $G_{m_{t}|v_{t},\xi_{t}}(\cdot|\cdot,\xi_{it})$,
we can identify $\mathbb{M}_{t}^{-1}(\cdot)$ up to scale and location.
From $\omega_{it}=\mathbb{M}_{t}^{-1}(m_{it},w_{it},u_{it})$ and
$E[\eta_{it}|m_{it-1},w_{it-1},u_{it-1},z_{it-1}^{h}]=0$, we can
identify $\bar{h}_{t}(m_{it-1},w_{it-1},u_{it-1},z_{it-1}^{h})=E\left[\omega_{it}|m_{it-1},w_{it-1},u_{it-1},z_{it-1}^{h}\right]$
and $\eta_{it}=\omega_{it}-\bar{h}_{t}(m_{it-1},w_{it-1},u_{it-1},z_{it-1}^{h})$.
Thus, we can identify the distribution of $\eta_{it}$, $G_{\eta_{t}}(\cdot)$.

Once $\phi_{t}(m_{it},w_{it},u_{it})$ and $\mathbb{M}_{t}^{-1}(m_{it},w_{it},u_{it})$
are identified, the step 3 can identify the same objects as before.
\subsection{Discrete Firm Characteristics}

\label{subsec:Discrete_z}

\subsubsection{Exogenous Characteristics}
\label{subsec:Exogenous-Characteristics}

Suppose $z_{it}^{s}$, $z_{it}^{d}$ and $z_{it}^{h}$ are discrete variables and have finite support $\mathcal{Z}_{s}:=\{z_{s}^{1},...,z_{s}^{J_{s}}\}$, $\mathcal{Z}_{d}:=\{z_{d}^{1},...,z_{d}^{J_{d}}\}$ and $\mathcal{Z}_{h}:=\{z_{h}^{1},...,z_{h}^{J_{h}}\}$. In Step 1, the identification of the IVQR model does not require the continuity of firm characteristics. Therefore, this section proves Propositions \ref{P-step1} and \ref{P-step3}. The following assumption modifies Assumption \ref{A-1} for discrete $z_{it}^{s}$ and $z_{it}^{d}$. 

\begin{assumption} \label{A-1-discrete} (a) For every $z^{s}\in\mathcal{Z}_{s}$, $f_{t}(\cdot,z^{s})$ is continuously differentiable with respect to $(m,k,l)$ on $\mathcal{M}\times\mathcal{K}\times\mathcal{L}$ and strictly increasing in $m$. (b) For every $\left(z^{d},u\right)\in\mathcal{Z}_{d}\times[0,1]$, $\varphi_{t}(\cdot,z^{d},u)$ is strictly increasing and invertible with its inverse $\varphi_{t}^{-1}(r,z^{d},u_{t})$, which is continuously differentiable with respect to $(r,u)$ on $\mathcal{R}\times[0,1]$. (c) For every $(k,l,z^{s},z^{d},u)\in\mathcal{K}\times\mathcal{L}\times\mathcal{Z}_{s}\times\mathcal{Z}_{d}\times[0,1]$, $\mathbb{M}_{t}(\cdot,k,l,z^{s},z^{d},u)$ is strictly increasing and invertible with its inverse $\mathbb{M}_{t}^{-1}(m,k,l,z^{s},z^{d},u)$, which is continuously differentiable with respect to the continuous arguments $(m,k,l,u)$ on $\mathcal{M}\times\mathcal{K}\times\mathcal{L}\times[0,1]$ for every fixed $(z^{s},z^{d})\in\mathcal{Z}_{s}\times\mathcal{Z}_{d}$. (d) $(\zeta_{it},...,\zeta_{it-\upsilon})$ are independent from $\eta_{it}$.
\end{assumption} 

The following assumption modifies Assumption \ref{A-3} for discrete $z_{it}^{s}$ and $z_{it}^{d}$, incorporating the relaxed support condition.

\begin{assumption} \label{A-3-discrete}
(a) The distribution $G_{\eta_{t}}(\cdot)$ of $\eta$ is absolutely continuous with a density function $g_{\eta_{t}}(\cdot)$ that is continuous on its support. 
(b) $\eta_{it}$ is independent of $v_{it}:=(w_{it},u_{it},m_{it-1},w_{it-1},u_{it-1},z_{it-1}^{h})'$ with $E[\eta_{it}|v_{it}]=0$. 
(c) Let $\mathbf{x}_{it} := (m_{it}, k_{it}, l_{it}, u_{it})$. For every fixed tuple of discrete variables $\mathbf{z} := (z_{it}^s, z_{it}^d, z_{it-1}^s, z_{it-1}^d, z_{it-1}^h)$, the conditional support of the continuous variables $(\mathbf{x}_{it}, \mathbf{x}_{it-1})$, denoted $\mathcal{X}(\mathbf{z})$, is an open, connected subset of Euclidean space. Furthermore, the normalization points defined in the proof lie in the interior of these supports.
(d) The support $\varOmega$ of $\omega$ is an interval $[\text{\underbar{\ensuremath{\omega}}},\bar{\omega}]\subset\mathbb{R}$ where $\text{\underbar{\ensuremath{\omega}}}<0$ and $1<\bar{\omega}$.
(e) For every $z^{h}\in\mathcal{Z}_{h}$, $h(\cdot,z^{h})$ is continuously differentiable with respect to $\omega$ on $\Omega$. 
(f) The set $\mathcal{A}_{q_{t-1}}$ (defined analogously to Assumption \ref{A-3}(f) for the discrete case) is nonempty for some $q_{it-1}\in\{m_{it-1},k_{it-1},l_{it-1},u_{it-1}\}$.
(g) For each $(m_{t-1},w_{t-1},u_{t-1},z_{t-1}^{h})\in\text{Proj}_{v}(\mathcal{M}\times \mathcal{V})$, it is possible to find $(\tilde{m}_{t},\tilde{w}_{t},\tilde{u}_{t})$ in the support such that $\partial G_{m_{t}|v_{t}}(\tilde{m}_{t}|\tilde{w}_{t},\tilde{u}_{t},m_{t-1},w_{t-1},u_{t-1},z_{t-1}^{h})/\partial m_{t}>0$.
\end{assumption} 

The following proposition establishes the identification of $\mathbb{M}_{t}^{-1}(\cdot)$. 
\begin{prop}
\label{P-step1-discrete} Suppose that Assumptions \ref{A-data}, \ref{A-2}, \ref{A-1-discrete}, and \ref{A-3-discrete} hold. Then, we can identify $\mathbb{M}_{t}^{-1}(\cdot)$ up to scale and location, and identify $G_{\eta_{t}}(\cdot)$ up to scale. 
\end{prop}

\begin{proof}
Choose normalization points $(m_{t1}^{*},k_{t}^{*},l_{t}^{*},u_{t}^{*})$ and $(m_{t0}^{*},k_{t}^{*},l_{t}^{*},u_{t}^{*})$ as well as $\left(m_{t-1}^{*},k_{t-1}^{*},l_{t-1}^{*},u_{t-1}^{*}\right)$ in the interior of the continuous support such that, for $\left(z_{t}^{s},z_{t-1}^{s},z_{t}^{d},z_{t-1}^{d},z_{t-1}^{h}\right)\in\mathcal{Z}_{s}^{2}\times\mathcal{Z}_{d}^{2}\times\mathcal{Z}_{h}$,
\begin{align}
\mathbb{M}_{t}^{-1}(m_{t0}^{*},k_{t}^{*},l_{t}^{*},z_{t}^{s},z_{t}^{d},u_{t}^{*}) & =c_{0}(z_{t}^{s},z_{t}^{d}),\,\mathbb{M}_{t}^{-1}(m_{t1}^{*},k_{t}^{*},l_{t}^{*},z_{t}^{s},z_{t}^{d},u_{t}^{*})=c_{1}(z_{t}^{s},z_{t}^{d})\text{, }\label{eq:norm_dis}\\
\text{and } & \ensuremath{\bar{h}_{t}(m_{t-1}^{*},k_{t-1}^{*},l_{t-1}^{*},z_{t-1}^{s},z_{t-1}^{d},u_{t-1}^{*},z_{t-1}^{h})=c_{2}(z_{t-1}^{s},z_{t-1}^{d},z_{t-1}^{h}),}
\end{align}
where $\{c_{0}(z_{t}^{s},z_{t}^{d}),c_{1}(z_{t}^{s},z_{t}^{d})\}$ and $\{c_{2}(z_{t-1}^{s},z_{t-1}^{d},z_{t-1}^{h})\}$ are unknown constants. Without loss of generality, let $\left(z_{t}^{s*},z_{t}^{d*}\right)$ in Assumption \ref{A-2} be $z_{t}^{s*}=z_{s}^{1}$ and $z_{t}^{d*}=z_{d}^{1}.$ Thus, the normalization in Assumption \ref{A-2} implies $c_{0}(z_{s}^{1},z_{d}^{1})=0\text{ and }c_{1}(z_{s}^{1},z_{d}^{1})=1$.

Following the derivation in the proof of Proposition \ref{P-step1}, the partial derivative of the control function with respect to continuous arguments $q_{it} \in \{m_{it}, k_{it}, l_{it}, u_{it}\}$ is identified locally as:
\begin{align}
\frac{\partial\mathbb{M}_{t}^{-1}(m_{it},w_{it},u_{it})}{\partial q_{it}} & =\tilde{S}_{q_{t-1}}T_{q_{t}q_{t-1}}(m_{it},w_{it},u_{it}),\label{eq:dM/dm_id_dis}
\end{align}
where $\tilde{S}_{q_{t-1}}$ is a scalar identified by integration over $m$ at the normalization point (analogous to $S_{q_{t-1}}$ in the main proof) and
\begin{align*}
T_{q_{t}q_{t-1}}(m_{it},w_{it},u_{it}) & :=\frac{\partial G_{m_{t}|v_{t}}\left(m_{it}|w_{it},u_{it},\tilde{m}_{it-1},\dots\right)/\partial q_{it}}{\partial G_{m_{t}|v_{t}}\left(m_{it}|w_{it},u_{it},\tilde{m}_{it-1},\dots\right)/\partial q_{it-1}}.
\end{align*}

To recover the level of $\mathbb{M}_{t}^{-1}$ from these partial derivatives, we invoke the connected support condition in Assumption \ref{A-3-discrete}(c). Let $\mathbf{x}_{0}^* := (m_{t0}^{*},k_{t}^{*},l_{t}^{*},u_{t}^{*})$ be the reference point for the continuous variables. For any fixed discrete characteristics $(z_{t}^{s},z_{t}^{d})$ and any target point $\mathbf{x} := (m_{it}, k_{it}, l_{it}, u_{it})$ in the continuous support, there exists a piecewise smooth path $\gamma: [0,1] \to \mathcal{X}(\mathbf{z})$ such that $\gamma(0) = \mathbf{x}_{0}^*$ and $\gamma(1) = \mathbf{x}$. By the Fundamental Theorem of Line Integrals, we identify:

\begin{equation}
\mathbb{M}_{t}^{-1}(\mathbf{x}, z_{t}^{s}, z_{t}^{d}) = c_{0}(z_{t}^{s},z_{t}^{d}) + \int_{\gamma} \nabla_{\mathbf{x}} \mathbb{M}_{t}^{-1}(\mathbf{z}(\tau), z_{t}^{s}, z_{t}^{d}) \cdot d\mathbf{z}(\tau), \label{eq:M_dis_line}
\end{equation}
where $\nabla_{\mathbf{x}} \mathbb{M}_{t}^{-1}$ is the vector of identified partial derivatives from (\ref{eq:dM/dm_id_dis}). This generalizes the previous Manhattan-path integral to any path within the connected support.

Similarly, we identify the gradient of $\bar{h}_t$ with respect to its continuous arguments and recover its level via path integration:
\begin{equation}
\bar{h}_{t}(m_{t-1},w_{t-1},u_{t-1},z_{t-1}^{h})=c_{2}(z_{t-1}^{s},z_{t-1}^{d},z_{t-1}^{h})+\Lambda_{\bar{h}_{t}}(m_{t-1},w_{t-1},u_{t-1},z_{t-1}^{h}),\label{eq:hbar_dis}
\end{equation}
where $\Lambda_{\bar{h}_{t}}$ is the identified line integral of $\nabla \bar{h}_t$ starting from the normalization point $(m_{t-1}^{*}, \dots)$.

Thus, we have identified $\mathbb{M}_{t}^{-1}(\cdot)$ and $\bar{h}_{t}\left(\cdot\right)$ up to the unknown constants $\{ c_{0}(z_{t}^{s},z_{t}^{d}),c_{2}(z_{t-1}^{s},z_{t-1}^{d},z_{t-1}^{h})\}$.

Define $\widetilde{H}_{t}$ as the conditional expectation of the identified components:
\begin{align*}
 \widetilde{H}_{t}(z_{t}^{s},z_{t}^{d},z_{t-1}^{s},z_{t-1}^{d},z_{t-1}^{h})
 :=E[\Lambda_{m}(\mathbf{x}_{it},z_{t}^{s},z_{t}^{d})-\Lambda_{\bar{h}_{t}}(\mathbf{x}_{it-1},z_{t-1}^{h})|z_{t}^{s},z_{t}^{d},z_{t-1}^{s},z_{t-1}^{d},z_{t-1}^{h}].
\end{align*}
Using the mean independence condition $E[\eta_{it}|z_{t}^{s}, \dots]=0$, we obtain the equation linking the constants:
\begin{align}
0 &= \widetilde{H}_{t}(z_{t}^{s},z_{t}^{d},z_{t-1}^{s},z_{t-1}^{d},z_{t-1}^{h})+c_{0}(z_{t}^{s},z_{t}^{d})-c_{2}(z_{t-1}^{s},z_{t-1}^{d},z_{t-1}^{h}).\label{eq:H_evaluation}
\end{align}
Evaluating (\ref{eq:H_evaluation}) at the normalization point $(z_{s}^{1},z_{d}^{1})$ where $c_{0}(z_{s}^{1},z_{d}^{1})=0$, we identify:
\[
c_{2}(z_{t-1}^{s},z_{t-1}^{d},z_{t-1}^{h})=\widetilde{H}_{t}(z_{s}^{1},z_{d}^{1},z_{t-1}^{s},z_{t-1}^{d},z_{t-1}^{h}).
\]
Next, evaluating (\ref{eq:H_evaluation}) at the base lag $(z_{s}^{1},z_{d}^{1},z_{h}^{1})$, we identify $c_{0}(z_{t}^{s},z_{t}^{d})$:
\begin{align*}
c_{0}(z_{t}^{s},z_{t}^{d}) & =c_{2}(z_{s}^{1},z_{d}^{1},z_{h}^{1})-\widetilde{H}_{t}(z_{t}^{s},z_{t}^{d},z_{s}^{1},z_{d}^{1},z_{h}^{1}).
\end{align*}
With all constants identified, $\mathbb{M}_{t}^{-1}(\cdot)$, $\bar{h}_{t}(\cdot)$, and the distribution of $\eta_{it}$ are fully identified.
\end{proof}
\begin{prop}
\label{P-step3-discrete} Suppose that Assumptions \ref{A-data}, \ref{A-2}, \ref{A-1-discrete}, \ref{A-3-discrete}, and \ref{A-FOC} hold. Then, we can identify $\varphi_{t}^{-1}(\cdot)$ and $f_{t}(\cdot)$ up to scale and location and each firm's markup $\partial\varphi_{t}^{-1}(\bar{r}_{it},z_{it})/\partial r_{it}$ up to scale.
\end{prop}

\begin{proof}
From (\ref{eq:phi_qt}) and (\ref{eq:foc_m}), the markup $\partial\varphi_{t}^{-1}(r_{it},z_{it}^{d},u_{it})/\partial r_{it}$ is identified as (\ref{eq:markup_identified}). From $\phi_{t}$ and (\ref{eq:markup_identified}), the markup function $\mu_{t}(m_{it},w_{it},u_{it})$ is also identified as a function of $(m_{it},w_{it},u_{it})$ as (\ref{eq:markup_function}). Substituting (\ref{eq:markup_function}) into (\ref{eq:phi_qt}), we identify the partial derivatives $\partial f_{t}(x_{it},z_{it}^{s})/\partial q_{it}$ for $q_{it}^{s}\in\{m_{it},k_{it},l_{it}\}$ as (\ref{eq:output_elasticities_pro2}).

Define $c_{f}(z_{t}^{s}):=f_{t}(m_{t0}^{*},k_{t}^{*},l_{t}^{*},z_{t}^{s})$ for the normalization point $(m_{t0}^{*},k_{t}^{*},l_{t}^{*})$ defined in Assumption \ref{A-2}. To recover the level of the production function, we invoke the connected support condition for the inputs $(m_t, k_t, l_t)$ for any fixed $z_t^s$. Let $\mathbf{x}_{prod} := (m_t, k_t, l_t)$ and $\mathbf{x}_{prod}^* := (m_{t0}^{*}, k_{t}^{*}, l_{t}^{*})$. For any target inputs $\mathbf{x}_{prod}$ in the support, there exists a path $\gamma: [0,1] \to \mathcal{X}_{prod}(z_t^s)$ connecting $\mathbf{x}_{prod}^*$ to $\mathbf{x}_{prod}$. By the Fundamental Theorem of Line Integrals, we identify $f_t$ up to the constant $c_f(z_t^s)$:
\[
f_{t}(\mathbf{x}_{prod},z_{t}^{s})=c_{f}(z_{t}^{s})+\Lambda_{f}(\mathbf{x}_{prod},z_{t}^{s})
\]
where $\Lambda_{f}$ is the identified path integral of the marginal products:
\begin{align*}
\Lambda_{f}(\mathbf{x}_{prod},z_{t}^{s}) & = \int_{\gamma} \nabla_{\mathbf{x}} f_{t}(\mathbf{z}(\tau), z_{t}^{s}) \cdot d\mathbf{z}(\tau) \\
& = \int_{0}^{1} \left[ \sum_{q \in \{m, k, l\}} \frac{\partial f_{t}(\gamma(\tau), z_t^s)}{\partial q_{it}} \frac{d\gamma_{q}(\tau)}{d\tau} \right] d\tau.
\end{align*}

To identify the unknown constants $c_{f}(z_{t}^{s})$, we use the normalization in Assumption \ref{A-2}. Recall that at the global normalization point $(z_t^{s*}, z_t^{d*})$:
\begin{align}
& \varphi_{t}^{-1}\left(\phi_{t}\left(m_{t0}^{*},k_{t}^{*},l_{t}^{*},z_{t}^{s*},z_{t}^{d*},u_{t}^{*}\right),z_{t}^{d*},u_{t}^{*}\right)\nonumber \\
= & f_{t}(m_{t0}^{*},k_{t}^{*},l_{t}^{*},z_{t}^{s*})+\mathbb{M}_{t}^{-1}(m_{t0}^{*},k_{t}^{*},l_{t}^{*},z_{t}^{s*},z_{t}^{d*},u_{t}^{*})=0.\label{eq:normalization}
\end{align}
For an arbitrary $z_t^s$, the relationship holds as:
\begin{align*}
& \varphi_{t}^{-1}\left(\phi_{t}\left(m_{t0}^{*},k_{t}^{*},l_{t}^{*},z_{t}^{s},z_{t}^{d*},u_{t}^{*}\right),z_{t}^{d*},u_{t}^{*}\right)\\
= & f_{t}(m_{t0}^{*},k_{t}^{*},l_{t}^{*},z_{t}^{s})+\mathbb{M}_{t}^{-1}(m_{t0}^{*},k_{t}^{*},l_{t}^{*},z_{t}^{s},z_{t}^{d*},u_{t}^{*})\\
= & c_{f}(z_{t}^{s})+\mathbb{M}_{t}^{-1}(m_{t0}^{*},k_{t}^{*},l_{t}^{*},z_{t}^{s},z_{t}^{d*},u_{t}^{*}).
\end{align*}
The LHS can be expressed by integrating the identified markup term $\partial\varphi_{t}^{-1}/\partial r_{it}$. Let $r_{base} = \phi_{t}(m_{t0}^{*},k_{t}^{*},l_{t}^{*},z_{t}^{s*},z_{t}^{d*},u_{t}^{*})$ and $r_{target} = \phi_{t}(m_{t0}^{*},k_{t}^{*},l_{t}^{*},z_{t}^{s},z_{t}^{d*},u_{t}^{*})$. Then:
\begin{align*}
& \int_{r_{base}}^{r_{target}}\frac{\partial\varphi_{t}^{-1}(s,z_{t}^{d*},u_{t}^{*})}{\partial r_{it}}ds\\
= & \varphi_{t}^{-1}(r_{target},z_{t}^{d*},u_{t}^{*}) - \varphi_{t}^{-1}(r_{base},z_{t}^{d*},u_{t}^{*})\\
= & \varphi_{t}^{-1}(r_{target},z_{t}^{d*},u_{t}^{*}) \quad \text{ (from \ref{eq:normalization})}\\
= & c_{f}(z_{t}^{s})+\mathbb{M}_{t}^{-1}(m_{t0}^{*},k_{t}^{*},l_{t}^{*},z_{t}^{s},z_{t}^{d*},u_{t}^{*}).
\end{align*}
Since $\phi_{t}(\cdot)$, the markup $\partial\varphi_{t}^{-1}(\cdot)/\partial r_{it}$, and the control function $\mathbb{M}_{t}^{-1}(\cdot)$ are already identified, $c_{f}(z_{t}^{s})$ is identified as:
\begin{align*}
c_{f}(z_{t}^{s}) & =\int_{r_{base}}^{r_{target}}\frac{\partial\varphi_{t}^{-1}(s,z_{t}^{d*},u_{t}^{*})}{\partial r_{it}}ds -\mathbb{M}_{t}^{-1}(m_{t0}^{*},k_{t}^{*},l_{t}^{*},z_{t}^{s},z_{t}^{d*},u_{t}^{*}).
\end{align*}
Thus, $f_{t}(\cdot)$ is identified.

Finally, for any given $(r_{t},z_{t}^{d})\in\mathcal{R}\times\mathcal{Z}_{d}$, the set $B_{t}\left(r_{t},z_{t}^{d},u_{t}\right):=\left\{ \left(x_{t},z_{t}^{s}\right)\in\mathcal{X}\times\mathcal{Z}_{s}:\phi_{t}\left(x_{t},z_{t}^{s},z_{t}^{d},u_{t}\right)=r_{t}\right\} $ is non-empty. The output quantity $\varphi_{t}^{-1}(r_{t},z_{t},u_{t})$ is identified by summing the identified components:
\[
\varphi_{t}^{-1}(r_{t},z_{t}^{d},u_{t})=f_{t}(x_{t},z_{t}^{s})+\mathbb{M}_{t}^{-1}(m_{t},w_{t},u_{t})\text{ for any }\left(x_{t},z_{t}^{s}\right)\in B_{t}(r_{t},z_{t}^{d},u_{t}).
\]
The output price for individual firms is then identified as $p_{it} = r_{it}-\varphi_{t}^{-1}(r_{it},z_{it}^{d},u_{it})$.
\end{proof}

\subsubsection{Endogenous Characteristics}

Firm characteristics $(z_{it}^{s},z_{it}^{d})$ may correlate with
$u_{it}$ and $\eta_{it}$. In step 1, we can use $(z_{it-\upsilon}^{s},z_{it-\upsilon}^{d})$
instead of $(z_{it-\upsilon-1}^{s},z_{it-\upsilon-1}^{d})$ as instrument
variables to construct the moment condition similar to (\ref{eq: IVQR moment condition}).
In step 2, we consider the control variable approach as in subsection
\ref{subsec:Endogenous-Firm-Characteristics}. 
\begin{prop}
\label{P-step1-discrete-1} Suppose that Assumptions \ref{A-data},
\ref{A-2}, \ref{assu:IN}, \ref{A-3-endogenous}, \ref{A-1-discrete},
and \ref{A-3-discrete} hold. Then, we can identify $\mathbb{M}_{t}^{-1}(\cdot)$
up to scale and location, and identify $G_{\eta_{t}}(\cdot)$ up to
scale. 
\end{prop}

\begin{proof}
Following the same steps in the proof for Proposition \ref{P-step1-discrete}
by replacing $G_{m_{t}|v_{t}}(\cdot|\cdot)$ with $G_{m_{t}|v_{t},\xi_{t}}(\cdot|\cdot,\xi_{it})$,
we can identify $\mathbb{M}_{t}^{-1}(m_{it},w_{it},u_{it})$ and $\bar{h}_{t}\left(m_{t-1},w_{t-1},u_{t-1},z_{t-1}^{h}\right)$
up to 
\[
\left\{ c_{0}(z_{t}^{s},z_{t}^{d}),c_{2}(z_{t-1}^{s},z_{t-1}^{d},z_{t-1}^{h})\right\} _{z\in\mathcal{Z}}.
\]
Define 
\begin{align*}
 & \widetilde{H}_{t}(z_{t}^{s},z_{t}^{d},z_{t-1}^{s},z_{t-1}^{d},z_{t-1}^{h})\\
 & :=E[\Lambda_{m}(x_{it},z_{t}^{s},z_{t}^{d},u_{it})-\Lambda_{\bar{h}_{t}}(x_{it-1},z_{t}^{s},z_{t}^{d},u_{it-1},z_{t-1}^{h})|z_{t}^{s},z_{t}^{d},z_{t-1}^{s},z_{t-1}^{d},z_{t-1}^{h}].
\end{align*}
Applying the law of iterated expectations with $E\left[\eta_{it}|v_{it},\xi_{it}\right]=0$
from Assumption \ref{A-3-endogenous}(b), we have 
\begin{align*}
E\left[\eta_{it}|z_{t}^{s},z_{t}^{d},z_{t-1}^{s},z_{t-1}^{d},z_{t-1}^{h}\right] & =E\left[E\left[\eta_{it}|v_{it},\xi_{it}\right]|z_{t}^{s},z_{t}^{d},z_{t-1}^{s},z_{t-1}^{d},z_{t-1}^{h}\right]=0.
\end{align*}
From this, we have (\ref{eq:H_evaluation}) as follows: 
\begin{align*}
0 & =E\left[\eta_{it}|z_{t}^{s},z_{t}^{d},z_{t-1}^{s},z_{t-1}^{d},z_{t-1}^{h}\right]\\
 & =E\left[\mathbb{M}_{t}^{-1}(x_{it},z_{t}^{s},z_{t}^{d},u_{it})-\bar{h}_{t}\left(x_{it-1},z_{t-1}^{s},z_{t-1}^{d},u_{it-1},z_{t-1}^{h}\right)|z_{t}^{s},z_{t}^{d},z_{t-1}^{s},z_{t-1}^{d},z_{t-1}^{h}\right]\\
 & =\widetilde{H}_{t}(z_{t}^{s},z_{t}^{d},z_{t-1}^{s},z_{t-1}^{d},z_{t-1}^{h})+c_{0}(z_{t}^{s},z_{t}^{d})-c_{2}(z_{t-1}^{s},z_{t-1}^{d},z_{t-1}^{h}).
\end{align*}
Therefore, following the same steps in the proof for Proposition \ref{P-step1-discrete},
we can identify $\mathbb{M}_{t}^{-1}(\cdot)$ up to scale and location,
and identify $G_{\eta_{t}}(\cdot)$ up to scale. 
\end{proof}
In step 3, Proposition \ref{P-step3-discrete} holds with the same
proof.

\subsection{Identification with Persistent Demand Shocks}
\label{app:persistent_demand}

In Section~\ref{subsec:Setting}, identification relies on the limited
persistence of $\epsilon_{it}$ (Assumption~\ref{A-0}) to justify
using lagged inputs as instruments. If $\epsilon_{it}$ is persistent
(e.g., $\epsilon_{it}=\rho_{\epsilon}\epsilon_{it-1}+\zeta_{it}$),
lagged inputs become endogenous as they depend on past demand shocks
correlated with current demand.

This subsection establishes that the revenue function remains nonparametrically
identified under persistent demand shocks, provided one observes a
\emph{lagged supply-side shifter} $z_{it-1}^{h}$ (e.g., R\&D investment
or technology adoption). Let $w_{it}:=(k_{it},l_{it},z_{it}^{s},z_{it}^{d})$
be the vector of controls. We replace the baseline restrictions with
the following assumption:

\begin{assumption}[Supply-side instrument for persistent demand]\label{ass:supplyIV_persistent}
Suppose that the followings hold: (a) {Relevance.} Conditional on
$w_{it}$, $z_{it-1}^{h}$ shifts the distribution of current physical
productivity $\omega_{it}$, and consequently shifts the flexible
input $m_{it}$ via the input demand function. (b) {Exclusion and
Orthogonality.} $z_{it-1}^{h}$ is excluded from the inverse demand
function (conditional on controls) and is independent of the \emph{current}
demand innovation $\zeta_{it}$ and the past demand shock $\epsilon_{it-1}$ conditional on $w_{it}$: $(\zeta_{it}, \epsilon_{it-1})\ \perp\!\!\!\perp\ z_{it-1}^{h}\ \big|\ w_{it}$.
(c) {Completeness.} The family of conditional distributions of $m_{it}$
given $(z_{it-1}^{h},w_{it})$ is complete. That is, for any measurable
function $g(m,w)$, $E[g(m_{it},w_{it})\mid z_{it-1}^{h},w_{it}]=0$
almost surely implies $g(m_{it},w_{it})=0$ almost surely. 
\end{assumption}

Assumption \ref{ass:supplyIV_persistent}(b) assumes that the instrument $z_{it-1}^{h}$ is pre-determined relative to demand; alternatively, we may take the extra lag, i.e., replacing $z_{it-1}^{h}$ with $z_{it-2}^{h}$.

Under Assumption~\ref{ass:supplyIV_persistent} and the strict monotonicity
of the revenue function (Assumption~\ref{A-1}), the function $\phi_{t}(\cdot)$
is uniquely identified by the IVQR moment condition using $(z_{it-1}^{h},w_{it})$:
\[
\Pr \!\left[r_{it}\le\phi_{t}(m_{it},w_{it},u)\ \big|\ z_{it-1}^{h},w_{it}\right]=u,\quad\forall\,u \in[0,1].
\]
Intuitively, $z_{it-1}^{h}$ isolates variation in input choice driven
by supply/productivity shifts that are orthogonal to the current demand
innovation $\zeta_{it}$, even if $\epsilon_{it}$ itself is serially
correlated.

This extension replaces the requirement of transitory demand shocks
with the requirement of a valid supply-side instrument. The validity
of Assumption~\ref{ass:supplyIV_persistent} requires that the firm's
choice of $z_{it-1}^{h}$ is not systematically driven by the persistent
component of demand ($\epsilon_{it-1}$). Therefore, this strategy
is most credible when $z_{it-1}^{h}$ reflects supply-side forces
(e.g., technological opportunities or cost incentives) rather than
responses to medium-run demand fluctuations.


\subsection{Unobserved Quality Heterogeneity}

\label{app:quality}

This subsection introduces persistent latent product quality as a
structural source of persistent demand heterogeneity. The goal is
\emph{not} to identify the baseline model under a persistent $\epsilon_{it}$
process. Rather, we maintain the baseline requirement that $\epsilon_{it}$
captures high-frequency demand fluctuations (and hence satisfies the
limited-persistence logic used in Step~1), while allowing persistent
movements in demand to operate through a persistent quality component
that is explicitly modeled. Put differently, persistent demand heterogeneity
is accommodated by $\delta_{it}$ below, while $\epsilon_{it}$ is
interpreted as the remaining transitory demand shock.

Let $\Delta_{it}$ denote the unobserved quality of firm $i$'s product
at time $t$, and let $\delta_{it}:=\ln\Delta_{it}$. Define quality-adjusted
price and quantity by 
\[
\tilde{P}_{it}:=\frac{P_{it}}{\Delta_{it}}\qquad\text{and}\qquad\tilde{Y}_{it}:=Y_{it}\Delta_{it}.
\]
In logs, $\tilde{p}_{it}=p_{it}-\delta_{it}$ and $\tilde{y}_{it}=y_{it}+\delta_{it}$.
The representative consumer derives utility from quality-adjusted
units, so inverse demand is written as 
\[
\tilde{P}_{it}=\Psi_{t}(\tilde{Y}_{it},z_{it}^{d},\epsilon_{it}),
\]
with $\Psi_{t}$ strictly decreasing in $\tilde{Y}_{it}$. In observed
variables this implies 
\[
p_{it}=\psi_{t}(y_{it}+\delta_{it},z_{it}^{d},\epsilon_{it})+\delta_{it},
\]
where $\psi_{t}(\cdot):=\ln\Psi_{t}(\exp(\cdot))$. Hence revenue
satisfies 
\begin{equation}
r_{it}=p_{it}+y_{it}=\varphi_{t}(y_{it}+\delta_{it},z_{it}^{d},\epsilon_{it}):=\psi_{t}(y_{it}+\delta_{it},z_{it}^{d},\epsilon_{it})+(y_{it}+\delta_{it}).\label{eq:quality_revenue_revised}
\end{equation}

On the production side, physical output is 
\[
y_{it}=f_{t}(x_{it},z_{it}^{s})+\omega_{it},\qquad x_{it}=(m_{it},k_{it},l_{it}).
\]
Substituting into~\eqref{eq:quality_revenue_revised} yields 
\begin{equation}
r_{it}=\varphi_{t}\!\Big(f_{t}(x_{it},z_{it}^{s})+\underbrace{(\omega_{it}+\delta_{it})}_{\nu_{it}},\,z_{it}^{d},\,\epsilon_{it}\Big),\label{eq:quality_revenue_composite}
\end{equation}
where $\nu_{it}:=\omega_{it}+\delta_{it}$ is \emph{revenue productivity}.
In this extension, the objects identified by the revenue-based strategy
naturally correspond to revenue productivity $\nu_{it}$, rather than
physical productivity $\omega_{it}$ alone.

To preserve the control-function structure used in Section~3, we
make explicit the information set underlying input choice. We assume
the firm observes the composite revenue-productivity state $\nu_{it}$
(or equivalently observes $\omega_{it}$ and $\delta_{it}$) when
choosing the flexible input, so the intermediate input policy takes
the form 
\[
m_{it}=\mathbb{M}_{t}(\nu_{it},w_{it},u_{it}),\qquad w_{it}:=(k_{it},l_{it},z_{it}^{s},z_{it}^{d}),
\]
and $\epsilon_{it}$ (equivalently $u_{it}$) remains the demand-side
shock entering the revenue function monotonically as in Assumption~2.
With this timing, Step~1 continues to recover $\varphi_{t}(\cdot)$
and $u_{it}$ (under the same orthogonality conditions as in the baseline
for $u_{it}$), and Step~2--3 recover the law of motion for the
\emph{identified} productivity state (now $\nu_{it}$).

Assume physical productivity and quality each follow AR(1) processes:
\[
\omega_{it}=\rho_{\omega}\omega_{it-1}+\eta_{it}^{\omega},\qquad\delta_{it}=\rho_{\delta}\delta_{it-1}+\eta_{it}^{\delta},
\]
with innovations $(\eta_{it}^{\omega},\eta_{it}^{\delta})$. If $\rho_{\omega}=\rho_{\delta}=\rho$,
then $\nu_{it}=\omega_{it}+\delta_{it}$ also follows AR(1): 
\begin{equation}
\nu_{it}=\rho\nu_{it-1}+(\eta_{it}^{\omega}+\eta_{it}^{\delta}).\label{eq:nu_ar1}
\end{equation}
In this case, the Markov structure exploited in Section~3 applies
directly to $\nu_{it}$, and the multi-step identification strategy
carries through with $\nu_{it}$ interpreted as the productivity state.

If instead $\rho_{\omega}\neq\rho_{\delta}$, then $\nu_{it}$ need
not be AR(1). A natural formulation is to treat $(\omega_{it},\delta_{it})$
as a two-dimensional Markov state. The revenue equation~\eqref{eq:quality_revenue_composite}
depends on these states only through $\nu_{it}$, so revenue data
alone generally do not allow one to separately identify $\omega_{it}$
and $\delta_{it}$ without additional structure (e.g., quality proxies,
direct price/quantity information beyond revenue, or restrictions
that link quality to observables). For expositional clarity and to
remain close to Section~3, we focus on the empirically relevant case
in which $\nu_{it}$ admits a parsimonious first-order representation
such as~\eqref{eq:nu_ar1}.

This extension rationalizes persistent demand heterogeneity through
a persistent quality component while keeping $\epsilon_{it}$ as the
high-frequency demand shock required for Step~1. Under this interpretation,
the baseline method identifies (i) the revenue function and demand
shock $u_{it}$ and (ii) the dynamics of revenue productivity $\nu_{it}$.
Separating physical productivity $\omega_{it}$ from quality $\delta_{it}$
is not identified from revenue alone and would require additional
assumptions or measurements.

\subsection{Identification with Heterogeneous Material Prices}
\label{subsec:Identification-Expenditure}

In the main text, we treat deflated expenditures on materials as measures
of inputs. This abstracts from unobserved heterogeneity in material
prices. For instance, geographically segmented input markets may induce
systematic price differences across regions. This Appendix shows that
the identification strategy extends to such environments when material
prices vary across firms as a function of observed characteristics.

Suppose that econometricians observe only material expenditure,
$m_{it}^{*} = m_{it} + p_{it}^{m}$, but not $m_{it}$ and $p_{it}^{m}$
separately, and that the material price $p_{it}^{m}$ varies across firms
and is determined by exogenous observables $z_{it}^{m}$:

\begin{equation}
p_{it}^{m} = p_{t}^{m}\left(z_{it}^{m}\right).
\label{eq:material_price_eq}
\end{equation}
One motivation for specification \eqref{eq:material_price_eq} is the
presence of regional input markets. When input markets are geographically
segmented, material prices may differ across regions. In this case,
$z_{it}^{m}$ can be specified as a region dummy indicating the location
of firm $i$.

Substituting $m_{it} = m_{it}^{*} - p_{t}^{m}\left(z_{it}^{m}\right)$
into $f_{t}$, we define the transformed production function $f_{t}^{*}$
as
\begin{align*}
f_{t}\left(m_{it}, k_{it}, l_{it}, z_{it}^{s}\right)
&= f_{t}\left(m_{it}^{*} - p_{t}^{m}\left(z_{it}^{m}\right),
k_{it}, l_{it}, z_{it}^{s}\right) \\
&:= f_{t}^{*}\left(m_{it}^{*}, k_{it}, l_{it}, z_{it}^{s}, z_{it}^{m}\right).
\end{align*}
 Because this transformation is additive in $m_{it}$, the elasticities
of $f_t$ and $f_t^{*}$ coincide. Hence, identification of $f_t^{*}$ is
sufficient for identification of $f_t$.

Since $z_{it}^{m}$ is exogenously given for each firm, the firm's
profit-maximization problem can be written as a choice of $m_{it}^{*}$
rather than $m_{it}$:
\begin{align*}
\max_{m_{it}} \;
&\exp\left(
\varphi_{t}\left(
f_{t}\left(m_{it}, k_{it}, l_{it},z_{it}^{s}\right) + \omega_{it},
z_{it}^{d}, u_{it}
\right)
\right)
- \exp\left(m_{it} + p_{t}^{m}\left(z_{it}^{m}\right)\right) \\
= \max_{m_{it}^{*}} \;
&\exp\left(
\varphi_{t}\left(
f_{t}^{*}\left(m_{it}^{*}, k_{it}, l_{it}, z_{it}^{s}, z_{it}^{m}\right)
+ \omega_{it}, z_{it}^{d}, u_{it}
\right)
\right)
- \exp\left(m_{it}^{*}\right).
\end{align*}
Since $z_{it}^{m}$ is observed and exogenous, it can be treated in the
same way as $z_{it}^{s}$ throughout the identification argument.
Conditional on $(z_{it}^{s},z_{it}^{m})$, the model is isomorphic to
the baseline case in the main text, and all steps of the three-step
identification strategy apply without modification.

\section{Robustness to Decreasing Returns to Scale}
\label{app:robustness_rts}

In the main text, we impose constant returns to scale (CRS), $\theta_{m}+\theta_{k}+\theta_{l}=1$, as a scale normalization for the Cobb--Douglas production function. To assess the sensitivity of our empirical results to this restriction, we estimate the model under decreasing returns to scale (DRS), setting $\theta_{m}+\theta_{k}+\theta_{l}=0.9$ and $\theta_{m}+\theta_{k}+\theta_{l}=0.95$. All other aspects of the estimation procedure remain identical to the baseline.

\begin{table}[H]
\begin{centering}
\begin{tabular}[t]{lccccc}
\toprule 
Industry  & $n$  & $\hat{\theta}_{m}$  & $\hat{\theta}_{k}$  & $\hat{\theta}_{l}$  & $\hat{\bar{\mu}}$ \tabularnewline
\midrule 
31  & 736  & 0.766  & 0.011  & 0.122  & 1.252 \tabularnewline
 &  & (0.032)  & (0.008)  & (0.031)  & (0.053) \tabularnewline
 &  &  &  &  & \tabularnewline
32  & 463  & 0.680  & 0.063  & 0.157  & 1.351 \tabularnewline
 &  & (0.058)  & (0.035)  & (0.037)  & (0.114) \tabularnewline
 &  &  &  &  & \tabularnewline
38  & 391  & 0.617  & 0.039  & 0.245  & 1.495 \tabularnewline
 &  & (0.057)  & (0.032)  & (0.051)  & (0.139) \tabularnewline
\bottomrule
\end{tabular}
\par\end{centering}
\caption{Chilean Manufacturing plant estimation under $\theta_{m}+\theta_{k}+\theta_{l}=0.9$: Step 1, Step 2, and Step 3
(Industries 31, 32, and 38 in 1996). Standard errors in parentheses
with 100 non-parametric bootstrap iterations. }
\label{tab:estimates1_09} 
\end{table}

\begin{table}[H]
\begin{centering}
\begin{tabular}[t]{lcccc}
\toprule 
Industry  & $n$  & $\hat{\beta}$  & $\hat{\gamma}$  & $\hat{\delta}$ \tabularnewline
\midrule 
31  & 623  & 0.360  & 3.383  & -8.667 \tabularnewline
 &  & (0.005)  & (0.205)  & (0.205) \tabularnewline
 &  &  &  & \tabularnewline
32  & 354  & 0.094  & 0.911  & -6.912 \tabularnewline
 &  & (0.018)  & (0.265)  & (0.702) \tabularnewline
 &  &  &  & \tabularnewline
38  & 318  & 0.070  & 0.703  & -6.009 \tabularnewline
 &  & (0.016)  & (0.291)  & (0.437) \tabularnewline
\bottomrule
\end{tabular}
\par\end{centering}
\caption{Chilean Manufacturing plant estimation under $\theta_{m}+\theta_{k}+\theta_{l}=0.9$: Step 4 (Industries 31, 32,
and 38 in 1996). Standard errors in parentheses with 100 non-parametric
bootstrap iterations. }
\label{tab:estimates2_09} 
\end{table}

\begin{table}[H]
\begin{centering}
\begin{tabular}[t]{lccccc}
\toprule 
Industry  & $n$  & $\hat{\theta}_{m}$  & $\hat{\theta}_{k}$  & $\hat{\theta}_{l}$  & $\hat{\bar{\mu}}$ \tabularnewline
\midrule 
31  & 736  & 0.809  & 0.012  & 0.129  & 1.322 \tabularnewline
 &  & (0.034)  & (0.009)  & (0.032)  & (0.056) \tabularnewline
 &  &  &  &  & \tabularnewline
32  & 463  & 0.718  & 0.067  & 0.166  & 1.426 \tabularnewline
 &  & (0.061)  & (0.037)  & (0.039)  & (0.120) \tabularnewline
 &  &  &  &  & \tabularnewline
38  & 391  & 0.651  & 0.041  & 0.258  & 1.578 \tabularnewline
 &  & (0.060)  & (0.034)  & (0.053)  & (0.147) \tabularnewline
\bottomrule
\end{tabular}
\par\end{centering}
\caption{Chilean Manufacturing plant estimation under $\theta_{m}+\theta_{k}+\theta_{l}=0.95$: Step 1, Step 2, and Step 3
(Industries 31, 32, and 38 in 1996). Standard errors in parentheses
with 100 non-parametric bootstrap iterations. }
\label{tab:estimates1_095} 
\end{table}

\begin{table}[H]
\begin{centering}
\begin{tabular}[t]{lcccc}
\toprule 
Industry  & $n$  & $\hat{\beta}$  & $\hat{\gamma}$  & $\hat{\delta}$ \tabularnewline
\midrule 
31  & 668  & 0.183  & 1.976  & -8.068 \tabularnewline
 &  & (0.006)  & (0.112)  & (0.408) \tabularnewline
 &  &  &  & \tabularnewline
32  & 386  & 0.088  & 0.916  & -6.918 \tabularnewline
 &  & (0.017)  & (0.259)  & (0.765) \tabularnewline
 &  &  &  & \tabularnewline
38  & 336  & 0.066  & 0.695  & -6.515 \tabularnewline
 &  & (0.015)  & (0.285)  & (0.464) \tabularnewline
\bottomrule
\end{tabular}
\par\end{centering}
\caption{Chilean Manufacturing plant estimation under $\theta_{m}+\theta_{k}+\theta_{l}=0.95$: Step 4 (Industries 31, 32,
and 38 in 1996). Standard errors in parentheses with 100 non-parametric
bootstrap iterations. }
\label{tab:estimates2_095} 
\end{table}

The output elasticity and average markup estimates under DRS are proportionally scaled down relative to the CRS baseline. The HSA demand parameter $\hat{\beta}$ remains statistically significantly different from zero across all specifications and industries.

Figures \ref{fig:rvsfittedr_drs} and \ref{fig:uvseps_drs} present the analogues of Figures \ref{fig:rvsfittedr} and \ref{fig:uvseps} from the main text under the concerning DRS specifications. As in the CRS case, observed revenue aligns closely with fitted revenue from the HSA demand system, and the quantile--quantile plots of demand shocks generally track the 45-degree line, confirming that the model fit is robust to the choice of returns to scale (RTS).

\begin{figure}[H]
\caption{Observed revenue vs.\ fitted revenue from Step 4 under DRS}
\label{fig:rvsfittedr_drs} \centering{}\includegraphics[scale=0.4]{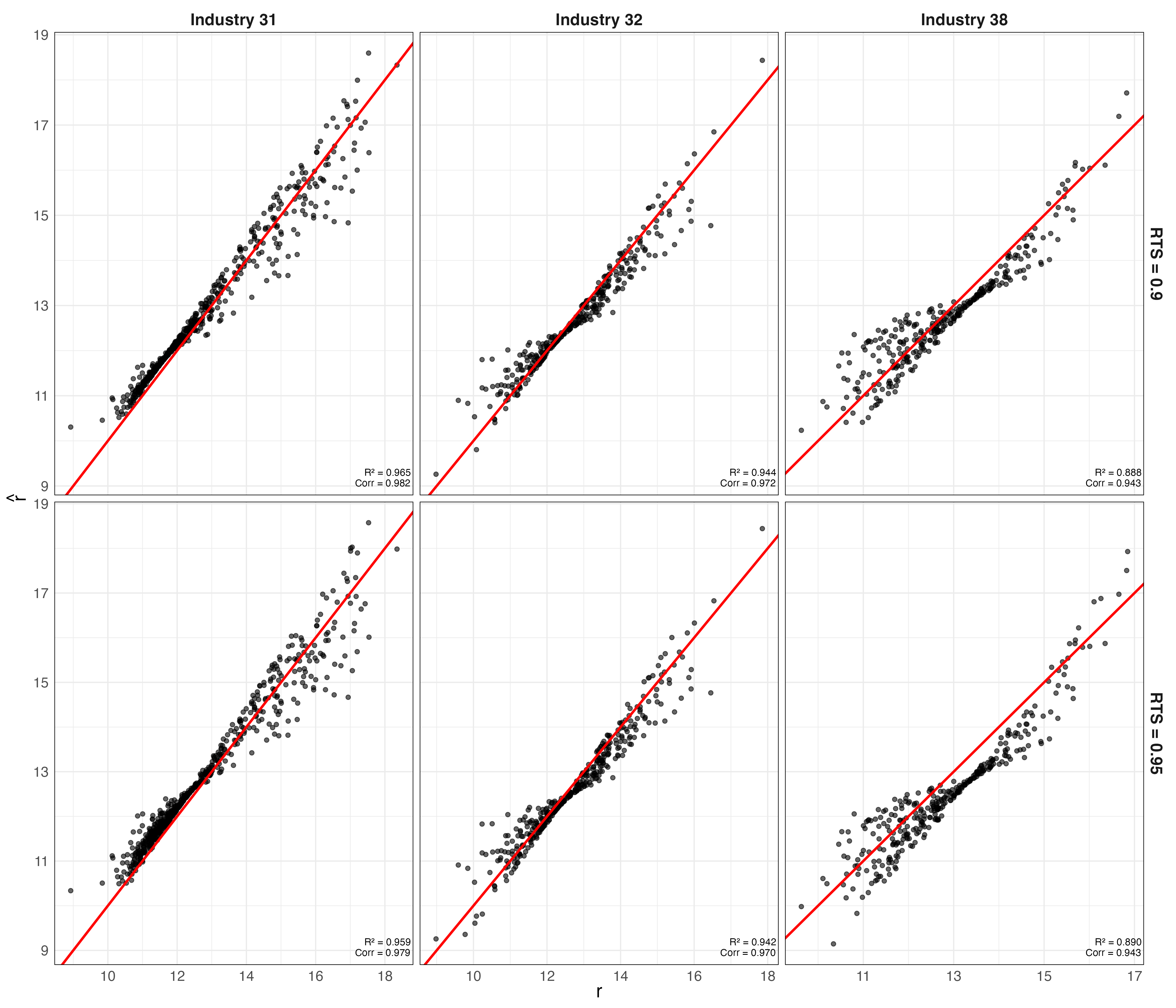} 
\begin{flushleft}
\vspace{-0.2cm}
Notes: The red line indicates the 45-degree line. Top panels: RTS $= 0.9$; bottom panels: RTS $= 0.95$.
\end{flushleft}
\end{figure}

\begin{figure}[H]
\caption{Rank of demand shock from Step 1 vs.\ Step 4 under DRS}
\label{fig:uvseps_drs} \centering{}\includegraphics[scale=0.4]{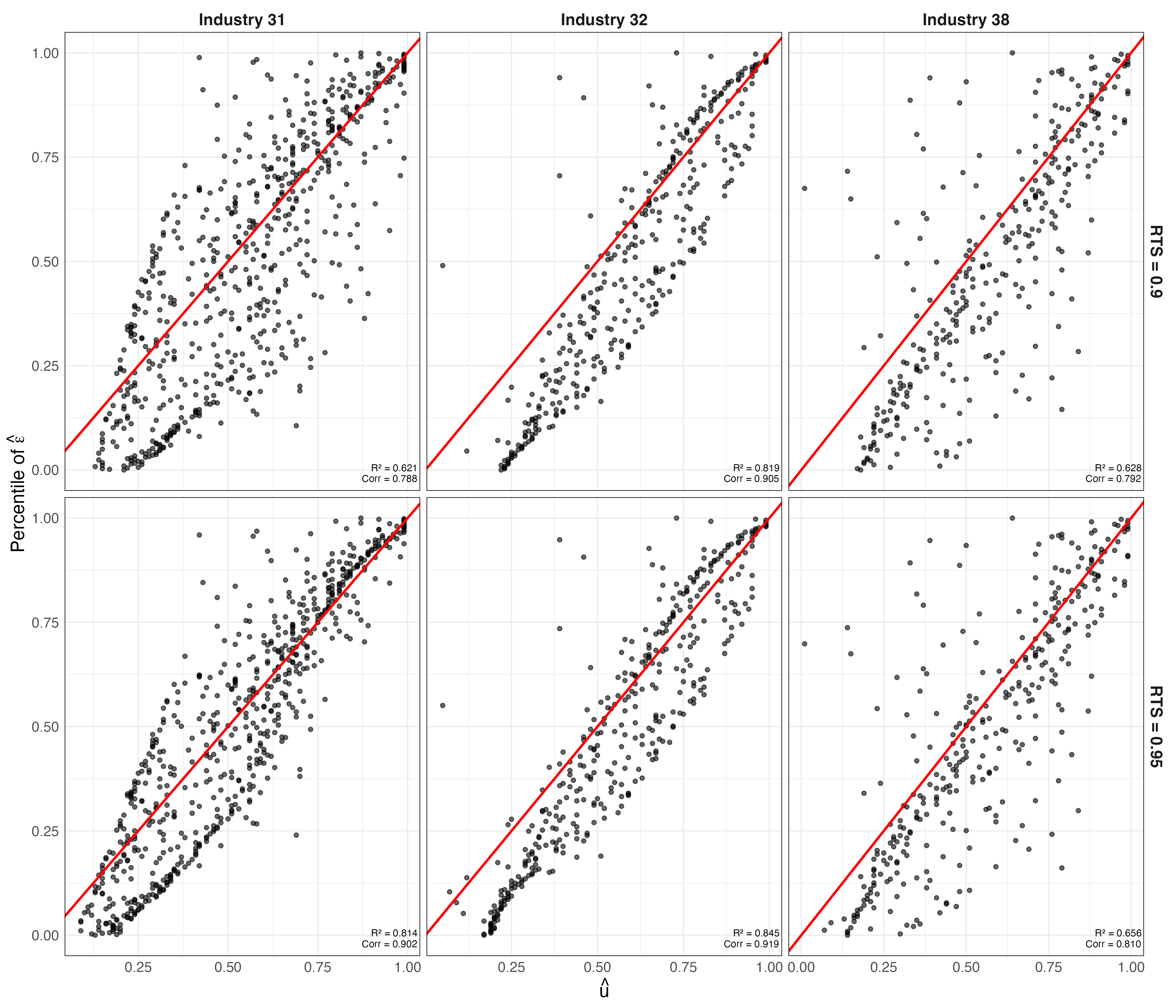} 
\begin{flushleft}
\vspace{-0.2cm}
Notes: The red line indicates the 45-degree line. Top panels: RTS $= 0.9$; bottom panels: RTS $= 0.95$.
\end{flushleft}
\end{figure}

Tables \ref{tab:welfare_09} and \ref{tab:welfare_095} report the counterfactual welfare results under DRS. The overall welfare losses from market power range from approximately 3\%--10\% of industry revenue, broadly consistent with the 3\%--6\% range found under CRS.

\begin{table}[H]
\centering %
\begin{tabular}[t]{lccc}
\toprule 
Industry  & CV  & $\Delta\Pi$  & Overall \tabularnewline
\midrule 
31  & -9.97  & -6.66  & 3.31 \tabularnewline
 & (3.09)  & (2.30)  & (0.98) \tabularnewline
 &  &  & \tabularnewline
32  & -19.81  & -12.32  & 7.49 \tabularnewline
 & (4.44)  & (3.92)  & (1.40) \tabularnewline
 &  &  & \tabularnewline
38  & -20.00  & -10.26  & 9.74 \tabularnewline
 & (5.97)  & (3.35)  & (2.96) \tabularnewline
\bottomrule
\end{tabular}\caption{Compensating Variation, profit loss, and overall welfare change in
percentage of industry revenue $\exp(\Phi_{t})$ in the transition
from original equilibrium to MCPE of Chilean Industries 31, 32, and
38 in 1996 under HSA demand system with $\theta_{m}+\theta_{k}+\theta_{l}=0.9$. Standard errors in parentheses
with 100 non-parametric bootstrap iterations.}
\label{tab:welfare_09} 
\end{table}

\begin{table}[H]
\centering %
\begin{tabular}[t]{lccc}
\toprule 
Industry  & CV  & $\Delta\Pi$  & Overall \tabularnewline
\midrule 
31  & -11.02  & -7.93  & 3.10 \tabularnewline
 & (3.31)  & (2.76)  & (0.87) \tabularnewline
 &  &  & \tabularnewline
32  & -22.58  & -14.77  & 7.81 \tabularnewline
 & (6.03)  & (5.44)  & (1.70) \tabularnewline
 &  &  & \tabularnewline
38  & -14.19  & -7.58  & 6.61 \tabularnewline
 & (5.49)  & (3.64)  & (2.38) \tabularnewline
\bottomrule
\end{tabular}\caption{Compensating Variation, profit loss, and overall welfare change in
percentage of industry revenue $\exp(\Phi_{t})$ in the transition
from original equilibrium to MCPE of Chilean Industries 31, 32, and
38 in 1996 under HSA demand system with $\theta_{m}+\theta_{k}+\theta_{l}=0.95$. Standard errors in parentheses
with 100 non-parametric bootstrap iterations.}
\label{tab:welfare_095} 
\end{table}

\end{document}